\documentclass[a4paper,11pt]{article}

\usepackage{mathpazo} 
\usepackage{listings}
\usepackage{anysize}
\marginsize{1.5cm}{1.5cm}{1.5cm}{1.5cm}
\usepackage{amsmath, amsthm, amssymb, amsfonts}
\usepackage{mathtools,dsfont}
\usepackage{graphicx}
\usepackage{lscape}
\usepackage{tikz}
\usepackage{caption,subcaption}
\usetikzlibrary{shapes,arrows}
\usetikzlibrary{decorations.pathmorphing} 
\usetikzlibrary{fit}                    
\usetikzlibrary{backgrounds}
\usepackage{float}
\usepackage{setspace}
\usepackage{color}
\restylefloat{table}
\usepackage{epstopdf}
\usepackage{mathabx}
\usepackage{xcolor}
\usepackage{appendix}
\usepackage{xcolor}

\newtheorem{theorem}{Theorem}
\newtheorem{thm}{Theorem}

\newtheorem{corollary}[theorem]{Corollary}

\newtheorem{defi}[theorem]{Definition}

\newtheorem{lemma}[theorem]{Lemma}

\newtheorem{prop}[theorem]{Proposition}

\newtheorem{rem}[theorem]{Remark}

\makeatother

\newcommand{\pconst}{\mathfrak{p}}
\newcommand{\qconst}{\mathfrak{q}}
\newcommand{\mcQ}{\mathcal{Q}}

\newcommand{\mfT}{\mathfrak{T}}
\newcommand{\mcF}{\mathcal{F}}
\newcommand{\process}[1][X]{(#1_t)_{t\in\mfT}}
\renewcommand{\d}{\mathrm{d}}
\usepackage[noblocks, auth-sc]{authblk}

\usepackage{mathtools}
\mathtoolsset{showonlyrefs}

\begin{document}

\title{
\Large
Market Making with Fads, Informed, and Uninformed Traders
}

\author[$\mathsection$]{\normalsize Emilio Barucci}
\author[$\dag$,$\ddag$]{\normalsize Adrien Mathieu}
\author[$\dag$,$\ddag$]{\normalsize Leandro S\'anchez-Betancourt}

\affil[$\mathsection$]{Department of Mathematics, Politecnico di Milano}
\affil[$\dag$]{Mathematical Institute, University of Oxford}
\affil[$\ddag$]{Oxford-Man Institute of Quantitative Finance, University of Oxford}

\maketitle
\renewcommand{\thefootnote}{} 
\footnotetext{\textit{Email addresses:} adrien.mathieu@maths.ox.ac.uk (Adrien Mathieu), 
sanchezbetan@maths.ox.ac.uk (Leandro Sánchez-Betancourt), 
emilio.barucci@polimi.it (Emilio Barucci).}
\footnotetext{Corresponding author details: Leandro Sánchez-Betancourt, Mathematical Institute, University of Oxford, Radcliffe Observatory, Andrew Wiles Building, Woodstock Rd, Oxford OX2 6GG.}
\footnotetext{LBS and AM would like to acknowledge  support from the Oxford-Man Institute of Quantitative Finance.}
\footnotetext{The data that support the findings of this study are available from the corresponding author upon reasonable request.}
\footnotetext{The authors declare that they have no conflicts of interest to report.}
\renewcommand{\thefootnote}{\arabic{footnote}} 
\abstract{
We characterise the solution to a continuous-time optimal liquidity provision problem in a market populated by informed and uninformed traders. 
In our model, the asset price exhibits fads ---these are short-term deviations from the fundamental value of the asset. 
Conditional on the value of the fad, we model how informed traders and uninformed traders arrive in the market. 
The market maker is aware of the existence of the two groups of traders but only observes the anonymous order arrivals. 
We study both, the complete and the partial information versions of the control problem faced by the market maker. 
In such frameworks, we characterise the value of information, and we find the price of liquidity as a function of the proportion of informed traders in the market. We  discuss how to calibrate model parameters from market data and we quantify the outperformance over strategies that ignore these unobservable short-term deviations.} Lastly, for the partial information setup, we explore how to go beyond the Kalman-Bucy filter to extract information about the fad from the market arrivals.

\vspace{1cm}

\begin{flushleft}
\end{flushleft}

\section{Introduction}

The market making problem has been studied in a large body of literature dealing mainly with two issues: asymmetric information and inventory management; see the seminal contribution of \cite{GLOSTEN} and \cite{AVELL,HOSTO}, respectively.\footnote{See also the books \cite{FPR,gueant2016financial,CART-book} for a comprehensive account. }
Asymmetric information is concerned with the financial losses from trading with agents who have privileged information about the fundamental value of the asset (adverse selection). Inventory management focuses on controlling the risk associated with an open inventory, that is, the risk that the price goes up when the inventory is negative (short position) or that the price goes down when the inventory is positive (long position). 

In what follows, we approach the problem through an angle that combines asymmetric information and inventory management in a new way: we study the market making problem when the market is affected by ``fads''. By fads we mean short-term deviations from the fundamental price ---a component that is associated with market inefficiencies and that may be related to over-optimism, short-termism, or other behavioural explanations (e.g. herding); see \cite{GEN_SHLEI,SHLEI}. 
The fad affects the market in two different ways. First, it affects the price of the asset. In addition to the classical arrival of information, which is modelled as a Brownian motion, there is  a mean reverting process (the fad) that drives mid-prices ---the mean-reversion well describes the fad's long run bursting properties. 
Second, the fad also affects the arrival of market orders. This is because a fraction of agents populating the market (informed traders) distinguish fads from  fundamental values. 

The objective of the market maker is to manage inventory while facing both (i) uninformed traders, who are sensitive to the displacements of bid-ask quotes from the mid-price as in \cite{AVELL}, and (ii) informed traders, who identify the true fundamental value of the asset (disentangling it from the fad component) and arrive in the market according to the displacements of the quotes from the fundamental.

The problem crucially depends on the information of the market maker about the fad. In this paper, we work out optimal bid-ask pricing strategies under both, a perfect information setup and an incomplete information setup. This allows us to obtain insights into the value of information depending on the proportion of informed traders in the market. We find that the optimal displacements are affected by the fad (or its filter) asymmetrically: as the fad increases, the market maker decreases the price of liquidity in the ask side and increases the price of liquidity in the bid side to manage the inventory optimally. This arises because the market maker aims to balance the arrival of trades from informed traders (who consider the fundamental value rather than the mid-price as reference price) and uninformed traders (who consider the mid-price as reference price). This asymmetric effect represents a novelty of our model 
over the classical Avellaneda-Stoikov framework \cite{AVELL}: private information (or its filter) plays a key role in the quote asymmetry we characterise; this is in addition to the classical asymmetry coming from inventory management.

Confirming this interpretation, in our analysis, we show that when the trades sent by informed traders become ``sharper'' or more ``toxic'', i.e., trades carried out when the fundamental value of the asset is farther from the mid-price, the market maker's performance worsens and the bid-ask spread increases.

Our paper is closely related to the strand of the literature that studies  how brokers deal with toxic order flow. Toxic flow is the trading activity of informed traders who have access to signals about the short-term trend of the asset and exploit them to their advantage.
For example, \cite{CART-SAN} analyses a model where informed and uninformed traders carry out their trading with a broker.\footnote{See \cite{CDB,ELH} for methods to detect toxic order flow. } In their model, private information is described by a stochastic process that affects the fundamental value of the asset. Similarly, in our model, the market price is affected by a mean reverting process describing deviations (fad) from the fundamental.  Conditional on the fad, we model how informed traders and uninformed traders arrive in the market in a framework similar to \cite{AVELL}.\footnote{An interesting example of a mean-field game of strategic players interacting  in an Avellaneda-Stoikov framework is \cite{baldacci2023mean}.} To the best of our knowledge, this is the first paper where the order arrival from informed and uninformed traders is incorporated in the Avellaneda-Stoikov framework.

Since the reference paper by \cite{ALM}, several papers have incorporated predictive signals into a stochastic control framework of execution cost minimisation. Predictive signals refer mainly to measurable processes built from order book dynamics, e.g., the order flow or the order book imbalance, which affects the asset price in the short term; see, e.g., \cite{CONT,CART-JAI}. There are a number of contributions in the optimal execution literature that address the problem in a full information setting, see e.g.,  \cite{CART-JAI,LEHA,NEU,BAN_CART}, other papers that consider a noisy signal of future stock prices, see e.g.,  \cite{BANK}, and  others that deal with the incomplete information setting, see e.g., \cite{CASGRAIN}. Within the market making 
problem there are fewer examples. One of the examples is \cite{cartea2020market} who study the problem of posting at-the-touch when the market maker observes the short-term drift of the asset. Another example is \cite{drissi2022solvability} who studies market making when prices are  mean-reverting (here the signal is whether the price is above or below the long run level).
In what follows, we consider a setting in which both the 
arrival rate and the asset price contain information about the fundamental value. We consider both a full information and a partial information setting  with the market maker solving a filtering problem.

\subsection*{Summary of contributions}

This paper contributes to the literature on market making by considering a market with fads that is populated by informed and uninformed traders. We provide three key insights relevant to  market microstructure.

First, temporary mispricing generates an endogenous asymmetry in optimal liquidity provision. When the price is temporarily above fundamentals (positive fad), informed traders sell more aggressively and the market maker optimally responds by widening the bid relative to the ask, even in the absence of inventory pressure. A negative fad leads to the widening on the ask side. Thus, the model provides a structural mechanism to link short-term price deviations (short living piece of information) to asymmetric liquidity. 

Second, we extend the Avellaneda-Stoikov expressions for the spread to be an explicit function of the fraction of informed traders (or toxic flow) in the market. This provides a tractable mapping between market composition and optimal spreads, an approach that can be employed using market data.

Third, by comparing the full-information problem and the partial-information problem, we quantify the value of observing or filtering the temporary price distortions (fads). We show that the value of information increases with the persistence of the fad and increases and then decreases as the filter of the fad becomes better.

The remainder of the paper is organised as follows. In Section \ref{MOD} we introduce the model. In Section \ref{sec: Perfect Information} we address the full information problem. In Section \ref{sec: Imperfect Information} we study the incomplete information problem. Section \ref{sec : Simulations} carries out a sensitivity analysis on the optimal strategies through simulations,  Section \ref{CALIB} provides a framework to calibrate model parameters from market data, and Section \ref{sec : beyond approx} presents directions for future research going beyond
the Kalman-Bucy filtering framework. In particular, we discuss how the market maker would filter order arrivals in search of the unobserved value of the fad. In Section \ref{CONC} we provide our conclusions. 
We collect proofs and further robustness checks in the Appendix.

\section{The model}
\label{MOD}
We consider a market maker who trades an asset over a finite time horizon $\mfT:=[0,T]$. Over this period, she provides liquidity by setting the bid and the ask price at which she is willing to buy and sell one unit of the asset, respectively.\footnote{We use one unit without loss of generality. For example, one unit can be thought as the average trade size in the market in which the market maker trades.} 

The model builds on \cite{AVELL,GUE,GUEANT2013}.
Similar to \cite{AVELL}, we take the mid-price of the asset to be exogenous and assume that it evolves as
\begin{equation}
\label{eq: mid-price}
\d S_t=\mu\, \d t+ \sigma\,\d W_t\,,\qquad S_0\in \mathbb{R}^+\,,
\end{equation}
where $\mu\in\mathbb{R}$ and $\sigma>0$.\footnote{In \cite{AVELL} the authors take $\mu=0$.}
Unlike \cite{AVELL}, and similar to \cite{BANK,CAMP,GUA}, we split the noise  $\process[W]$ into a persistent and a temporary component, respectively $\process[U]$ and $\process[Z]$. More precisely,
\begin{equation}\label{eq: noise decomposition}
W_t=\pconst\,Z_t+\qconst\,U_t\,,    
\end{equation}
with $\pconst,\qconst >0$ satisfying $\pconst^2+\qconst^2=1$, and $\process[U]$ follows a mean reverting Ornstein-Uhlenbeck (OU) process
\begin{equation}\label{eq:OU process}
\d U_t=-\eta\,U_t\, \d t+  \d B_t\,,\qquad U_0=0\,,
\end{equation}
where $\process[B]$ and $\process[Z]$ are independent $\process[\mathcal{F}]$-Brownian motions;{\footnote{All our results can still be obtained if we assume correlation between $\process[B]$ and $\process[Z]$.  } below, we provide a rigorous characterisation of the probability space.
The two components $\process[Z]$ and $\process[U]$ aim to capture different features: $Z_t$ is the classical martingale contribution that describes the arrival of information in the market affecting the fundamental value, $U_t$ represents the fad or bubble in the market, a component affecting the price that after the hype tends to burst and therefore to dampen;\footnote{In our model $\process[S]$ satisfies $\d S_t = (\mu - \eta \, \sigma \, \qconst \, U_t) \, \d t + \sigma \, \d \Bar{W}_t$, where $\Bar{W}_t = \pconst \, Z_t + \qconst \, B_t$ is a Brownian Motion and $\process[U]$ is the OU process in (\ref{eq:OU process}) with $\langle \Bar{W},B \rangle _t =  \qconst \, t$.} this can be due to cognitive limits, irrational exuberance or herding effects. This component reverts to zero in the long run, which is why it has been proposed in several classical papers that deal with temporary deviations from the fundamental price, see \cite{CAMP,FAFRE,SUMME}.\footnote{The Ornstein-Uhlenbeck process is widely used to model the dynamic of order flow imbalance/liquidity, see \cite{CART-JAI,LEHA}.} Within our model, we can interpret $S_t-\sigma \qconst U_t$ as the fundamental price, \textit{i.e.}, the asset price without the fad component.
The values of the model parameters $\pconst$ and $\qconst$ control the relevance of the martingale/fundamental component and of the fad component, respectively.

As in \cite{GUEANT2013}, 
the inventory of the market maker evolves according to the following law of motion
\begin{equation}\label{eq: inventory}
\d Q_t = \d N^b_t - \d N^a_t\,,\qquad Q_0\in\mcQ\,,
\end{equation}
where $Q_0$ denotes the initial amount of the asset held by the market maker at $t=0$. The processes  $\process[N^b]$ and $\process[N^a]$ count the number of 
shares the market maker bought and sold, respectively; each transaction is of constant size equal to unity. 
We assume that the inventory is bounded from above and below, \textit{i.e.},  $\underline{q}\leq Q_t \leq \overline{q}$ for $\underline{q} < 0 < \overline{q}$ with $\underline{q}, \overline{q}\in \mathbb{Z}$, and therefore the market maker does not quote on a given side if a trade on that side would breach these boundaries. The set $\mcQ = \left\{\underline{q}, \underline{q}+1, \dots, \overline{q}-1, \overline{q}\right\}$ is the set of possible values for the inventory.
The market maker controls the intensity of the counting processes $\process[N^{a,b}]$ by defining the bid and the ask price, but arrival of 
orders is not fully in his hands as the fad component affects it. To model both components, we write the intensities of the two processes  as follows
\begin{align}
\lambda^a_t&= \Big(\underbrace{\varphi\,e^{-k\,\delta^a_t}}_{\text{uninformed}} + \underbrace{\psi\,e^{-k \, \delta^a_t - \gamma \, (\sigma \,\qconst \,U_t \, \vee \, \mathcal{S}^-)}}_{\text{informed}} \Big) \, \mathds{1}_{\{Q_{t^-} > \underline{q} \}}\,, \label{eq: intensity lambda a}\\
\lambda^b_t&=\Big(\underbrace{\varphi\,e^{-k\,\delta^b_t}}_{\text{uninformed}} + \underbrace{\psi\,e^{-k\delta^b_t +  \gamma \, (\sigma \,\qconst \,U_t \,  \wedge  \, \mathcal{S}^+)}}_{\text{informed}}\Big) \, \mathds{1}_{Q_{t^-} < \overline{q} \}}\,, \label{eq: intensity lambda b}
\end{align}
for the ask quote $S^a_t=S_t + \delta^a_t$  and the bid quote $S^b_t=S_t-\delta^b_t$, respectively, and where $k, \gamma>0$ and ${\mathcal{S}}^- \leq 0 \leq \mathcal{S}^+$. In what follows, we refer to $\delta^a_t$ and $\delta^b_t$ as the ask and bid displacements, respectively. The parameters $\varphi$ and $\psi$ control the proportion of uninformed and informed traders in the market.
They modulate the baseline  intensity of order arrivals from informed and uninformed traders, respectively. Thus, their ratio is linked to the  fraction of informed-uninformed traders populating the market. We return to this point below.

\begin{rem}
The well-known model in \cite{AVELL} is recovered by setting $\qconst=\psi=0$ (the fad does not affect the asset price nor the arrival rate of informed traders).
\end{rem}

The stochastic intensities are the result of the activity of two groups of traders: uninformed and informed traders.  (a) Uninformed traders do not know about the fad and therefore cannot disentangle the fad component from the mid-price to determine the fundamental price. Thus,  their arrivals depend solely on how generous the quotes $S^{a,b}_t$ are with respect to the mid-price $S_t$ in the market (fundamental plus fad), \textit{i.e.}, $\delta^a_t, \ \delta^b_t$. 
On the other hand, (b) informed traders observe the fad (as smart traders in \cite{CAMP}), and therefore they observe the fundamental value of the asset and their arrivals depend on how generous the quotes are with respect to that level. The formulation of the arrival rates in \eqref{eq: intensity lambda a}-\eqref{eq: intensity lambda b} can be derived as follows. Let $F_t = S_t - \sigma\,\qconst\,U_t$ be the fundamental at time $t$ (mid-price minus fad), then the `generosity' of the ask quote with respect to fundamental 
is $S^a_t - F_t = S_t + \delta^a_t - S_t + \sigma\,\qconst\,U_t=\delta^a_t + \sigma\,\qconst\,U_t$. 
Similarly, for the bid quote (at which the informed trader may sell one share), the `generosity' is $F_t - S^b_t = S_t - \sigma\,\qconst\,U_t - S_t + \delta^b_t=\delta^b_t -  \sigma\,\qconst\,U_t$.
The minus factor in the exponent of the arrival rate of informed traders in \eqref{eq: intensity lambda a} is because larger values of $\delta^a_t + \sigma \, \qconst \, U_t$ corresponds to less generosity of the quote and therefore a lower arrival rate of informed traders, the same argument holds true for \eqref{eq: intensity lambda b}.

This argument holds for the baseline case $\gamma = k$. 
The parameters $\gamma$ and $k$  modulate how the order flow of informed and uninformed traders reacts to the fad and the displacements of the quotes from mid-price. The larger (resp.~smaller) the values of these parameters, the more (resp.~less) market order arrivals react to the fad and the displacements. 
The model has enough flexibility to allow informed traders to react more to the fad ($\gamma>k$) or less ($\gamma<k$) than in the baseline case.

The above arguments provide an explanation of the arrival rate of informed traders in
\eqref{eq: intensity lambda a}-\eqref{eq: intensity lambda b}. The peculiarity is that the fads affect the intensities of order arrivals in opposite directions. More precisely, a  positive (resp.~negative) $U_t$ signals an inflated (resp.~deflated) price to the informed traders when compared to the fundamental, this  induces informed traders to sell (resp.~buy) with a higher intensity and buy (resp.~sell) with a lower intensity than they would have done otherwise.

For technical reasons we impose that the effect of the fad in the intensity of order arrivals is bounded. 
More precisely, we ask the baseline intensity $\psi\,\exp{(\pm\gamma\,\sigma \,\qconst\,U_t)}$ of informed traders to be bounded by $\psi\,\exp{(-\gamma\, \mathcal{S}^-)}$ and $\psi\,\exp{(\gamma\,\mathcal{S}^+)}$ in the ask and bid price, respectively. This assumption is needed to show that the Radon-Nikodym derivative that we introduce is a true martingale and that the problem is well defined. 
From a financial point of view, the motivation of the assumption comes from the fact that the intensity of market orders cannot be arbitrarily large because exchanges have a maximum orders they can process per unit time.\footnote{One way to see this is to rewrite the stochastic intensities in  (\ref{eq: intensity lambda a})-(\ref{eq: intensity lambda b}) as $\psi \, (e^{\pm \gamma \, \sigma \, \qconst \, U_t} \wedge \mathcal{Z})$ for $\mathcal{Z}$ a function of $\mathcal{S}^+$ and $\mathcal{S}^-$ and where $\mathcal{Z}$ is the bound (per unit time) from the exchange. For example, $1000$ orders per millisecond.}

The  cash of the market maker satisfies
\begin{equation}\label{eq: cash}
\d X_t =(S_t +\delta^a_t) \d N^a_t -
(S_t - \delta^b_t) \d N^b_t\,,\qquad X_0 = 0\,.
\end{equation}
The market maker controls the bid and ask quotes and optimises the following criterion
\begin{equation}\label{eq: payoff}
    X_T +Q_T\,S_T   - \alpha\,Q_T^2 - \phi\,\int_0^T Q_u^2\,\d u \,.
\end{equation}
The interpretation is as follows. Over the horizon $[0,T]$, the market maker is interested in the final realised cash $X_T$, the market value of the inventory at time $T$, given by $Q_T\,S_T$, and inventory control (over time and at the end of the period), which is given by the remaining two terms.
This performance criterion is widely used in the algorithmic trading literature, e.g. see \cite{CART-book}. The mark-to-market valuation of the inventory is at the mid-price in the market $S_T$ (regardless of whether the fad is positive or negative). This assumption implies that the market maker is not interested in discerning {\em per se} the fundamental value. In Section \ref{sec : beyond approx} we discuss the role of this assumption and study the case where the market maker evaluates the inventory at the fundamental price.\footnote{Note that $Q_T\,S_T - \alpha\,Q_T^2 = Q_T (S_T - \alpha\,Q_T))$. This renders the interpretation that at $T$ the market maker liquidates the open inventory and $\alpha\,Q_T$ captures the slippage of such a liquidation (the asset price is affected linearly by the amount to be liquidated). Here, liquidation is from $S_T$ because this is the mid-price in the market. This is similar to \cite{boyce2024market} where their mark-to-market is done at the prices of the competing market makers.}

In what follows, we study the problem depending on the information available to the market maker. 
First, we address the full information scenario, \textit{i.e.}, the case where the market maker observes the realization of the  process $U_t$ and defines the control problem accordingly.\footnote{The full observation setting is relevant in its own right given that in some OTC markets the liquidity provider is able to profile their clients and is able to detect the short-term signal from the trading activity of their clients \cite{CDB,CART-SAN}.}
Second, we address the more realistic case in which the market maker filters the unobserved fad of the observed price process. We conclude by presenting alternative filters (filtering arrivals instead of prices) and other directions for future research.

\section{Full information}
\label{sec: Perfect Information}
\subsection{Formal derivation of the probability space}
Before pursuing the solution to the full information problem, it is important to provide a rigorous characterisation of the probability space we employ.
Consider $\Omega_d$ the set of increasing piecewise constant càdlàg functions from $[0,T]$ into $\mathbb{N}$ with jumps equal to one and $\Omega_c$ the set of continuous functions from $[0,T]$ into $\mathbb{R}$. We define $\Omega = \Omega_c^2 \times \Omega_d^2$ as the sample space.
 We let $(B_t,Z_t,\Bar{N}_t^a,\Bar{N}_t^b)_{t \in \mfT}$ be the canonical process on $\Omega$. The associated filtration is $\widebar{\mathcal{F}}= (\mathcal{F}_t^c \otimes \mathcal{F}_t^c \otimes \widebar{\mathcal{F}}_t^d \otimes \widebar{\mathcal{F}}_t^d)_{t \in \mfT}$ where $\process[\widebar{\mathcal{F}}^d]$ (resp.~$\process[\mathcal{F}^c]$) is the right continuous completed filtration associated with $\Bar{N}^a$ and $\Bar{N}^b$ (resp.~$B$ or $Z$).
We denote by $\mathbb{P}_0$ the probability measure on $(\Omega, \mathcal{F})$ such that $(\Bar{M}^a_s = \Bar{N}^a_s - s \,, \Bar{M}^b_s = \Bar{N}^b_s -s )_{s \in \mfT}$ are martingales and $\process[Z]$, $\process[B]$ are Brownian motions. We also assume independence between the processes. We define the point processes $\process[N^a], \process[N^b]$ as solutions of the coupled SDE
\begin{equation}
    \d N^b_t = \mathds{1}_{\{N_{t^-}^b-N_{t^-}^a < \overline{q}\}} \, \d \Bar{N}_t^b\,,  \qquad \d N^a_t = \mathds{1}_{\{N_{t^-}^b-N_{t^-}^a > \underline{q}\}} \, \d \Bar{N}_t^a \,.
\end{equation} 
Then, under $\mathbb{P}_0$, $\process[N^a]$ and $\process[N^b]$ have intensities $$\Lambda^a_t = \mathds{1}_{\{N_{t^-}^b-N_{t^-}^a > \underline{q}\}} \, \qquad \,\Lambda^b_t= \mathds{1}_{\{N_{t^-}^b-N_{t^-}^a < \overline{q}\}}  \,,$$
and they are independent of the Brownian motions.

We let $\mathcal{U}$ be the set of admissible controls given by 
    $$\mathcal{U} := \left\{\delta = (\delta^a_t, \delta^b_t)_{t \in \mfT} \in L^2(\Omega \times [0,T]): \delta \geq - \delta_\infty \text{ and $ \mathcal{F} $ predictable}\right\} \, ,$$
    where $\delta_\infty = (\delta_\infty^a, \delta_\infty^b)$ characterises the lower bound of the control. Next, for  $\delta \in \mathcal{U}$, we define $\mathbb{X}_t = (X_t,Q_t,S_t,U_t)$ to be the state of the system, where the dynamics of $\process[\mathbb{X}]$ are given by 
    \begin{align}
    \label{state_var_FI}
        \begin{cases}
            &\d X_t = (S_t + \delta^a_t) \, \d N^a_t -  (S_t - \delta^b_t) \, \d N^b_t \, , \qquad X_0 = 0 \, ,\\
            &\d Q_t = \d N^b_t - \d N^a_t \,, \qquad Q_0 \in \mathbb{R} \, , \\
            & \d S_t = \mu \, \d t + \sigma \, (\qconst\,  \d U_t + \pconst \, \d Z_t) \, , \qquad S_0 \in \mathbb{R}^+ \, ,\\
            &\d U_t = - \eta \, U_t \, \d t + \d B_t \, , \qquad U_0 = 0 \, .
        \end{cases}
    \end{align}
Note that we do not include $Z_t$ as state variable because knowledge of $U_t$ and $S_t$ implies knowledge of $Z_t$ (given that $S_t = S_0 + \mu\,t + \sigma\left(\pconst\,Z_t+\qconst\,U_t\right)$).

Using these processes, we make an explicit change of measure associated with each control process. More precisely, we consider the functions 
\begin{align}
    & \lambda^a(U_t,Q_{t^-}, \delta^a_t) := \lambda^a_t =\underbrace{(\varphi + \psi \, e^{-\gamma \, (\sigma \,\qconst \,U_t \, \vee \,  \mathcal{S}^-)}) \, e^{-k \, \delta^a_t} }_{\Gamma_t^a}\, \mathds{1}_{\{ Q_{t^-} > \underline{q} \}} \, , \label{sto_intens_a}\\
     &  \lambda^b(U_t,Q_{t^-}, \delta^b_t) := \lambda^b_t = \underbrace{(\varphi + \psi \, e^{ \gamma \, (\sigma \,\qconst \,U_t \, \wedge \, \mathcal{S}^+)}) \, e^{-k \, \delta^b_t}}_{\Gamma_t^b} \, \mathds{1}_{\{Q_{t^-} < \overline{q} \}} \, , \label{sto_intens_b}
\end{align}
which represent the ask and bid intensity for the control $\delta \in \mathcal{U}$. Then, for any $\delta \in \mathcal{U}$, we define $\mathbb{P}^\delta$ by 
\begin{equation}
    \frac{\d \mathbb{P}^\delta}{\d \mathbb{P}_0} = L_T^{\delta} \, ,
\end{equation}
where $L_t^{\delta}$ is the Doléans-Dade exponential of $\mathcal{Z}_t$ and
$$\mathcal{Z}_t = \int_0^t (\Gamma^a_s -1 )\,  \d \Bar{M}^{a}_s + \int_0^t (\Gamma^b_s -1 ) \, \d \Bar{M}^{b}_s \, .$$
Since $(\delta^a_t,\delta^b_t)_{t \in \mfT}$ are bounded from below and the intensities are uniformly bounded,  $\process[L^\delta]$ is a $\mathbb{P}_0$ martingale by the Novikov criterion;\footnote{
See Section 3 in \cite{euch2021optimal}. 
} and under $\mathbb{P}^\delta$, the processes 
$$\Bar{M}_t^{a,\delta} := \Bar{N}^a_t - \int_0^t \Gamma^a_u \, \d u \, \qquad \, \text{ and } \, \qquad \Bar{M}_t^{b,\delta} := \Bar{N}^b_t - \int_0^t \Gamma^b_u \, \d u \,  ,$$
are true martingales. Thus, the processes $\process[\Bar{N}^a]$ and $\process[\Bar{N}^b]$ have stochastic intensities $\process[\Gamma^a]$ and $\process[\Gamma^b]$,} and then $\process[N^a]$ and $\process[N^b]$ have stochastic intensities $\process[\lambda^a]$ and $\process[\lambda^b]$.
 Finally, on the probability space $(\Omega, \mathcal{F}, \mathbb{P}^\delta)$, the processes $\process[Z]$ and $\process[B]$ are $\mathcal{F}$-Brownian motions, and the point processes $\process[N^a]$ and $\process[N^b]$ are $\mathcal{F}$-Poisson processes.

\subsection{Characterisation of the solution}
The market maker aims to solve the following control problem 
\begin{equation}
\label{eq:VF_1}
 \sup_{\delta \,  \in \, \mathcal{U}}\mathbb{E}^{\delta}\left[ X_T +Q_T\,S_T - \alpha\,Q_T^2 - \phi\,\int_0^T Q_u^2\,\d u \right]\,.
\end{equation}

Using the Itô formula for $X_T + Q_T \, S_T$, we get 
$$\mathbb{E}^{\delta}[X_T + Q_T \, S_T] = \mathbb{E}^{\delta}[X_0 + Q_0 \, S_0] + \mathbb{E}^{\delta}\Bigg[\int_0^T (\delta^a_s \, \lambda^a_s + \delta^b_s \, \lambda_s^b +Q_s \, \mu - \eta \, \qconst \, \sigma \,Q_s \, U_s) \, \d s \Bigg] \, .$$
The maximisation problem \eqref{eq:VF_1} is then equivalent to maximise 
\begin{equation}
    \mathbb{E}^{\delta}\left[ -\alpha \, Q_T^2 + \int_0^T (\delta^a_s \, \lambda_s^a + \delta^b_t \, \lambda_s^b - \phi \, Q_s^2 + Q_s \, \mu - \eta \, \qconst \, \sigma \,Q_s \, U_s) \, \d s \right] \, ,
\end{equation}
from which we define the value function associated with the control problem
\begin{align}
     V(t,q,u) = \sup_{\delta \, \in \, \mathcal{U}}  \mathbb{E}_{t,q,u}^{\delta}\Bigg[ -\alpha \, Q_T^2 + \int_t^T &\Big(
      - \phi \, Q_s^2 + Q_s \, \mu - \eta \, \qconst \, \sigma \,Q_s \, U_s \nonumber \\
     & +\delta^a_s \, \lambda^a(U_s,Q_{s^-},\delta^a_s) + \delta^b_s \, \lambda^b(U_s,Q_{s^-},\delta^b_s)\Big) \, \d s \Bigg] \, .\label{eq:new_value_fct}
\end{align}
The expression above is finite since all processes belong to $L^2(\Omega \times [0,T])$. Below, we also prove that $V$ is locally bounded. Standard techniques from the theory of stochastic control show that the Hamilton Jacobi Bellman (HJB) equation associated with such a problem is 
\begin{align}
\label{HJB_FI}
0&=\partial_t V - \phi\,q^2 - \eta \, u \, \partial_u V + \frac{1}{2}\,\partial_{u,u}^2 V + (\mu - \eta \, \sigma \, \qconst \, u)\, q\\
&\qquad + \sup_{\delta^a}\Big\{e^{-k\,\delta^a}(\varphi+ \psi \, e^{-\, \gamma \, (\sigma \,\qconst \,u \, \vee \, \mathcal{S}^-)}\times\nonumber\\
&\hspace{20mm}\Big(\delta^a + V(t, q-1,u) - V(t,q,u) \Big)   \Big\}\mathds{1}_{q>\underline{q}}\nonumber\\
&\qquad + \sup_{\delta^b}\Big\{e^{-k\,\delta^b}(\varphi+ \psi \, e^{\gamma\, (\sigma \,\qconst \,u\, \wedge \mathcal{S}^+)})\times\nonumber\\
&\hspace{20mm}\Big(\delta^b + V(t, q+1,u) - V(t,q,u) \Big)   \Big\}\mathds{1}_{q<\overline{q}}\nonumber\,, 
\end{align}
with terminal condition $V(T,q,u) = - \alpha \, q^2$.
We compute the maximisers of the above expression explicitly obtaining that
\begin{align}
    \delta^{a*} &= \Big(\frac{1}{k}- V(t, q-1,u) + V(t,q,u) \Big) \vee - \delta_\infty^a\,,\qquad q\neq \underline{q}\,, \\
    \delta^{b*} &= \Big(\frac{1}{k} - V(t, q+1,u) + V(t,q,u) \Big) \vee -\delta_\infty^b \,,\qquad q\neq \overline{q}\,, 
\end{align}
which transforms the HJB into the following partial differential equation (PDE)
\begin{equation}
\label{eq:HJB_EQ_FI}
\begin{aligned}
0&=\partial_t V - \eta \, u \, \partial_u V - \phi\,q^2  + \frac{1}{2}\,\partial_{u,u}^2 V   + (\mu - \eta \, \sigma \, \qconst \, u) \, q\\
&\qquad + \frac{1}{k}\,  \exp( -1 + k [V(t,q-1,u)-V(t,q,u)]) \, (\varphi + \psi \, e^{-   \gamma \, (\sigma \,\qconst \,u \, \vee \, \mathcal{S}^-)}) \mathds{1}_{ \{q>\underline{q} \}} \, \mathds{1}_{ \{\delta^{a,*} = \delta^{a,**} \}} \\
&\qquad + \frac{1}{k}\,  \exp{\left( -1+ k [V(t,q+1,u)-V(t,q,u)]\right)} \, (\varphi + \psi \, e^{   \gamma \, (\sigma \,\qconst \,u \, \wedge \, \mathcal{S}^+)}) \mathds{1}_{ \{q<\overline{q} \}} \, \mathds{1}_{ \{\delta^{a,*} = \delta^{b,**}\}}  \\ 
 &\qquad+ e^{k\,\delta^a_\infty}(\varphi+ \psi \, e^{- \gamma \, (\sigma \,\qconst \,u \, \vee \, \mathcal{S}^-)}) 
\, (-\delta^a_\infty + V(t, q-1,u) - V(t,q,u) )\mathds{1}_{\{q>\underline{q} \}} \, \mathds{1}_{\{\delta^{a,*} >\delta^{a,**}  \}}  \\
&\qquad + e^{k\,\delta^b_\infty}(\varphi+ \psi \, e^{ \gamma \, (\sigma \,\qconst \,u\, \wedge \, \mathcal{S}^+)}) \,(-\delta^b_\infty + V(t, q+1,u) - V(t,q,u))\mathds{1}_{\{q<\overline{q}\}} \, \mathds{1}_{\{\delta^{a,*} > \delta^{b,**}  \}}  \, ,
\end{aligned}
\end{equation}
where 
$$\delta^{a,**} = \frac{1}{k} - V(t,q-1,u) + V(t,q,u) \, , \, \qquad \, \qquad \, \delta^{b,**} = \frac{1}{k} - V(t,q+1,u) + V(t,q,u) \, .$$

\subsubsection{Viscosity solution: complete information problem}
In this section, we prove that the value function $V$ in \eqref{eq:new_value_fct} is a viscosity solution of the HJB equation in \eqref{HJB_FI}.  We recall that $\mathcal{U}$ is the set of predictable processes $\process[\delta] = (\delta^a_t, \delta^b_t)_{t \in \mfT}$, bounded from below and taking values in $ U = U^a \times U^b = (-\delta^a_\infty, + \infty) \times (-\delta^b_\infty , + \infty)$ and where the state processes are defined in \eqref{state_var_FI}. We rewrite the price process as  
\begin{equation}
\label{eq:mid-priceFI}
    \d S_t = (\mu - \eta \, \sigma \, \qconst \, U_t) \, \d t + \sigma \, \d \Bar{W}_t \, ,
\end{equation}
where  $\process[\Bar{W}]$ is an $\mathcal{F}$-Brownian motion, with correlation $\langle B, \Bar{W} \rangle _t = \qconst \, t$.
The stochastic intensities of the Poisson processes are given by \eqref{sto_intens_a} and \eqref{sto_intens_b}.

We define the restriction of our control space $\mathcal{U}_t$ as 
$$\mathcal{U}_t := \big\{\process[\delta] \in \mathcal{U} : (\delta_s)_{s \geq t} \text{ is independent of } \mathcal{F}_t\big\} \, ,$$
and we recall an important result. Let $\Tilde{V}$ given by
$$\Tilde{V}(t,q,u) = \sup_{\delta \, \in \, \mathcal{U}_t}  \mathbb{E}_{t,q,u}^{\delta}\left[ -\alpha \, Q_T^2 + \int_t^T (\delta^a_s \, \lambda_s^a + \delta^b_t \, \lambda_s^b - \phi \, Q_s^2  +Q_s \, \mu - \eta \, \qconst \, \sigma \,Q_s \, U_s) \, \d s \right] \,,$$
where the supremum is taken over $\mathcal{U}_t$ instead of $\mathcal{U}$.
It follows from the Remark 2.2 (iv) in \cite{touzi2012optimal} that  $V = \Tilde{V}$. To simplify notations, we denote  by $\mathcal{E} := [0,T] \times \mathcal{Q} \times \mathbb{R}$ the space in which $(t,q,u)$ takes values.
Next, we prove that $V$ is locally bounded.
\begin{lemma}
    The value function $V$ is locally bounded.
\end{lemma}
\begin{proof}
    We define $\widehat{V}: (t,q,u,\delta) \in \mathcal{E} \times \mathcal{U} \rightarrow \mathbb{R}$ by 
    $$ \widehat{V}(t,q,u,\delta) = \mathbb{E}_{t,q,u}^{\delta}\left[  -\alpha \, Q_T^2 + \int_t^T (\delta^a_s \, \lambda_s^a + \delta^b_t \, \lambda_s^b - \phi \, Q_s^2 +  Q_s \, \mu - \eta \, \qconst \, \sigma \,Q_s \, U_s) \, \d s \right] \,.$$
    From $|Q_t |\leq |\mathcal{Q}|$ almost surely and using Jensen inequality, we bound $\widehat{V}$ as follows
    \begin{align*}
        | \widehat{V}(t,q,u,\delta) | &\lesssim  \mathbb{E}^{\delta}[Q_T^2] + \int_t^T \mathbb{E}^{\delta}[\delta^a_s \, \lambda^a_s + \delta^b_s \, \lambda_s^b] \, \d s + \int_t^T  \mathbb{E}^{\delta}[Q_s^2] \, \d s +  \int_t^T  \mathbb{E}^{\delta}[U_s^2] \, \d s \,  \\& \lesssim|\mathcal{Q}|^2 + \int_t^T  \mathbb{E}^{\delta}[U_s^2] \, \d s + \int_t^T \mathbb{E}^{\delta}[\delta^a_s \, \lambda^a_s + \delta^b_s \, \lambda_s^b] \, \d s \,  \\&
        \lesssim 1 + |\mathcal{Q}|^2 + \int_t^T \bigg(u^2\,e^{-2\,\eta\,(s-t)} + \mathbb{E}^{\delta}\bigg[\bigg(\int_t^s e^{-\eta\,(s-r)}\,\d B_r\bigg)^2\bigg]\bigg) \, \d s
        \\&\lesssim  1+ |\mathcal{Q}|^2 + u^2 \, ,
    \end{align*}
where we used It\^o's isometry and that $\delta$ is bounded from below (so $\delta \, e^{-k \, \delta}$ and $e^{-k \, \delta}$ are bounded), and that $\varphi + \psi \, e^{-\gamma \, (\sigma \,\qconst \,U_t \, \vee \,  \mathcal{S}^-)}$ and $\varphi + \psi \, e^{\gamma \, (\sigma \,\qconst \,U_t \, \wedge \,  \mathcal{S}^+)}$ are uniformly bounded. The symbol $ a\lesssim b$ means that there exists a constant $C > 0$ (independent of $a$ and $b$) such that $a \leq C \, b$. Finally, we obtain $$ |\widehat{V}(t,q,u,\delta) | \lesssim 1+ |\mathcal{Q}|^2 +u^2 \, ,$$
and then $(t,q,u) \rightarrow \widehat{V}(t,q,u,\delta)$ is locally bounded. Taking the supremum of the above expression implies that the value function is locally bounded.
\end{proof}

\begin{thm}[Weak dynamic programming principle]
    Define the lower and upper semicontinuous envelopes $$V_*(t,q,u) = \liminf_{(t',q',u') \rightarrow (t,q,u)} V(t',q',u') \, , $$ and  $$V^*(t,q,u) = \limsup_{(t',q',u') \rightarrow (t,q,u)} V(t',q',u') \, .$$
Consider $(t,q,u) \in \mathcal{E}$ fixed and let $\{\theta^\alpha, \alpha \in \mathcal{U}_t\}$ be a family of finite stopping with values in $[t,T]$. Then
    \begin{equation}
        \label{eq:DPP1}
        V(t,q,u) \geq \sup_{\alpha \in \mathcal{U}_t} \mathbb{E}_{t,q,u}^{\alpha}\left[\int_t^{\theta^\alpha} \Big(\delta^a_s \, \lambda_s^a + \delta^b_t \, \lambda_s^b - \phi \, Q_s^2 +  (\mu - \eta \, \qconst \, \sigma \, U_s) \, Q_s \, \Big)\, \d s + V_*(\theta^\alpha, U_{\theta^\alpha}, Q_{\theta^\alpha})\right] \, ,
    \end{equation}
    and if for any $\alpha$, $\Big((U_s^{t,u},Q_s^{t,q})^T \, \mathds{1}_{s \in [t,\theta^\alpha]}\Big)_{t \in \mfT}$ is $L^\infty$ bounded, then
    \begin{equation}
        \label{eq:DPP2}
         V(t,q,u) \leq \sup_{\alpha \in \mathcal{U}_t} \mathbb{E}_{t,q,u}^{\alpha}\left[\int_t^{\theta^\alpha}\Big(\delta^a_s \, \lambda_s^a + \delta^b_t \, \lambda_s^b - \phi \, Q_s^2+  (\mu - \eta \, \qconst \, \sigma \, U_s) \, Q_s \,\Big) \, \d s + V^*(\theta^\alpha, U_{\theta^\alpha}, Q_{\theta^\alpha})\right] \, .
    \end{equation}
\end{thm} 

We do not present the proof since it follows the one in \cite{JUSS} because the stochastic intensities are uniformly bounded. Before dealing with the viscosity solution, let us define the Hamiltonian of our problem.
\begin{defi}[Hamiltonian]
    We define the Hamiltonian $\mathcal{H} :\mathcal{Q}\times \mathbb{R} \times \mathbb{R}  \rightarrow \mathbb{R}$  
    \begin{equation}
    \begin{aligned}
        \label{eq:Hamiltonian_}
        \mathcal{H}(q,u,V_1,V_2) =& \sup_{\delta^a \, \in \, U^a}(\varphi + \psi \, e^{- \gamma \, (\sigma \,\qconst \,u\, \vee \, \mathcal{S}^-)} \, e^{-k \, \delta^a}  \, (\delta^a + V_1 ) \, \mathds{1}_{\{q > \underline{q}\}}  \\& \qquad \, \qquad + \sup_{\delta^b \, \in \, U^b} (\varphi + \psi \, e^{ \gamma \, (\sigma \,\qconst \,u \, \wedge \, \mathcal{S}^+)}) e^{-k \, \delta^b} \, (\delta^b +V_2 ) \, \mathds{1}_{\{q < \overline{q}\} } \, .
    \end{aligned}
    \end{equation}
\end{defi}
Finally, we define the HJB equation associated with the control problem.
\begin{defi}[Hamilton-Jacobi-Bellman equation]
Let $V \in C^{1,2}(\mathcal{E})$, then
    \begin{equation}
    \begin{aligned}
    \label{HJB:eq:fin}
       0 =& \partial_t V(t,q,u) - \phi \, q^2 +q \, \mu - \eta \, \qconst \, \sigma \, q \, u - \eta \, u \, \partial_u V(t,q,u) + \frac{1}{2} \partial^2_{u,u} V(t,q,u) \\&+ \mathcal{H}(q,u,V(t,q-1,u)-V(t,q,u),V(t,q+1,u)-V(t,q,u))  \, .
       \end{aligned}
    \end{equation}
    To shorten the notation, we write the equation  above as $$0 = F(t,q,u, V(t,q,u), \partial_t V(t,q,u),\partial_u V(t,q,u), \partial^2_{u,u}V(t,q,u)).$$
\end{defi}

\begin{rem}
   By a ball around a point $(t,q,u) \in \mathcal{E}$ with radius $r>0$ we mean $$\{(s,j,v) \in \mathbb{R}^3 : \lVert (s,j,v) - (t,q,u) \rVert < r\} \cap \mathcal{E} \, ,$$ 
    where $\lVert \cdot \rVert$ is the Euclidean norm.
\end{rem}

Now we can define the notions of viscosity sub-solution and viscosity super-solution

\begin{defi}[Viscosity sub-solution]
    We say that $V$ is a viscosity sub-solution if, for all $\Phi \in \mathcal{C}^{1,2}(\mathcal{E})$ such that $V^*(t_0,q_0,u_0) - \Phi (t_0,q_0,u_0)$ is a local maximum, we have 
    \begin{equation}
       F(t_0,q_0,u_0, \Phi(t_0,q_0,u_0),\partial_t \Phi(t_0,q_0,u_0), \partial_u \Phi(t_0,q_0,u_0), \partial^2_{u,u} \Phi(t_0,q_0,u_0)) \geq 0 \, .
    \end{equation}
\end{defi}

\begin{defi}[Viscosity super-solution]
    Let $V$ be locally bounded on $\mathcal{E}$. We say that $V$ is a viscosity super-solution if, for all $\Phi \in \mathcal{C}^{1,2}(\mathcal{E})$ such that $V_*(t_0,q_0,u_0) - \Phi(t_0,q_0,u_0)$ is a local minimum, we have 
    \begin{equation}
       F(t_0,q_0,u_0, \Phi(t_0,q_0,u_0), \partial_t \Phi(t_0,q_0,u_0), \partial_u \Phi(t_0,q_0,u_0), \partial^2_{u,u}\Phi(t_0,q_0,u_0)) \leq 0 \, .
    \end{equation}
\end{defi}
\begin{defi}[Viscosity solution]
    We say that $V$ is a viscosity solution, if it is a  viscosity sub-solution and a viscosity super-solution. 
\end{defi}

Before we state the main theorem of the section, we prove an important lemma.

\begin{lemma}
    The Hamiltonian $\mathcal{H}$ defined in \eqref{eq:Hamiltonian_} is continuous.
\end{lemma}
\begin{proof}
    First, we derive explicitly the value of the  Hamiltonian. We have that
    \begin{align*}
    \mathcal{H}(q,u,V_1,V_2) =& \sup_{\delta^a \, \in \, U^a} (\varphi + \psi \, e^{- \gamma \, (\sigma \,\qconst \,u\, \vee \, \mathcal{S}^-)}) \,  e^{-k \, \delta^a} (\delta^a + V_1) \, \mathds{1}_{\{q > \underline{q} \}}  \\& \qquad \, \qquad \, \qquad+  \sup_{\delta^b \, \in \, U^b} (\varphi + \psi \, e^{ \gamma \, (\sigma \,\qconst \,u \, \wedge \, \mathcal{S}^+)}) \,  e^{-k \, \delta^b} (\delta^b +V_2) \, \mathds{1}_{\{q < \overline{q} \}} \\=&  (\varphi + \psi \, e^{- \gamma \, (\sigma \,\qconst \,u \, \vee \, \mathcal{S}^-)}) \,  e^{-k \, \delta^{a,*}} (\delta^{a,*} + V_1) \, \mathds{1}_{\{q > \underline{q} \}}  \\& \qquad \, \qquad \, \qquad+  (\varphi + \psi \, e^{ \gamma \, (\sigma \,\qconst \,u \, \wedge \, \mathcal{S}^+)}) \,  e^{-k \, \delta^{b,*}} (\delta^{b,*} +V_2) \, \mathds{1}_{\{q < \overline{q} \}} \, ,
    \end{align*}
    where the optimal displacements  $\delta^* = (\delta^{a,*}, \delta^{b,*})$ are given by 
    \begin{equation}
        \delta^{a,*}(V_1) = \Big(\frac{1}{k} - V_1 \Big) \vee  - \delta^\infty_a \, , \qquad \, \qquad      \delta^{b,*}(V_2) =\Big(\frac{1}{k} - V_2\Big) \vee - \delta^\infty_b \, ,
    \end{equation}
    from which it follows that  $\mathcal{H}$ is continuous. Indeed, take $(q,u,V_1,V_2) \in \mathcal{Q} \times \mathbb{R}^3$ and consider $(q^n,u^n,V_1^n,V_2^n)_{n \in \mathbb{N}}$ that converges toward $(q,u,V_1,V_2)$. Since $\mathcal{Q}$ is discrete, there exists $n_0 \in \mathbb{N}$, such that, for all $n \geq n_0$, $q^n = q$. Then, for $n \geq n_0$, we have $\mathcal{H}(q^n,u^n,V_1^n,V_2^n) = \mathcal{H}(q,u^n,V_1^n,V_2^n)$, and convergence of $\mathcal{H}(q,u^n,V_1^n,V_2^n)$ toward $\mathcal{H}(q,u,V_1,V_2)$ follows by continuity of composition of continuous functions.
\end{proof}
Now, we are able to state the main theorem of the section.
\begin{thm}
    $V$ is a viscosity solution of the HJB equation \eqref{HJB:eq:fin}.
\end{thm}
\begin{proof}
    We first prove that $V$ is a viscosity super-solution of the equation. Consider $(t,q,u)$ and $\Phi \in C^{1,2}$ such that
    $$0 = (V_* - \Phi)(t,q,u) = \min (V_*-\Phi) \, ,$$
    and for all $(s,j,v) \neq (t,q,u), (V- \Phi)(s,j,v) > (V_*- \Phi)(t,q,u)$.
Let assume by contradiction that  $$h(t,q,u):=F(t,q,u,\Phi(t,q,u), \partial_t \Phi(t,q,u), \partial_u \Phi(t,q,u),\partial^2_{u,u} \Phi(t,q,u)) > 0 \, .$$ 
By continuity of the function $\Phi$ and the Hamiltonian, there exists an open neighbourhood $B_r$ of $(t,q,u)$ such that
$$h(s,j,v)=F(s,j,v,\Phi(s,j,v), \partial_t \Phi(s,j,v),\partial_u \Phi(s,j,v),\partial^2_{u,u} \Phi(s,j,v)) > 0 \,  \qquad \text{ on } B_r \, .$$ 

    We can define the set $J(t,q,u)$ as the set of attainable values if a jump happens in $B_r$.
    We define $$2 \, \eta := \min_{\partial B_r \cup J(t,q,u)}(V_*-\Phi) > 0 \, .$$
    Now, we consider $(t_n,q_n,u_n)$ a sequence such that 
    $$(t_n,q_n,u_n) \rightarrow (t,q,u) \text{ and } V(t_n,q_n,u_n) \rightarrow V^*(t,q,u) \, .$$
    Since $(V-\Phi)(t_n,q_n,u_n) \rightarrow 0$, we can assume that the sequence $(t_n,q_n,u_n)$ also satisfies 
    \begin{equation}
        \label{eq:ineq_first}
        \lvert (V-\Phi) \rvert (t_n,q_n,u_n) \leq \eta \, \text{ for all } \, n \geq 1 \, .
    \end{equation}
    
    Then, we consider the control $\nu \in \mathcal{U}_{t_n}$ defined by 
    \begin{align}
         &\nu^a_t = \delta^{a,*}\Big(\Phi(t,Q_{t^-}^{t_n,q_n,\nu}-1,U_t^{t_n,u_n,\nu})-\Phi(t,Q_{t^-}^{t_n,q_n,\nu},U_t^{t_n,u_n,\nu})\Big) \, ,
         \\&  \nu^b_t =  \delta^{b,*}\Big(\Phi(t,Q_{t^-}^{t_n,q_n,\nu},U_t^{t_n,u_n,\nu}))-\Phi(t,Q^{t_n,q_n,\nu}_{t^-},U_t^{t_n,u_n,\nu})\Big) \, ,
    \end{align}
    we define the stopping time 
    $$\theta_n^\nu :=\inf \{ t > t_n : (t,Q_t^{t_n,q_n,\nu},U_t^{t_n,u_n,\nu} ) \notin B_r \} \, ,$$
    and we observe that $(\theta_n^\nu, Q_{\theta_n^\nu}^{t_n,q_n,\nu}, U_{\theta_n^\nu}^{t_n,u_n,\nu}) \in \partial B_r \cup J(t,q,u)$. Then, we have 
\begin{equation}
\label{eq:ineq_2}
        \Phi(\theta_n^\nu,Q_{\theta_n^\nu}^{t_n,q_n,\nu}, U_{\theta_n^\nu}^{t_n,u_n,\nu}) \leq - 2 \, \eta + V^*(\theta_n^\nu,Q_{\theta_n^\nu}^{t_n,q_n,\nu}, U_{\theta_n^\nu}^{t_n,u_n,\nu})\, .
    \end{equation}
 
    Now, using equation (\ref{eq:ineq_first}) and Itô formula, we obtain 
 \begin{align*}
        V(t_n,q_n,u_n)&  \leq  \eta + \Phi(t_n,q_n,u_n) \, \\
        &= \eta + \mathbb{E}^{\nu}\left[\Phi(\theta_n^\nu,Q_{{\theta_n^\nu}}^{t_n,q_n,\nu}, U_{\theta_n^\nu}^{t_n,u_n,\nu})\right] 
        \\&- \mathbb{E}^{\nu}\Bigg[\int_{t_n}^{\theta_n^\nu} \Bigg(\partial_t \Phi(s,Q_{s-}^{t_n,q_n,\nu},U_s^{t_n,u_n,\nu}) - \eta \, U_s^{t_n,u_n,\nu} \, \partial_u \Phi(s,Q_{s-}^{t_n,q_n,\nu},U_s^{t_n,u_n,\nu}) \\& \qquad + \frac{1}{2} \partial^2_{u,u} \Phi(s,Q_{s-}^{t_n,q_n,\nu},U_s^{t_n,u_n,\nu}) + \lambda^a_s \Big(\Phi(s,Q_{s-}^{t_n,q_n,\nu}-1,U_s^{t_n,u_n,\nu}) \\&\qquad \, \qquad \, \qquad- \Phi(s,Q_{s-}^{t_n,q_n,\nu},U_s^{t_n,u_n,\nu}) \Big)  + \lambda^b_s \Big(\Phi(s,Q_{s-}^{t_n,q_n,\nu}+1,U_s^{t_n,u_n,\nu}) \\&\qquad \, \qquad \, \qquad \, \qquad \, \qquad \, \qquad \, \qquad \, \qquad \, \qquad- \Phi(s,Q_{s-}^{t_n,q_n,\nu},U_s^{t_n,u_n,\nu}) \Big)  \Bigg) \, \d s\Bigg]  \, \\
          &\leq  \eta + \mathbb{E}^{\nu}\Bigg[ \Phi(\theta_n^\nu,Q_{\theta_n^\nu}^{t_n,q_n,\nu}, U_{\theta_n^\nu}^{t_n,u_n,\nu}) + \int_{t_n}^{\theta_n^\nu} \Big(\nu^a_s \, \lambda_s^a + \nu^b_s \, \lambda_s^b \\&\qquad \, \qquad- \phi \, (Q_{s^-}^{t_n,q_n,\nu})^2 +  Q_{s^-}^{t_n,q_n,\nu} \, \mu - \eta \, \qconst \, \sigma \,Q_{s^-}^{t_n,q_n,\nu} \, U_s^{t_n,u_n,\nu}\Big) \, \d s \Bigg] \, ,
        \\&=  \eta + \mathbb{E}^{\nu}\Bigg[ \Phi(\theta_n^\nu,Q_{\theta_n^\nu}^{t_n,q_n,\nu}, U_{\theta_n^\nu}^{t_n,u_n,\nu}) + \int_{t_n}^{\theta_n^\nu} \Big(\nu^a_s \, \lambda_s^a + \nu^b_s \, \lambda_s^b \\&\qquad \, \qquad- \phi \, (Q_s^{t_n,q_n,\nu})^2 +  Q_s^{t_n,q_n,\nu} \, \mu - \eta \, \qconst \, \sigma \,Q_s^{t_n,q_n,\nu} \, U_s^{t_n,u_n,\nu} \Big) \, \d s \Bigg] \, ,
    \end{align*}
  and using equation (\ref{eq:ineq_2}) we obtain 
    \begin{align*}
        V(t_n,q_n,u_n) &\leq -\eta + \mathbb{E}^{\nu}\Bigg[ V_*(\theta_n^\nu,Q_{\theta_n^\nu}^{t_n,q_n,\nu}, U_{\theta_n^\nu}^{t_n,u_n,\nu}) + \int_{t_n}^{\theta_n} (\nu^a_s \, \lambda_s^a + \nu^b_s \, \lambda_s^b \\&\qquad \, \qquad \, \qquad- \phi \, (Q_s^{t_n,q_n,\nu})^2 +  Q_s^{t_n,q_n,\nu} \, \mu - \eta \, \qconst \, \sigma \,Q_s^{t_n,q_n,\nu} \, U_s^{t_n,u_n,\nu} \,) \d s \Bigg] \, .
        \end{align*}
        
    Since $\eta > 0$ is independent of $\nu$, it follows that the latter inequality is in contradiction with the second inequality of the dynamic programming principle (\ref{eq:DPP2}). Then $V$ is a viscosity super-solution.  Now, let us prove that $V$ is a viscosity sub-solution. 
    
    Consider $(t,q,u)$ and $\Phi \in C^{1,2}$ such that
    $$0 = (V^* - \Phi)(t,q,u) = \max (V^*-\Phi) \, ,$$
    and such that for all $(s,j,v) \neq (t,q,u)$, $(V-\Phi)(s,j,v) < (V^*-\Phi)(t,q,u)$. Let assume by contradiction that $$h(t,q,u)=F(t,q,u,\Phi(t,q,u), \partial_u \Phi(t,q,u),\partial^2_{u,u} \Phi(t,q,u)) < 0 \, .$$ 
By continuity of the function $\Phi$ and the Hamiltonian, there exists an open neighbourhood $B_r$ of $(t,q,u)$ such that 
$$h(s,j,v)=F(s,j,v,\Phi(s,j,v), \partial_u \Phi(s,j,v),\partial^2_{u,u} \Phi(s,j,v)) < 0 \,  \qquad \text{ on } B_r \, .$$ 
    We define $$-2 \, \eta := \max_{\partial B_r \cup J(t,q,u)}(V^*-\Phi) < 0 \, .$$
    Now, we consider $(t_n,q_n,u_n)$ a sequence such that 
    $$(t_n,q_n,u_n) \rightarrow (t,q,u) \text{ and } V(t_n,q_n,u_n) \rightarrow V^*(t,q,u) \, .$$
    Since $(V-\Phi)(t_n,q_n,u_n) \rightarrow 0$, we can assume that the sequence $(t_n,q_n,u_n)$ also satisfies 
    \begin{equation}
        \label{eq:ineq}
        \lvert (V-\Phi) \rvert (t_n,q_n,u_n) \leq \eta \, \text{ for all } \, n \geq 1 \, .
    \end{equation}
    
    Then, for an arbitrary control $\nu \in \mathcal{U}_{t_n}$, we define the stopping time 
    $$\theta_n^\nu :=\inf \{ t > t_n : (t,Q_t^{t_n,q_n,\nu},U_t^{t_n,u_n,\nu} ) \notin B_r \} \, ,$$
    and we observe that $(\theta_n^\nu, Q_{\theta_n^\nu}^{t_n,q_n,\nu}, U_{\theta_n^\nu}^{t_n,u_n,\nu}) \in \partial B_r \cup J(t,q,u)$. Then, we have 
\begin{equation}
\label{eq:ineq2}
        \Phi(\theta_n^\nu,Q_{\theta_n^\nu}^{t_n,q_n,\nu}, U_{\theta_n^\nu}^{t_n,u_n,\nu}) \geq 2 \, \eta + V^*(\theta_n^\nu,Q_{\theta_n^\nu}^{t_n,q_n,\nu}, U_{\theta_n^\nu}^{t_n,u_n,\nu})\, .
    \end{equation}
 
    Now, using equation (\ref{eq:ineq}) and Itô formula, we obtain 
 \begin{align*}
        V(t_n,q_n,u_n) & \geq - \eta + \Phi(t_n,q_n,u_n) \, \\
        &= - \eta + \mathbb{E}^{\nu}\left[\Phi(\theta_n^\nu,Q_{\theta_n^\nu}^{t_n,q_n,\nu}, U_{\theta_n^\nu}^{t_n,u_n,\nu})\right]\\&-\mathbb{E}^{\nu}\Bigg[\int_{t_n}^{\theta_n^\nu} \Bigg(\partial_t \Phi(s,Q_{s-}^{t_n,q_n,\nu},U_s^{t_n,u_n,\nu})- \eta \, U_s^{t_n,u_n,\nu} \, \partial_u \Phi(s,Q_{s-}^{t_n,q_n,\nu},U_s^{t_n,u_n,\nu}) \\&\qquad+ \frac{1}{2} \partial^2_{u,u} \Phi(s,Q_{s-}^{t_n,q_n,\nu},U_s^{t_n,u_n,\nu}) + \lambda^a_s \Big(\Phi(s,Q_{s-}^{t_n,q_n,\nu}-1,U_s^{t_n,u_n,\nu}) \\&\qquad \, \qquad \, \qquad- \Phi(s,Q_{s-}^{t_n,q_n,\nu},U_s^{t_n,u_n,\nu}) \Big)  + \lambda^b_s \Big(\Phi(s,Q_{s-}^{t_n,q_n,\nu}+1,U_s^{t_n,u_n,\nu}) \\&\qquad \, \qquad \, \qquad \, \qquad \, \qquad \, \qquad \, \qquad \, \qquad \, \qquad- \Phi(s,Q_{s-}^{t_n,q_n,\nu},U_s^{t_n,u_n,\nu}) \Big) \Bigg) \, \d s\Bigg]  \, \\
        &\geq - \eta + \mathbb{E}^{\nu}\Bigg[ \Phi(\theta_n^\nu,Q_{\theta_n^\nu}^{t_n,q_n,\nu}, U_{\theta_n^\nu}^{t_n,u_n,\nu}) \\&\qquad \, \qquad \, \qquad \, \qquad+ \int_{t_n}^{\theta_n^\nu} \Big(\nu^a_s \, \lambda_s^a + \nu^b_s \, \lambda_s^b - \phi \, (Q_{s^-}^{t_n,q_n,\nu})^2  \\& \qquad \, \qquad \, \qquad \, \qquad \, \qquad +  Q_{s^-}^{t_n,q_n,\nu} \, \mu -\eta \, \qconst \, \sigma \,Q_{s^-}^{t_n,q_n,\nu} \, U_s^{t_n,u_n,\nu}\Big) \, \d s \Bigg] \,,
         \\&= - \eta + \mathbb{E}^{\nu}\Bigg[ \Phi(\theta_n^\nu,Q_{\theta_n^\nu}^{t_n,q_n,\nu}, U_{\theta_n^\nu}^{t_n,u_n,\nu}) \\&\qquad \, \qquad \, \qquad \, \qquad+ \int_{t_n}^{\theta_n^\nu} \Big(\nu^a_s \, \lambda_s^a + \nu^b_s \, \lambda_s^b - \phi \, (Q_s^{t_n,q_n,\nu})^2  \\& \qquad \, \qquad \, \qquad \, \qquad \, \qquad +  Q_s^{t_n,q_n,\nu} \, \mu -\eta \, \qconst \, \sigma \,Q_s^{t_n,q_n,\nu} \, U_s^{t_n,u_n,\nu} \Big)\, \d s \Bigg] \,,
    \end{align*}
  and using equation (\ref{eq:ineq2}) we obtain 
    \begin{align*}
        V(t_n,q_n,u_n) &\geq \eta + \mathbb{E}^{\nu}\Bigg[ V^*(\theta_n^\nu,Q_{\theta_n^\nu}^{t_n,q_n,\nu}, U_{\theta_n^\nu}^{t_n,u_n,\nu}) + \int_{t_n}^{\theta_n} (\nu^a_s \, \lambda_s^a + \nu^b_s \, \lambda_s^b \\&\qquad \, \qquad- \phi \, (Q_s^{t_n,q_n,\nu})^2 +  Q_s^{t_n,q_n,\nu} \, \mu - \eta \, \qconst \, \sigma \,Q_s^{t_n,q_n,\nu} \, U_s^{t_n,u_n,\nu} \,) \d s \Bigg] \, .
        \end{align*}
    Since $\eta > 0$ is independent of $\nu$, it follows from the arbitrariness of $\nu \in \mathcal{U}_{t_n}$ that the latter inequality is in contradiction with the second inequality of the dynamic programming principle (\ref{eq:DPP2}). Then $V$ is a viscosity sub-solution and hence a viscosity solution.
\end{proof}

\subsubsection{Approximate solution to the full informed problem}
\label{subsec:approx}
Next, we study an approximate solution to the HJB in (\ref{eq:HJB_EQ_FI}) with terminal condition  $V(T,q,u) = -\alpha\,q^2$; in Appendix \ref{sec:how-good-second-order-approx} (Figure \ref{fig:optimal quotes ODE PDE}), we show that this approximation  is good when compared to a classical finite-difference scheme.  To carry out the approximation, we take $- \delta_\infty = -\infty$, $\underline{q} = - \overline{q} = + \infty$ and $\mathcal{S}^+ = - \mathcal{S}^- = + \infty$. Moreover, we follow the ideas in \cite{bergault2018closed}, and consider that $\Delta^- V(t,q,u) := V(t,q-1,u) - V(t,q,u) $ and  $\Delta^+ V(t,q,u) := V(t,q+1,u) - V(t,q,u)$ are small. Using this, we perform a second order Taylor expansion around zero for $ p \rightarrow \frac{1}{k} \exp(k \, p)$  of the exponential term.
Formally, this means that 
\begin{align*}
    e^{k[V(t,q+1,u) - V(t,q,u)]} &\approx \underbrace{1 + k \, [V(t,q+1,u) - V(t,q,u)] + \frac{1}{2}(k\, [V(t,q+1,u) - V(t,q,u)])^2}_{ \hat{V}^a\bigl(V(q+1)- V(q) \bigl)}\,, \\
    e^{k[V(t,q-1,u) - V(t,q,u)]} &\approx \underbrace{ 1 + k\, [V(t,q-1,u) - V(t,q,u) ] + \frac{1}{2}(k \, [V(t,q-1,u) - V(t,q,u)])^2}_{\hat{V}^b\bigl(V(q-1)-V(q)\bigl)} \, ,
\end{align*}
where we used the notation $V(q)$ instead of  $V(t,q,u)$ to make the expression compact. The approximation transforms (\ref{eq:HJB_EQ_FI}) into the following PDE
    \begin{align*}
        0&=\partial_t V - \eta \, u \, \partial_u V - \phi\,q^2  + \frac{1}{2}\,\partial_{u,u}^2 V   + (\mu - \eta \, \sigma \, \qconst \, u) \, q\\
&\qquad + \frac{1}{k}\,  e^{-1} \, \bigl(1+ k\Delta^- V + \frac{k^2}{2} (\Delta^- V)^2 \bigl) \,  (\varphi + \psi \, e^{ -\qconst \, \sigma \, \gamma  \, u})\\
&\qquad + \frac{1}{k}\,  e^{-1} \, \bigl(1+ k\Delta^+ V + \frac{k^2}{2} (\Delta^+ V)^2 \bigl) \,   (\varphi + \psi \, e^{ \qconst \, \sigma \, \gamma \, u})  \,.
    \end{align*}

Furthermore, let us consider $\gamma \, \qconst \, u$ small enough so that we can also perform a first order Taylor expansion around zero for both of the exponential terms of $\pm \,\sigma \, \qconst \, \gamma \, u$. We have that
\begin{equation}
    \label{eq:HJB_2nd_FI_tot}
\begin{aligned}
     0&=\partial_t V - \eta \, u \, \partial_u V - \phi\,q^2  + \frac{1}{2}\,\partial_{u,u}^2 V   + (\mu - \eta \, \sigma \, \qconst \, u) \, q\\
&\qquad + \frac{1}{k}\,  e^{-1} \, \bigl(1+ k\Delta^- V + \frac{k^2}{2} (\Delta^- V)^2 \bigl) \,  (\varphi + \psi \, (1 -\qconst \, \sigma \, \gamma \, u \bigl)) \\
&\qquad + \frac{1}{k}\,  e^{-1} \, \bigl(1+ k\Delta^+ V + \frac{k^2}{2} (\Delta^+ V)^2 \bigl) \,   (\varphi + \psi (1 +\qconst \, \sigma \, \gamma  \, u))  \,.
\end{aligned}
\end{equation}
We employ the ansatz that $V$ is quadratic in $q$, \textit{i.e.}, there exists $A: \mathbb{R}^+ \rightarrow \mathbb{R}$ and $B,C : \mathbb{R}^+ \times \mathbb{R} \rightarrow \mathbb{R}$ such that $$V(t,q,u) = q^2 \, A(t) + q \, B(t,u) + C(t,u) \, .$$
We obtain the following characterisation of the approximate solution.
\begin{prop}
    Let $A : \mathbb{R}^+ \rightarrow \mathbb{R}$ and $B\,, C : \mathbb{R}^+ \times \mathbb{R} \rightarrow \mathbb{R}$ be a solution to the following system of PDEs 
    \begin{align}
    \begin{cases}
    \label{eq:HJB_2nd_FI}
        0 &= \partial_t A - \phi + 4 \, (\varphi + \psi )\, e^{-1} \, k \, A^2 \, ,\\
        0&= \partial_t B + \frac{1}{2} \, \partial^2_{u,u} B + (\mu - \eta  \, \sigma \, \qconst \, u) - \eta \, u \partial_u B + 4 \,( \varphi + \psi) \, e^{-1}\, k \, A \, B + 4 \, e^{-1} \, \psi \, \qconst \,\sigma \, \gamma \, u \, A  \\&\qquad+ 4 \, e^{-1} \, k \, \qconst \, \gamma \,\sigma \, u \, \psi \, A^2  \, , \\
        0 &= \partial_t C + \frac{1}{2} \, \partial^2_{u,u} C - \eta \, u \, \partial_u C + \frac{1}{k} \, e^{-1} \, \bigl(2 \, (\varphi + \psi) + 2 \, k \, A \,( \psi + \varphi) + 2\, \psi \, k \, B \,\qconst \, \sigma \, \gamma \, u \\&\qquad+ k^2 (\varphi + \psi) \, ( A^2 + B^2) + 2 \, k^2 \, \psi \,\qconst \, \sigma \, \gamma \, u \, A \, B \bigl) \, , 
        \end{cases}
    \end{align}
    with terminal condition $A(T) = -\alpha \, , \,  B(T,u) = 0 \, , \,  C(T,u) = 0$. Define 
    \begin{equation}
        \Tilde{V}(t,q,u) = q^2 \, A(t) + q \, B(t,u) + C(t,u) \, . 
    \end{equation}
    It follows that $\Tilde{V}$ solves the approximate HJB (\ref{eq:HJB_2nd_FI_tot}).  
    \end{prop}
\begin{proof}
    The proof follows by direct substitution.
\end{proof}
Next, we study the system in (\ref{eq:HJB_2nd_FI}) in more detail. Standard results from the theory of Riccati differential equations show that if  $\sqrt{\phi}$ is not equal to $\sqrt{4 \, (\varphi + \psi) \, e^{-1} \, k} \, \alpha$, then there is a unique non trivial solution to the ODE satisfied by $A$, given by 

 $$A(t) = \frac{\sqrt{\phi}}{\sqrt{\kappa}} \, \frac{1- e^{2 \, \sqrt{\phi \, \kappa} \, (T-t)} \, \beta}{1+ e^{2 \, \sqrt{\phi \, \kappa} \, (T-t)} \, \beta} \, ,$$
 where $\beta = \frac{\sqrt{\phi} + \sqrt{\kappa} \, \alpha}{{\sqrt{\phi} - \sqrt{\kappa} \, \alpha}}$ and $\kappa = 4 \, (\varphi + \psi) \, e^{-1} \, k$. If  $\sqrt{\phi} =  \sqrt{4 \, (\varphi + \psi) \, e^{-1} \, k} \, \alpha$, then $A = - \alpha$ solves the ODE.

Then, we propose the ansatz that  $B$ is linear in $u$ and $C$ quadratic in $u$. We have the following result. 
\begin{prop}
    Let  $\sqrt{\phi} \neq \sqrt{\kappa} \, \alpha$ and let $$b_0 \,, b_1\,, c_0\,, c_1\,, c_2 : [0,T] \rightarrow \mathbb{R}$$ be the unique solution to the system of ODEs 
    \begin{align}
    \begin{cases}
    \label{eq:systemODE_FI}
        0&= b_0' + \mu + 4 \,  k \, (\psi + \varphi) \, e^{-1} \, A \, b_0 \, , \\
        0&= b_1' - \eta \, \sigma \, \qconst - \eta \, b_1 + 4 \, (\psi + \varphi) \, e^{-1} \, A \,k\, b_1 + 4 \, e^{-1} \, \psi \,\qconst \, \sigma \, \gamma \, A  + 4 \, e^{-1} \, k \,\gamma\,\qconst \, \sigma  \, \psi \, A^2  \, , \\
        0&=c_0' + c_2 + \frac{e^{-1}}{k} \bigl (2 \, (\psi + \varphi) + 2 \, k \, A \,  ( \varphi + \psi)+ k^2 \, (\varphi + \psi) \, (A^2 + b_0^2)\bigl) \, , \\
        0&=c_1' - \eta c_1 + \frac{e^{-1}}{k} \bigl( \, 2 \, \psi \, k \, \sigma \, \gamma \, \qconst \, b_0 + k^2 \, ( \psi + \varphi) \, (2 \, b_0 \, b_1) + 2 \, k^2 \, \psi \, \gamma \,  \sigma \,\qconst \, A \, b_0 \bigl)\, , \\
        0&= c_2' - 2 \, \eta \, c_2 + \frac{e^{-1}}{k} \bigl( \, k^2 \, (\psi + \varphi) \, b_1^2 + 2 \, k^2 \, \psi \, \sigma \, \gamma \, \qconst \, A \, b_1 + 2 \, \psi \, k \, \sigma \, \gamma \, \qconst \, b_1 \bigl) \, , 
    \end{cases}
    \end{align}
    with terminal condition zero. Then 
    \begin{align}
        B(t,u) &:= b_0(t) + u \, b_1(t) \, , \\
        C(t,u) &:= c_0(t) + u \, c_1(t) + u^2 \, c_2(t) \, ,
    \end{align}
    solve the PDEs in (\ref{eq:HJB_2nd_FI}).
\end{prop}
\begin{proof}
    The proof follows by direct substitution.
\end{proof}

Existence and uniqueness of the ODE follows by the existence and uniqueness of the Riccati ODE solved by $A$, and that all the other ODEs are linear. The explicit formula can be obtain through the Duhammel's formula.

\section{Partial Information}

\label{sec: Imperfect Information}
\subsection{Deriving the filtered processes dynamics}
In this section we consider the case where the market maker does not have perfect information. More precisely,  the information  comes only from mid-prices, \textit{i.e.}, $\process[S]$. We define  $\process[\mcF^S]$ the filtration generated by $\process[S]$. We directly remark that $\mathcal{F}_t^S \subset \mathcal{F}_t$ for $t\in\mfT$, and that  $U_t$ is not $\mathcal{F}^S_t$-measurable. The goal of this part is to formulate a stochastic control problem with complete information that employs filters to obtain the best estimate of the unobserved fad process.\footnote{This is similar to recent works in the mathematical finance literature such as \cite{BARZY,knochenhauer2024continuous}.} 
Thus, we face the problem of filtering the fads $U_t$ with knowledge of $(S_u)_{u\in[0,t]}$, and where
\begin{align}
\d U_t &= f(U_t) \d t + \d B_t\,,\\
\d S_t &= h(U_t)\d t + \sigma \,\d\Bar{W}_t\,,
\end{align}
with $f(x) = -\eta\,x$, $h(x) = \mu - \eta\,\qconst \,\sigma x$, and where $\Bar{W}_t = \pconst\, Z_t + \qconst\,B_t$ is a Brownian Motion, with $\langle B,\Bar{W} \rangle_t = \qconst\,t$. We define the innovation process $\process[I]$ as 
\begin{equation}
    \d I_t = \d S_t - \pi_t(h) \d t\,,\qquad I_0\in\mathbb{R}\,,
\end{equation}
where  $\pi_t(h):= \mathbb{E}[h(U_t) | \mathcal{F}_t^S]$.

The following proposition shows that $\sigma^{-1} \, (I_t - I_0)$ is a Brownian motion which later on we use to derive the dynamics of the filtered processes.
\begin{prop}
    If  $\mathbb{E}[\,\int_0^t \lvert h(U_s) \rvert \d s] < \infty$, then $\Tilde{I}_t := \sigma ^{-1}\, (I_t - I_0)$ is a $\mathcal{F}_t^S$ Brownian Motion under $\mathbb{P}$.
\end{prop}
\begin{proof}
   The proof is standard so we only describe the steps. First, we prove that   $I$ is an $L^2(\Omega, \mathcal{F}^S, \mathbb{P})$ martingale. Second, we show that $ \langle I \rangle _t =  \sigma^2 t$, and we conclude with the Lévy characterisation of the Brownian Motion; see Proposition 2.30 in \cite{bain2009fundamentals}. Finally, the condition  $\mathbb{E}[\,\int_0^t \lvert h(U_s) \rvert \d s] < \infty$ holds since $\process[U]$ is a continuous Gaussian process.
\end{proof}

\begin{prop}
\label{marting:rpz}
    If $\mathbb{P}(\int_0^t \lvert \pi_t(h) \rvert^2 ds < \infty) = 1$, then for all $L^2(\Omega, \mathcal{F}^S, \mathbb{P})$-martingales $\process[\eta]$, we have a martingale representation of the form $$\eta_t = \eta_0 +\int_0^t \nu_s \, \sigma^{-1} \, \d I_s\,,$$ where $\process[\nu]$ is $\process[\mcF^S]$-progressively measurable. 
    \begin{proof}
       The proof follows from Proposition 2.31 in  \cite{bain2009fundamentals}. The condition  $\mathbb{P}(\int_0^t \lvert \pi_t(h) \rvert^2 ds < \infty) = 1$ holds since we have $\mathbb{E}[\int_0^t | \pi_s(h) |^2 \, \d s] < \infty$ by Jensen inequality and the fact that $\process[U]$ is a continuous Gaussian process.
    \end{proof}
\end{prop}
Next we define 
\begin{equation}
    \hat{U}_t := \mathbb{E}[U_t | \mathcal{F}_t^S] \, ,
\end{equation} and $\Gamma_t := \hat{U}_t + \int_0^t \eta \, \hat{U}_s \d s$. We then have the following result.
\begin{prop}
The process    $\process[\Gamma]$ is an $\process[\mcF^S]$-martingale under $\mathbb{P}$.
\end{prop}
\begin{proof}
By construction $\Gamma_t$ is $\mathcal{F}_t^S$-measurable for all $t > 0$. Moreover, since $U \in L^1(\Omega, \mathcal{F}, \mathbb{P})$, then, by the Jensen inequality, $\hat{U} \in L^1(\Omega, \mathcal{F}^S, \mathbb{P})$. Then, take $0\leq s \leq t$, it follows that 
\begin{equation}
 \mathbb{E}[\Gamma_t | \mathcal{F}_s^S] = \mathbb{E}[\mathbb{E}[U_t | \mathcal{F}_t^S] | \mathcal{F}_s^S] + \int_0^s \eta\,\hat{U}_r \d r + \int_s^t \eta \, \mathbb{E}[\hat{U}_r | \mathcal{F}_s^S]\,  \d r\,.
\end{equation}
The first term equals $\mathbb{E}[U_t | \mathcal{F}_s^S]$, using $\mathcal{F}_s^S \subset \mathcal{F}_t^S$ and  the tower property. The third equals $\int_s^t \eta\,\mathbb{E}[U_r | \mathcal{F}_s^S]\,  \d r$ by the same argument. Moreover, $$\mathbb{E}[U_t | \mathcal{F}_s^S] = \mathbb{E}\left[- \eta \int_s^t U_r  \, \d r| \mathcal{F}_s^S \right] + \mathbb{E}[B_t - B_s | \mathcal{F}_s^S] + \hat{U}_s = \mathbb{E}\left[- \eta \int_s^t U_r\, \d r | \mathcal{F}_s^S\right]  + \hat{U}_s, $$ 
by the tower property. 
Finally, it follows that $\mathbb{E}[\Gamma_t | \mathcal{F}_s^S] = \Gamma_s$ and
  $\process[\Gamma]$ is an $\process[\mcF^S]$-martingale under $\mathbb{P}.$
\end{proof}

The following theorem presents the dynamic of the filtering process $\process[\hat{U}]$ which we use to formulate a stochastic control problem with complete information.
\begin{thm}
\label{thm:filter}
The process $\hat{U}_t$ satisfies 
\begin{equation}
\d \hat{U}_t = - \eta\, \hat{U}_t\, \d t + \sigma^{-1} \, (-\, \eta\, \qconst \,\hat{P}_t + \qconst) \, \d I_t\,,
\end{equation}
where $\hat{P}_t = \mathbb{E}[(U_t - \hat{U}_t)^2]$ is the conditional variance and is the unique solution of the Riccati equation 
    \begin{equation}
        \frac{\d }{\d t} \hat{P}_t=- \eta^2 \, \qconst^2 \, \hat{P}_t^2 - \hat{P}_t \, (2 \eta - 2 \eta \, \qconst^2) + \pconst^2 \, , \qquad \, \qquad \, \hat{P}_0 = 0 \, .
    \end{equation}
\end{thm}

Given the dynamics of $\process[\hat{U}]$, the dynamics of $\process[S]$ can be rewritten using the Brownian motion $\process[\Tilde{I}]$ as 
$$\d S_t = \sigma \, \d \Tilde{I}_t + (\mu - \eta \, \sigma \, \qconst \, \hat{U}_t) \, \d t \, .$$

Following the motivation in the previous section, our goal is to solve a control problem of the form
\begin{equation}
  \sup_{\delta^{a,b}}\mathbb{E}^{\delta}\left[ X_T +Q_T \, S_T - \alpha\,Q_T^2 - \phi\,\int_0^T Q_u^2\,\d u \right]\, .
\end{equation}

In line with the full information model, here, we model $\process[{N}^a]$ and $\process[{N}^b]$ to be counting processes with intensities
\begin{align}
 {\lambda}_t^a &=\Big(\varphi  \, e^{-k \, \delta_t^a} + \psi \, e^{-k  \delta^a_t -\gamma \,( \qconst \,\sigma \, \hat{U}_t ) \vee \mathcal{S}^-}\Big) \, \mathds{1}_{\{Q_{t^-} > \underline{q} \}} \,,\label{eq:lambdas hat_a}\\   
 {\lambda}_t^b &=\Big(\varphi  \, e^{-k \, \delta_t^b} + \psi \, e^{-k  \delta^b_t  +\gamma \,( \qconst \,\sigma \, \hat{
 U}_t ) \, \wedge \, \mathcal{S}^+)} \Big) \, \mathds{1}_{\{Q_{t^-} < \overline{q} \}}  \label{eq:lambdas hat_b}\,,
\end{align}
where $\varphi,\psi,k,\mathcal{S}^+,\mathcal{S}^-$ are as before.\footnote{One way to motivate this choice of stochastic intensity is that informed traders filter the fad, and then, market arrivals follow the same structure as before, but now depending on the filter of the fad. Another way to understand this structure is that the market maker uses this model as an approximation of reality; that is, the market maker faces informed traders that know exactly what the fad is, but without the knowledge of $\process[U]$, she uses the filter $\process[\hat{U}]$ as a proxy. In Section \ref{sec : beyond approx} we explore an alternative formulation.}
The inventory  ${Q}_t$ and the cash ${X}_t$ are defined as in \eqref{eq: inventory} and \eqref{eq: cash}.
Next, we state the derived full information control problem we associate with the partial information case. The agent wishes to solve the control problem 
\begin{equation}
 V(t,q,\hat{u})  =  \sup_{\delta^{a,b} \, \in \, \mathcal{A}_{t,T}}\mathbb{E}_{t,q,u}^{\delta}\left[ -\alpha \, Q_T^2 + \int_t^T (\delta^a_s \, \lambda_s^a + \delta^b_t \, \lambda_s^b - \phi \, Q_s^2 +  ( \mu - \eta \, \qconst \, \sigma \, \hat{U}_s) \, Q_s \, \d s \right]
 \, ,
\end{equation}
where $\mathcal{A}_{t,T}$ are controls bounded from below by $-\delta_\infty$ and predictable with respect to the filtration $\process[\mathcal{\widehat{F}}]$ generated by $\process[{\Tilde{I}}]$, $\process[\hat{U}]$ and $\process[{N}^{a,b}]$ and 
where
\begin{align*}
  \d \hat{U}_t &=- \eta \, \hat{U}_t \, \d t  - \, (\hat{P}_t  \, \qconst \, \eta -  \qconst )\, \d \Tilde{I}_t \, , \qquad \, \hat{U}_0 = 0 \, ,   \\
  \d {Q}_t &= \d {N}_t^b - \d {N}_t^a \, ,\qquad Q_0\in \mcQ \, .
\end{align*}
The processes $\process[{N}^{a,b}]$ count the number of orders filled by the market maker and have controlled stochastic intensities $\process[{\lambda}^{a,b}]$ defined in \eqref{eq:lambdas hat_a} and \eqref{eq:lambdas hat_b}.

Before characterising the solution to the partial information problem, it is important to provide a rigorous characterisation of the probability space we employ.

\subsection{Formal derivation of the probability space}
Consider $\Omega_d$ the set of increasing piecewise constant càdlàg functions from $[0,T]$ into $\mathbb{N}$ with jumps equal to one and $\Omega_c$ the set of continuous functions from $[0,T]$ into $\mathbb{R}$. We define $\Omega = \Omega_c \times \Omega_d^2$ as the sample space.
 We let $(\Tilde{I},N_t^a,N_t^b)_{t \in \mfT}$ be the canonical process on $\Omega$. The associated filtration is $\mathcal{\widehat{F}}= ( \mathcal{F}_t^c \otimes \mathcal{F}_t^d \otimes \mathcal{F}_t^d)_{t \in \mfT}$ where $\process[\mathcal{F}^d]$ (resp.~$\process[\mathcal{F}^c]$) is the right continuous completed filtration associated with $N^a$ (or $N^b$) (resp.~$\Tilde{I}$).
We denote by $\mathbb{P}_0$ the probability measure on $(\Omega, \mathcal{F})$ such that $\Big(M^a_s = N^a_s - \int_0^t\mathds{1}_{ \{N_{s^-}^b-N_{s^-}^a > \underline{q}\}} \, \d s \,, M^b_s = N^b_s -\int_0^t \mathds{1}_{\{N_{s^-}^b-N_{s^-}^a < \overline{q}\}} \, \d s\Big)_{s \in \mfT}$  are martingales and $\process[\Tilde{I}]$ is a Brownian motion.\footnote{The construction of $\process[N^a]$ and $\process[N^b]$ follows the same idea as in the full information setting.} We also assume independence between the processes $\process[N^a]$, $\process[N^b]$ and the Brownian motion. Next, we define  $\mathbb{X}_t = (\hat{U}_t, Q_t)$ to be the state of the system, where the dynamics of  $\process[\mathbb{X}]$ are given by
    \begin{align}
        \begin{cases}
         &\d Q_t = \d N^b_t - \d N^a_t \,, \qquad \, Q_0 \in \mathcal{Q} \, , \\
        &\d \hat{U}_t = - \eta \, \hat{U}_t \, \d t + (\qconst - \hat{P}_t \, \qconst \, \eta) \, \d \Tilde{I}_t\, , \qquad \, \hat{U}_0 = 0 \, ,
        \end{cases}
    \end{align}
and we  define an equivalent probability measure $\mathbb{P}^\delta$ through the Doléans-Dade exponential of $$\mathcal{Z}_t^\delta =\int_0^t (\lambda^a_s -1 )\,  \d M^a_s + \int_0^t (\lambda^b_s -1 ) \, \d M^{b}_s \, .$$
Under $\mathbb{P}^\delta$, $\process[\Tilde{I}]$ is a Brownian Motion, and $\process[N^a]$ and $\process[N^b]$ are Poisson processes with stochastic intensities given by 
\begin{align}
         & \lambda^a(\hat{U}_t,Q_{t^-}, \delta^a_t) := \lambda^a_t = (\varphi + \psi \, e^{-\gamma \, (\sigma \,\qconst \,\hat{U}_t \, \vee \,  \mathcal{S}^-)}) \, e^{-k \, \delta^a_t} \, \mathds{1}_{\{ Q_{t^-} > \underline{q} \}} \, , \\
     &  \lambda^b(\hat{U}_t,Q_{t^-}, \delta^b_t) := \lambda^b_t = (\varphi + \psi \, e^{ \gamma \, (\sigma \,\qconst \,\hat{U}_t \, \wedge \, \mathcal{S}^+)}) \, e^{-k \, \delta^b_t} \, \mathds{1}_{\{Q_{t^-} < \overline{q} \}} \, . 
\end{align}

\subsection{Characterisation of the solution}

The corresponding HJB equation associated with the above control problem is 
\begin{equation}
\begin{aligned}
\label{eq:HJB}
    0 &= \partial_t V(t,q,\hat{u}) - \phi \,  q^2 + q \, (\mu - \sigma \, \eta \, \qconst \hat{u}) \\ &- \partial_u V(t,q,\hat{u}) \, \eta \, \hat{u} + \frac{1}{2}\partial_{u,u}^2 V(t,q,\hat{u})(\hat{P}_t \,  \qconst \, \eta - \qconst)^2  \\ &+  \sup_{\delta^a}\{(\varphi \, e^{- k \, \delta^a} + \psi \,  e^{-k \, \delta^a - \gamma \, (\qconst \, \sigma \, \hat{u}) \, \vee \, \mathcal{S}^-}) (\delta^a + V(t,q-1,\hat{u}) - V(t,q,\hat{u}))\}  \mathds{1}_{q>\underline{q}}\\ 
    &+ \sup_{\delta^b}\{( \varphi \, e^{-k \, \delta^b} + \psi \, e^{-k \, \delta^b + \gamma \, (\qconst \, \sigma \, \hat{u}) \, \wedge \, \mathcal{S}^+)}) (\delta^b + V(t,q+1,\hat{u}) - V(t,q,\hat{u}))  \} \mathds{1}_{q<\overline{q}}\, ,
\end{aligned}
\end{equation}
where the  optimal controls in feedback form are given by
\begin{align*}
\delta^{a,*} &= \Big(\frac{1}{k}- V(t,q-1,u) + V(t,q,u) \Big) \vee -\delta_\infty^a\,,\\
  \delta^{b,*} &= \Big(\frac{1}{k} - V(t,q+1,u) + V(t,q,u) \Big) \vee -\delta_\infty^b\,.
\end{align*}
This transforms the HJB equation into the following PDE
\begin{equation}
\begin{aligned}
\label{eq: PDE incomplete inf}
 0 &= \partial_t V(t,q,u) - \phi \, q^2 + q \, (\mu - \sigma \, \eta \, \qconst u) \\ &- \partial_u V(t,q,u) \, \eta \, u + \frac{1}{2}\partial_{u,u}^2 V(t,q,u)(\hat{P}_t \,  \qconst \, \eta -  \qconst)^2 \\ &+  \frac{1}{k} \exp(- 1 + k[V(t,q-1,u) - V(t,q,u)]) \, (\varphi + \psi \,e^{-\gamma \, (\qconst \, \sigma \, u)  \, \vee \, \mathcal{S}^-}) \mathds{1}_{q>\underline{q}} \, \mathds{1}_{\{\delta^{a,*} > -\delta_\infty^a\}} \\
 &+ (\varphi \, e^{ k \, \delta^a_\infty} + \psi \,  e^{k \, \delta^a_\infty - \gamma \, (\qconst \, \sigma \, \hat{u}) \, \vee \, \mathcal{S}^-}) (-\delta^a_\infty + V(t,q-1,\hat{u}) - V(t,q,\hat{u})) \, \mathds{1}_{q>\underline{q}} \, \mathds{1}_{\{\delta^{a,*} = - \delta_\infty^a\}} \\ 
    &+ (\varphi \, e^{k \, \delta^b_\infty} + \psi \, e^{k \, \delta^b_\infty + \gamma \, (\qconst \, \sigma \, \hat{u}) \, \wedge \, \mathcal{S}^+)}) (-\delta^b_\infty + V(t,q+1,\hat{u}) - V(t,q,\hat{u})) \, \mathds{1}_{q<\overline{q}}\, \mathds{1}_{\{ \delta^{b,*} = - \delta_\infty^b\}}
 \\&+ \frac{1}{k} \exp(- 1 + k[V(t,q+1,u) - V(t,q,u)]) \, (\varphi + \psi \,e^{\gamma \, (\qconst \, \sigma \, u) \, \wedge \, \mathcal{S}^+)}) \mathds{1}_{q<\overline{q}}\,\mathds{1}_{\{\delta^{b,*}>-\delta_\infty^b\}} ,
\end{aligned}
\end{equation}
with terminal condition $V(T,q,u) = - \alpha \, q^2.$

\subsubsection{Viscosity solution: incomplete information problem}
Again, the value function $V$ is a viscosity solution of the HJB equation (\ref{eq:HJB}). We remark that the steps we followed in the proof of the viscosity solution of Section \ref{sec: Perfect Information} still hold in this setting. We omit the details for brevity. 

\subsubsection{Approximate solution to the partial informed problem}
Similar to Section \ref{subsec:approx}, we take $-\delta_\infty^{a,b} = - \infty$, $\overline{q}= -\overline{q} = + \infty$ and $\mathcal{S}^+ = - \mathcal{S}^- = + \infty$. We assume that $V(t,q-1,u) - V(t,q,u)$ and  $V(t,q+1,u) - V(t,q,u) $ are small and we perform 
a second order Taylor expansion around zero for $ p \rightarrow \frac{1}{k} \exp(k \, p)$ as before. Furthermore, we  perform a first order Taylor expansion of the exponential functions of $\mp \, \qconst \, \gamma \, \sigma \, u$. To simplify the notation further, we define $x_1 := \mu - \sigma \, \qconst \, \eta \, u $, and $x_2 := -\hat{P}_t \, \qconst \, \eta + \qconst$.
As a consequence of the second order approximation, and employing an ansatz for $V$ that is quadratic in $q$, we show that the  approximate solution satisfies a system of PDEs that can be easily solved numerically. The result is summarised in the following proposition.
\begin{prop}
Let  $A:  \mathbb{R}_+ \rightarrow \mathbb{R}$ and $B \,, C :  \mathbb{R}_+ \times \mathbb{R} \rightarrow \mathbb{R}$ be a solution to the system 
\begin{align}
\begin{cases}
\label{eq:SOA_PI}
\partial_t A &= - \phi + 4 \, e^{-1} \,  (\varphi + \psi) \, k \, A^2 \, ,\\
\partial_t B &= \partial_u B \, \eta \, u - \frac{1}{2} \partial_{u,u}^2 B \, x_2^2  + x_1 + \psi \, e^{-1}\bigl( 4A \, \qconst \, \gamma \, \sigma \, u  + 4 \, k\qconst \, \sigma \, \gamma \, u \, A^2 \bigl) \\& \qquad + 4 \, e^{-1} \,  (\varphi + \psi) \, k \, A \, B \,  ,\\
\partial_t C &=  \partial_u C \, \eta \, u - \frac{1}{2} \partial_{u,u}^2 C \, x_2^2+\psi \,  e^{-1} \Bigl (- 2B \, \qconst \, \sigma \, \gamma \, u  + 2A\, B \, k \, \qconst \, \sigma \, \gamma \, u \Bigl)\\& \qquad+ \frac{e^{-1} \, (\varphi + \psi)}{k} \, (2 - 2 \, k \, A + k^2 \, A^2 +  k^2 \, B^2 )\, ,
\end{cases}
\end{align}
with terminal conditions 
\begin{equation}
    A(T) =  \alpha \, , B(T,u) =  0 \, , C(T,u) = 0\, .
\end{equation}
Then, $\widecheck{V}(t,q,u) = - q^2 \, A(t) - q \, B(t,u) - C(t,u)$ is a solution to the approximation of the PDE in \eqref{eq: PDE incomplete inf}.
\end{prop}
\begin{proof}
    The proof follows by direct substitution of the  derivatives of $\widecheck{V}$ in the quadratic approximation we obtained for the PDE satisfied by $V$.
\end{proof}
Next, we use an ansatz in which the function $B$ is linear in $u$ and $C$ is quadratic in $u$. We have the following result.
\begin{prop}
    Let $A : [0,T] \rightarrow \mathbb{R}$ solve the Riccati ODE 
    \begin{equation}
        \label{eq:Riccati_PI}
         A'= - \phi + 4 \, e^{-1} \,  (\varphi + \psi) \, k \, A^2 \, , \qquad A(T) = \alpha \,.
    \end{equation}
 Let $(b_0\, , b_1\,  , c_0\, , c_1\, , c_2) : \mathbb{R}^+ \rightarrow \mathbb{R}$ solve the system of ODEs
\begin{align}
\label{eq:coupled_ODE_PI}
\begin{cases}
    b_0' &= \mu + 4 \, e^{-1} \, ( \varphi + \psi) \, k \, b_0 \, A \, ,\\
    b_1' &= \eta \, b_1 - \eta \, \sigma \, \qconst + \psi \, e^{-1} \, (- 4 \, A \, \gamma \, \sigma \, \qconst + 4 \, k \, \qconst \, \gamma \, \sigma \, A^2) + 4 \, e^{-1} \, (\varphi + \psi) \, k \, b_1 \, A \,  , \\
    c_0'&= - x_2^2 \, c_2 + \frac{e^{-1}}{k} \, ( \varphi + \psi) \, (2 - 2 \, k \, A + k^2 \, A^2 + k^2 \, b_0^2) \,  ,\\
    c_1' &= \eta \, c_1 - 2 \, e^{-1} \, \psi  \, b_0 \, \qconst \, \gamma \, \sigma + \Psi \, e^{-1} \, 2 \, k \, \gamma \, \sigma \, \qconst \, b_0 \, A + 2 \, \frac{e^{-1}}{k} \, (\varphi + \psi) \, k^2 b_0 \, b_1 \, , \\
    c_2'&= 2 \, \eta \, c_2 -2 \,\psi \, e^{-1} \, b_1 \, \gamma \, \sigma \, \qconst + 2 \, \psi \, e^{-1} \, A \, k \,\gamma \, \sigma \, \qconst \, b_1 + \frac{e^{-1}}{k} \, (\varphi + \psi) \, b_1^2 \, k^2 \, ,
    \end{cases}
\end{align}
Define $\widehat{B} \, ,\widehat{C} : [0,T] \times \mathbb{R} \rightarrow \mathbb{R}$ as 
    \begin{align*}
            \widehat{B}(t,u) &= u \, b_1(t) + b_0(t) \, ,\\
            \widehat{C}(t,u)&= u^2 \, c_2(t) + u \, c_1(t) + c_0(t) \, .
    \end{align*}
It follows that ($A, \widehat{B}, \widehat{C}$) solves the system of PDEs in (\ref{eq:SOA_PI}).
\end{prop}

\subsection{Connection with the separation principle}
In this section we compute the filter of the process $\process[U]$ with respect to the filtration generated by $\process[S]$, and then we construct new stochastic processes $\process[N^{a,b}]$ where their stochastic intensities depend on $\process[\hat{U}]$ instead of $\process[U]$. 
Then, as shown above, the optimal strategies are akin to those from the full information setup but using $\process[\hat{U}]$ instead of $\process[U]$. This modelling approach is similar to employing the separation principle. Quoting \cite{bensoussan2018estimation}, ``the result says that it is optimal to estimate the state
and then apply the rule to the estimate, instead of the state itself''. Of course, this  is not the optimal approach when one works in a non-linear-quadratic framework; however, as we show below, there is non-negligible economic value when employing this modelling approach.

\section{Sensitivity analysis and simulations}
\label{sec : Simulations}
In this section we investigate the optimal strategies derived in the previous sections. Our main goals are: (i) to understand the value of information for the market maker, (ii) to investigate the dependence of the optimal quotes on key model parameters, and (iii) to study the performance of the market maker as a function of the percentage of informed traders in the market.
To this end, we develop a sensitivity analysis of the optimal solutions with respect to key parameters and we simulate the model.
Similar to \cite{boyce2024market}, we consider the following baseline set of parameters:
$\alpha = 0.001 \, , 
 \sigma = 1 \, , 
 \phi = 0.1 \,  ,
 \varphi = 15 \, ,
k = 1 \, , 
\gamma = 1 \, , 
 T = 1 \, , 
\mu =0 \, , 
\eta = 10
\, .$
The value of $\qconst$ modulates the extent to which the fad influences the price process. We set  $\qconst = 0.6$  and $\pconst = \sqrt{1 -\qconst^2}$. Note that $k=\gamma$ and, therefore, informed traders are equally sensitive  to the displacement of the quotes and to the fad.

Note that the baseline expected number of market arrivals on the bid and ask side, i.e.,
\eqref{eq: intensity lambda a}-\eqref{eq: intensity lambda b} without displacements,
is $\varphi \, T+ \psi \, \int_0^T\mathbb{E}[e^{\pm \gamma \, \sigma \, \qconst \, U_t}]\d t$. In what follows, we set 
\begin{equation}
\label{eq_psi}
\psi = \frac{30-\varphi \, T}{\int_0^T\mathbb{E}[e^{\pm \gamma \, \sigma \, \qconst \, U_t}]\d t}
\end{equation}
so that the baseline proportion of market arrivals from informed and uninformed traders is constant (50-50\%) and the expected  arrivals are $30$. 

In this setting, the fad only affects the market maker through the arrival rates of orders \eqref{eq: intensity lambda a}-\eqref{eq: intensity lambda b} because the mark-to-market of the inventory  in \eqref{eq: payoff} at time $T$ is evaluated at the mid-price and not at the fundamental. In Section \ref{sec:fundamental} we study the case where the market marker values the inventory at fundamental.

\subsection{Optimal strategies under full information}
\label{FULL}

In Figure \ref{fig:optimal quotes} we show the optimal ask and bid displacements of the market maker (in solid and dashed lines respectively) in the full information setting as a function of the fad (in the $x$-axis) varying the inventory level in the colour bar (the lighter the colour, the higher the inventory). 
The left panel shows the optimal displacements when $\qconst = 0.3$, the middle panel shows the optimal displacements when $\qconst = 0.6$, and the right panel shows the optimal displacements when $\qconst = 0.9$. 

As expected, and confirming classical results, the higher is the market maker's inventory the lower is the ask displacement $(\delta^{a*,\text{FI}})$ and the higher the bid displacement $(\delta^{b*,\text{FI}})$. 

The effect of the fad on the displacements is asymmetric. The ask displacement $(\delta^{a*,\text{FI}})$ is decreasing in the fad $U_t$, the opposite is observed for the bid displacement $(\delta^{b*,\text{FI}})$. The rationale of this result is that as $U_t$ increases, $\lambda^a$ decreases and $\lambda^b$ increases due to the informed traders' market activity. Given the performance criterion \eqref{eq: payoff}, the market maker values her open inventory at mid-price, and therefore she decreases the price of liquidity in the ask and increases the price of liquidity in the bid to manage the inventory in an optimal way.
Moreover, in agreement with this argument, the absolute value of the slope of the displacements with respect to the fad component decreases as $\qconst$ becomes smaller. Indeed, the closer $\qconst$ is to zero, the less impact the fad has on the mid-price process and on the arrival rates of informed traders (informed and uninformed traders behave in a similar way), and therefore the dependence of the optimal displacements on the fad weakens. In the limit, when $\qconst = 0$, the slope is zero.

\begin{figure}[H]
\centering
\includegraphics[width=0.8\textwidth]{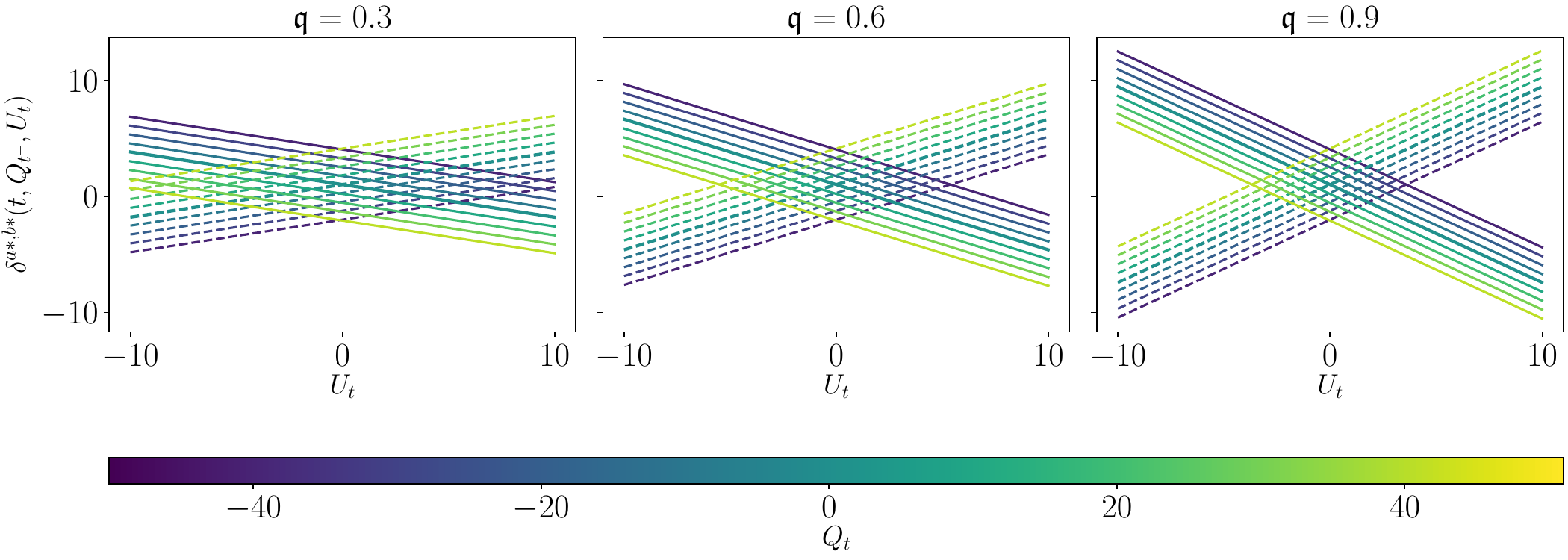}
\caption{Ask and bid displacements as a function of the fad ($x$-axis) and the inventory of the market maker (colour bar) under full information.
Solid lines: optimal ask displacement; dashed lines: optimal bid displacement. Left panel is for $\qconst = 0.3$, middle panel is for $\qconst = 0.6$, and right panel is for $\qconst = 0.9$.}\label{fig:optimal quotes} 
\end{figure} 

For a given path of the fad and of the asset price,  Figure \ref{fig:Q_U}  shows the dynamics of the inventory $\process[Q]$ (in red) and the fad $\process[U]$ (in blue) for two values of $\gamma$ ($\gamma=0.1$ and $\gamma=10$) which modulates the sensitivity of informed traders to the fad. Note that the baseline sensitivity is $\gamma = k = 1$.
We observe that a high value of $\gamma$ leads to a high sensitivity in the number of trades (by informed traders) with respect to $U_t$. On the contrary, when $\gamma$ is small, the executed trades are less related to changes of the value of $U_t$. There is also a higher correlation between the fad $\process[U]$ and the inventory $\process[Q]$ as $\gamma$ increases.\footnote{The mean correlation between the realizations of the paths of $U$ and $Q$ (1000 timesteps)  over 1000 simulations is $0.49$ when $\gamma=0.1$ and $0.57$ when $\gamma=10$; the correlation between the changes of the processes also  increases. }

\begin{figure}[H]
\centering
\includegraphics[width=0.8\textwidth]{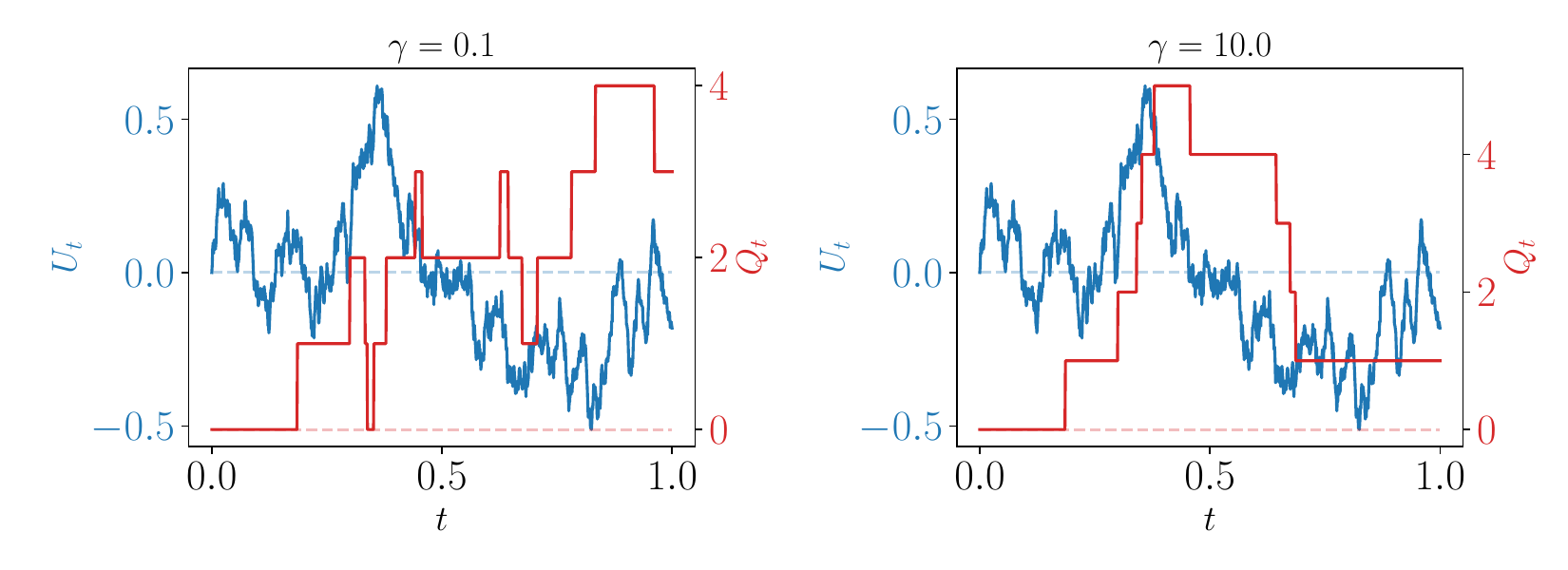}
\caption{Sample paths followed by $\process[Q]$ and $\process[U]$ for $\gamma =0.1$ (left panel) and $\gamma = 10$ (right panel).}\label{fig:Q_U} 
\end{figure}

\subsection{Simulations of optimal strategies under partial information}

Figure \ref{fig:filter00}  shows a simulation path of $\process[U]$ and its filter $\process[\hat{U}]$ for three 
 values of $\qconst$.
We observe that the closer $\qconst$ is to one, the closer the filter is to the fad process. Indeed, the more the price is driven by the fad, the more the filter mimics the true trajectory because the signal is more precise. On the other hand, when $\qconst$ approaches zero, knowledge of the price is not valuable to filter the fad, and so the filter becomes uninformative and collapses around the mean of the fad, which is zero.

\begin{figure}[H]
\centering
\includegraphics[width=0.9\textwidth]{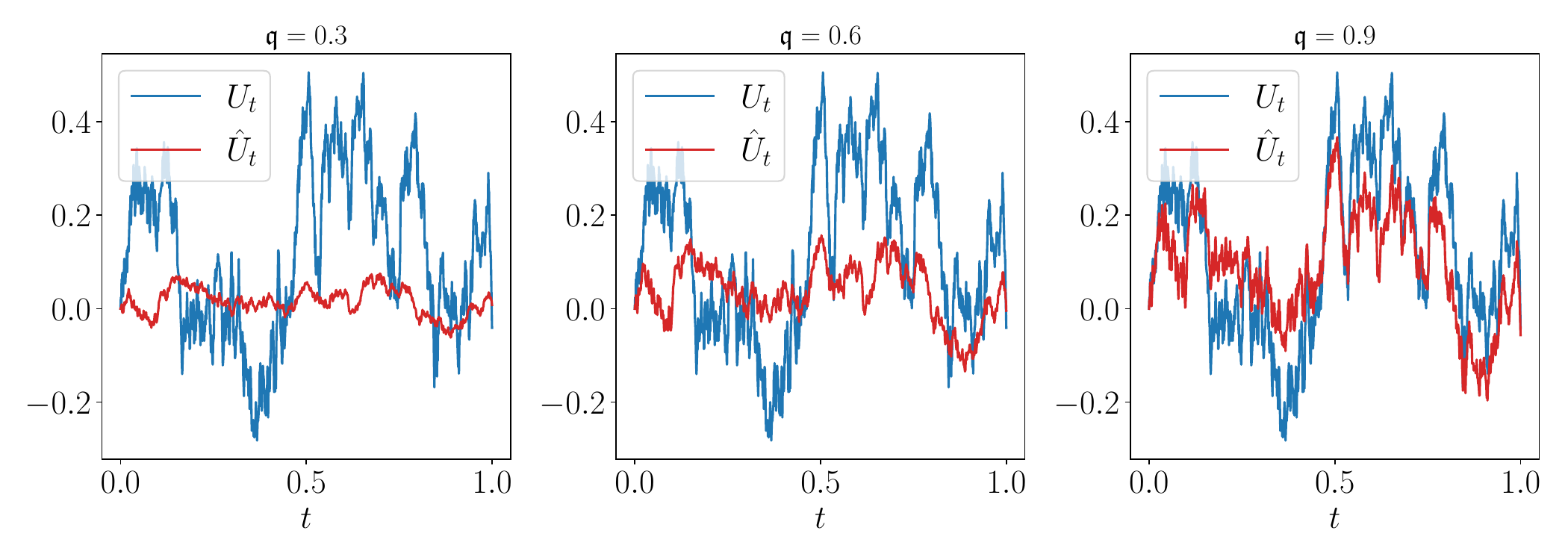}
\caption{Comparison between the process $\process[U]$, in blue, and the filter $\process[\hat{U}]$, in red, for (from left to right) $\qconst = 0.3$, $\qconst = 0.6$, and $\qconst = 0.9$.}\label{fig:filter00} 
\end{figure} 

The displacements obtained under partial information coincide with those obtained under full information as a function of $\process[\hat{U}]$ instead of $\process[U]$. Indeed, the PDE associated with the value function under partial information differs from the full information value function only in the volatilities of the processes $\process[U]$ and $\process[\hat{U}]$. However, the volatility processes do not appear in the linear and quadratic terms in $q$ (of the approximations), thus, their effects vanish when we compute the optimal displacements. Therefore, 
the pictures of the partial information displacements coincide with those in Figure \ref{fig:optimal quotes}, as a function of $\hat{U}$ instead of $U$.

\subsection{Performance evaluation}

We compare three strategies using the performance criterion in \eqref{eq: payoff}
as the yardstick: 
(i) the optimal strategy obtained for the full information setting (FI, $\delta^{*,\mathrm{FI}}(t,Q_{t^-}, U_t)$), 
(ii) the Cartea--Jaimungal--Penalva strategy \cite[Chapter 10]{CART-book} (CJP, $\delta^{*\mathrm{CJP}}(t,Q_{t^-})$), 
and (iii) the optimal strategy from the incomplete information setting (PI, $\delta^{*,\mathrm{PI}}(t,{Q}_{t^-}, \hat{U}_t)$). 
The ground truth in the simulations is that of the full information case. 

We consider the above set of model parameters with $Q_0 = 0$, $S_0 = 100$, $X_0 = 0$, and we carry out 100,000 simulations. We discretise $[0,T]$ in 1,000 timesteps. In Table \ref{table1} we report the mean performance and  standard deviation (in brackets) of the three strategies as the parameters of the model change. On each line, as a parameter ($\qconst, \gamma, \eta$)  changes, $\psi$ (the parameter that captures the presence of informed traders in the market) is calibrated according to \eqref{eq_psi} in a way that the expected number of market arrivals (the arrival rate without considering the effect of the displacements) in the full information setting is $30$ and the informed-uninformed composition is 50-50\%. The bottom rows of the table show the results when stressing the market composition.

As a parameter changes, we have three different types of effects: (i) the evaluation of the performance criterion changes, (ii) the expected number of market arrivals
changes, (iii) the informed-uninformed traders market composition changes. By imposing \eqref{eq_psi}, we neutralise the latter two effects and we concentrate on the first one. In Appendix \ref{NORENOR} we do not keep fixed the composition of the market and the number of trades;\footnote{See Figures \ref{fig:PnL_qconst_2}, \ref{fig:PnL_gamma_2}, \ref{fig:PnL_eta_2}, and Table \ref{table2}.} the main results of the analysis are confirmed.

The three strategies range according to the information available to the market maker. The first strategy is for when the market maker knows the ground truth, \textit{i.e.}, she knows the composition of the order flow \eqref{eq: intensity lambda a}-\eqref{eq: intensity lambda b} and of the asset price dynamics \eqref{eq: mid-price}-\eqref{eq:OU process} and she observes the order flow arrival as well as the realizations of the noise components ($B_t$ and $Z_t$), which allows her to disentangle the fundamental information from the fad component (fully informed market maker). Note that, in this setting, the market maker is not able to identify whether the trade comes from an informed or an uninformed trader. 
The second strategy is for when the market maker observes the order flow, the asset price, and does not know of the existence of a fad component that affects both the asset price and the arrival rates (uninformed market maker). The market maker believes in a misspecified model ($\qconst=0, \ \gamma=0$) but the simulation path for the asset price and the arrival rate reflect the existence of the fad.
The third strategy is for when the market maker knows the model, observes the order flow and the asset price, and filters the fad from price because she does not observe it. We remark that she knows that there is a fad component in the market and its relevance on the price process, \textit{i.e.}, the value of the parameters $\qconst$ and $\gamma$ (partially informed market maker).

To make the point more explicit, the optimal strategy $\delta^{*,\text{CJP}}(t,Q_{t-})$ is obtained assuming that the market maker ignores the existence of a fad in the market and, therefore, she ignores $\process[U]$, however, she receives market orders according to the rate $(\varphi + \psi \, e^{\mp \qconst \, \gamma \, U_t}) \, e^{- k \, \delta_t^{\text{CJP}}(t,Q_{t^-})}$ and, therefore, her performance is affected by the fad. As already observed, the fad component impacts the optimal displacements: in the full information setting, the optimal ask displacement is decreasing in $U_t$, while the optimal bid displacement is increasing in $U_t$ ($\delta^{*,\mathrm{FI}}(t,Q_{t^-}, U_t)$) balancing the activity of informed traders on the two sides of the market; a similar effect is observed for the optimal  displacement under partial information ($\delta^{*,\mathrm{PI}}(t,Q_{t^-}, \hat{U}_t)$) replacing the fad with its filtered value.

For all parameter sets, the optimal strategy under full information outperforms the other two, and the optimal strategy of the partial information setting outperforms the CJP-strategy.

\begin{table}[H]
\centering
\begin{tabular}{l|ccc}
\hline
\hline
& \multicolumn{3}{c}{Performance criterion }\\
& \multicolumn{3}{c}{$X_T + Q_T\, S_T  - \alpha \, Q_T^2 - \phi \, \int_0^T Q_u^2 \, \d u $}\\
\noalign{\vskip-1mm}
\noalign{\vskip 2mm}
\hline
\hline
\noalign{\vskip-1mm}
\noalign{\vskip 2mm}
parameters & $\delta^{*,\mathrm{FI}}(t,Q_{t^-}, U_t)$ & $\delta^{*\mathrm{CJP}}(t,Q_{t^-})$ & $\delta^{*,\mathrm{PI}}(t,{Q}_{t^-}, \hat{U}_t)$\\
\noalign{\vskip-1mm}
\noalign{\vskip 2mm}
\hline
\hline
\noalign{\vskip-1mm}
\noalign{\vskip 2mm}
$\varphi= 15$, $\eta=10$, $\gamma=1$, $\qconst=0.6$  & 21.33 (4.94) & 21.18 (4.94) & 21.20 (4.95)\\
\hline
$\qconst = 0.0$   & 21.34 (5.10) & 21.34 (5.10) & 21.34 (5.10)\\
$\qconst = 0.2$   & 21.34 (5.08) & 21.32 (5.08) & 21.32 (5.09)\\
$\qconst = 0.4$   & 21.34 (5.03) & 21.27 (5.03) & 21.28 (5.03)\\
$\qconst = 0.6$   & 21.33 (4.94) & 21.18 (4.94) & 21.20 (4.95)\\
$\qconst = 0.8$   & 21.31 (4.82) & 21.06 (4.82) & 21.14 (4.82)\\
$\qconst = 1.0$   & 21.30 (4.64) & 20.91 (4.65) & 21.30 (4.64)\\
\hline
$\gamma = 0$     & 21.46 (4.97) & 21.34 (4.97) & 21.36 (4.97)\\
$\gamma = 1$     & 21.33 (4.94) & 21.18 (4.94) & 21.20 (4.95)\\
$\gamma = 2$     & 21.17 (4.92) & 21.01 (4.93) & 21.03 (4.93)\\
$\gamma = 3$     & 21.00 (4.91) & 20.82 (4.93) & 20.85 (4.93)\\
\hline
$\eta = 2.5$     & 21.23 (4.97) & 20.99 (5.00) & 21.05 (5.01)\\
$\eta = 5.0$     & 21.29 (4.95) & 21.09 (4.97) & 21.12 (4.97)\\
$\eta = 7.5$     & 21.32 (4.95) & 21.15 (4.95) & 21.17 (4.96)\\
$\eta = 10.0$    & 21.33 (4.94) & 21.18 (4.94) & 21.20 (4.95)\\
$\eta = 12.5$    & 21.33 (4.94) & 21.21 (4.94) & 21.23 (4.94)\\
\hline
$\varphi,\psi : 0\%$ informed     & 21.46 (4.97) & 21.34 (4.97) & 21.36 (4.97)\\
$\varphi,\psi  : 25\%$ informed   & 21.39 (4.96) & 21.27 (4.96) & 21.28 (4.96)\\
$\varphi,\psi  : 50\%$ informed   & 21.33 (4.94) & 21.18 (4.94) & 21.20 (4.95)\\
$\varphi,\psi  : 75\%$ informed   & 21.26 (4.93) & 21.10 (4.94) & 21.12 (4.94)\\
$\varphi,\psi  : 100\%$ informed  & 21.17 (4.92) & 21.01 (4.93) & 21.02 (4.93)\\
\hline
\hline
\end{tabular}
\caption{Average performance (with standard deviation) for a market maker who follows either (i) the full information optimal strategy,  (ii) the optimal strategy of Cartea-Jaimungal-Penalva (CJP), or (iii) the imperfect information optimal strategy. The performance for the baseline set of parameters is provided in the first row, then, the performance where a parameter changes (one at a time) is reported in the following rows. On each line (varying $\qconst$, $\gamma$, or $\eta$) $\psi$ is rescaled as in \eqref{eq_psi}. The last rows deal with variations of $\varphi,\psi$ keeping the expected number of market arrivals at 30.}
\label{table1}
\end{table}

\textbf{Effect of $\qconst$}: The sensitivity of the performance with respect to $\qconst$ is shown in Figure \ref{fig:PnL_qconst}.
In the limit, we observe two interesting results. 
As $\qconst \rightarrow 0$, the gap between the strategies vanishes conforming with Proposition \ref{prop_q_0} in the Appendix. This is because the impact of the fad on the mid-price and on the arrival rates vanishes and therefore the three strategies render the same performance. 
As $\qconst \rightarrow 1$, the optimal strategy in the full information setting and in the partial information setting tend to coincide in agreement with Proposition \ref{prop_q_1} in the Appendix. This is because in the limit, the filter matches the fad. Instead, the performance of the CJP strategy deteriorates as $\qconst$ increases because it becomes more misspecified.  

The performance under full information weakly declines in $\qconst$. The rationale for this effect is that as $\qconst$ increases the arrival of informed traders becomes more sensitive to the fad component.
That is, all else being equal, the (expected) fifteen arrivals from informed traders occur when the value of the fad is farther away from zero and the fundamental value differs more from the mid-price. Thus, market arrivals from the informed traders tend to be more ``sharp'' or more ``toxic'' as $\qconst$ increases  and, therefore, hurt the market maker more. Note that the market maker can only use the bid and ask displacement to manage the inventory from trading with both, informed and uninformed traders. The classical adverse selection effect applies here: the market maker copes with both informed and uninformed traders setting the bid and the ask price, and when trading becomes more 
``toxic'', the market maker sets a larger bid-ask spread  and experiences a poorer performance. 

Under partial information, the performance 
decreases and then increases; 
the minimum in the figure is around $\qconst = 0.8$. When $\qconst$ is low, the filter is imprecise and becomes stable around its long run mean (which is zero), as a consequence the performance of the partial information strategy looks similar to the CJP-strategy, see Figure \ref{fig:filter00}. As $\qconst$ increases, the quality of the filter improves and the market maker 
 reacts to the fad in a more precise manner.  Note that the performance gap between the PI-strategy and the CJP-strategy is increasing as a function of $\mathfrak{q}$. This is due to the ever-increasing importance of the fad in driving the price as $\qconst$ increases. In the left hand side of the plot, the performance of the PI-strategy is closer to the CPJ-strategy because the quality of the filter is poor as shown in Figure \ref{fig:filter00}.

\begin{figure}[H]
\centering
\includegraphics[width=0.5\textwidth]{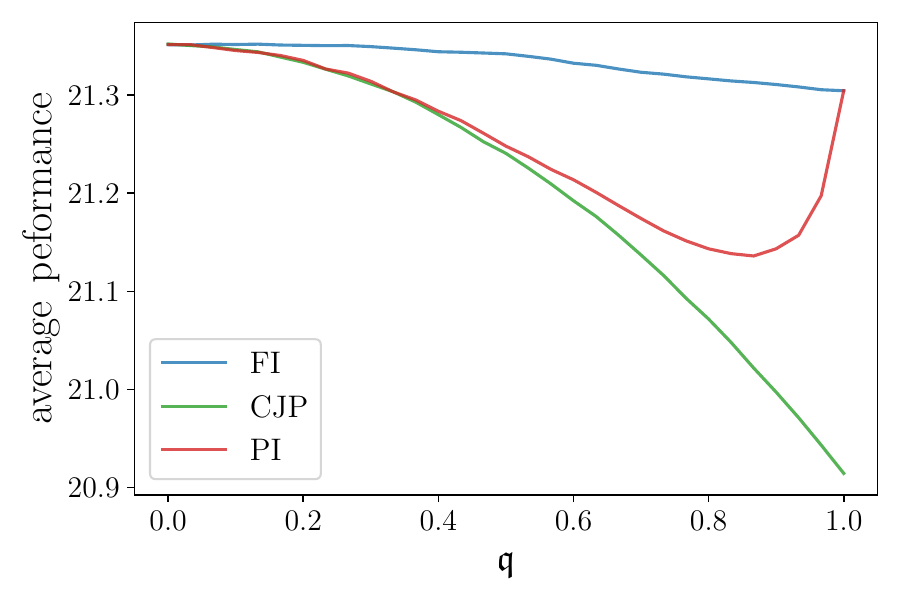}
\caption{Performance of the strategies as a function of $\qconst$. }\label{fig:PnL_qconst} 
\end{figure}

\textbf{Effect of $\gamma$}: As $\gamma$ increases, the performance of the three strategies decreases quasi linearly; see Figure \ref{fig:PnL_gamma}.
The rationale of this result is that, all else being equal, the higher the value of $\gamma$ the higher the expected number of arrivals coming from a counting process with intensity $\exp{(\pm \gamma \, \qconst \, \sigma \, U_t)}$. Since $\psi$ is rescaled according to \eqref{eq_psi}, the informed-uninformed composition in terms of market arrivals remains constant as $\gamma$ changes. However, as informed traders become more sensitive to the fad process ($\gamma$ increases), the market maker is penalised because the informed traders send orders that are sharper (or more toxic) and the market maker only uses the bid and the ask displacements to trade with both informed and uninformed traders.

\begin{figure}[H]
\centering
\includegraphics[width=0.5\textwidth]{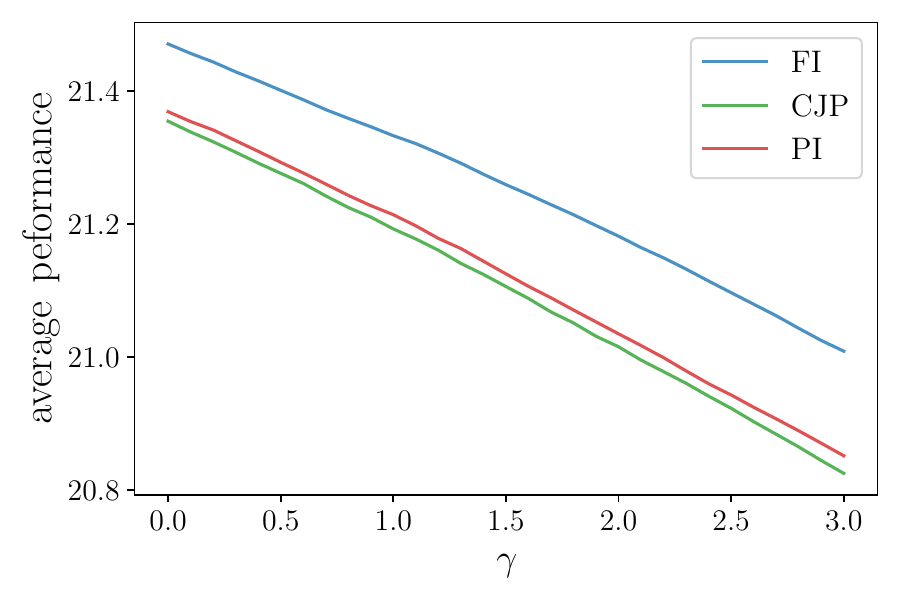}
\caption{Performance of the strategies as a function of $\gamma$. }\label{fig:PnL_gamma} 
\end{figure} 

The argument is similar to the one explaining the effect of a change of $\qconst$. 
As $\gamma$ increases, the market maker becomes more  exposed to ``toxic'' order flow. This is confirmed by  Figure \ref{fig:optimal quotes gamma} where we show the difference, in percentage of the baseline quote $1/k$, between the optimal bid (and ask) displacements for $\gamma = 10$ and $\gamma = 0.1$. 
In agreement with what is shown in Figure \ref{fig:optimal quotes} and  its interpretation, the market maker's sensitivity with respect to $U_t$  becomes stronger as $\gamma$ increases and, as shown above, this leads to a poorer performance. There is an asymmetric effect of the fad on the displacements: when the inventory is zero and $U_t$ is positive, as $\gamma$ increases, the optimal bid (resp.~ask) displacement becomes larger (resp.~smaller) because the arrival rate of informed traders on that side of the market increases (resp.~decreases). The reverse is observed for a negative $U_t$. It can be shown numerically that the bid-ask spread increases as $\gamma$ goes up.

\begin{figure}[H]
\centering
\includegraphics[width=0.5\textwidth]{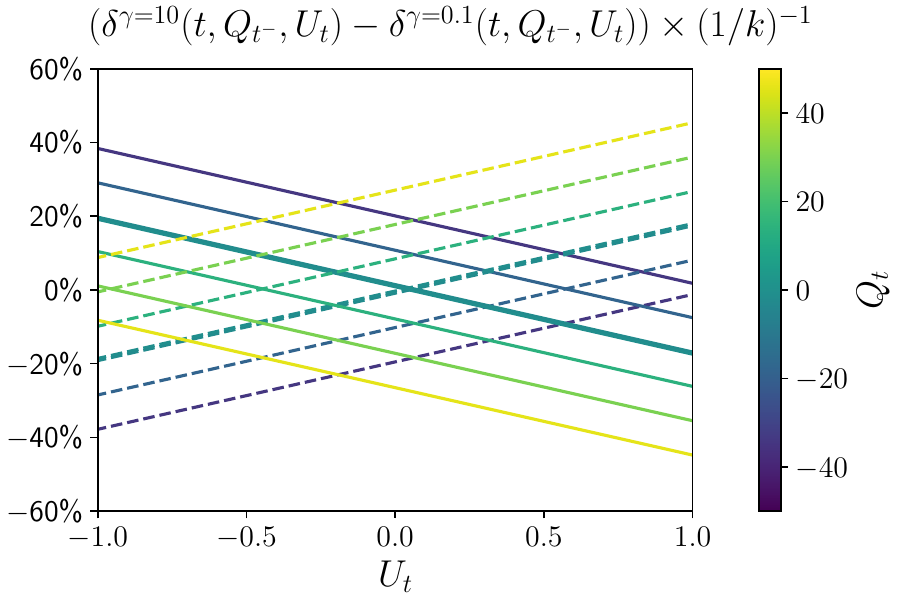}
\caption{Scaled difference in bid and ask displacements for $\gamma=10$ and $\gamma=0.1$ as a function of the fad ($x$-axis) and the inventory of the market maker (colour bar).
Solid lines: scaled difference of optimal ask displacement. Dashed lines: scaled difference of  optimal bid displacement. }\label{fig:optimal quotes gamma} 
\end{figure} 

\textbf{Effect of $\eta$}: The parameter $\eta$ governs the speed of mean reversion of the fad. When $\eta$ increases, the fad reverts more quickly to zero, while when $\eta$ decreases towards zero, the fad becomes pure noise. 
As $\eta$ increases, the performances of the three strategies increase.
The rationale is that as $\eta$ increases, the variance of the fad decreases and it becomes less persistent and, therefore, the market maker is exposed to less risk managing the inventory. This argument holds true for all the three settings. We also observe that the performance gap between the full information and the other two strategies decreases as $\eta$ increases. The reason for this result is that when $\eta$ is high, the fad collapses around zero and, therefore, the
model with no information becomes less misspecified. 
Confirming this interpretation, it can be shown that the bid-ask spread decreases as $\eta$ increases.

\begin{figure}[H]
\centering

\includegraphics[width=0.5\textwidth]{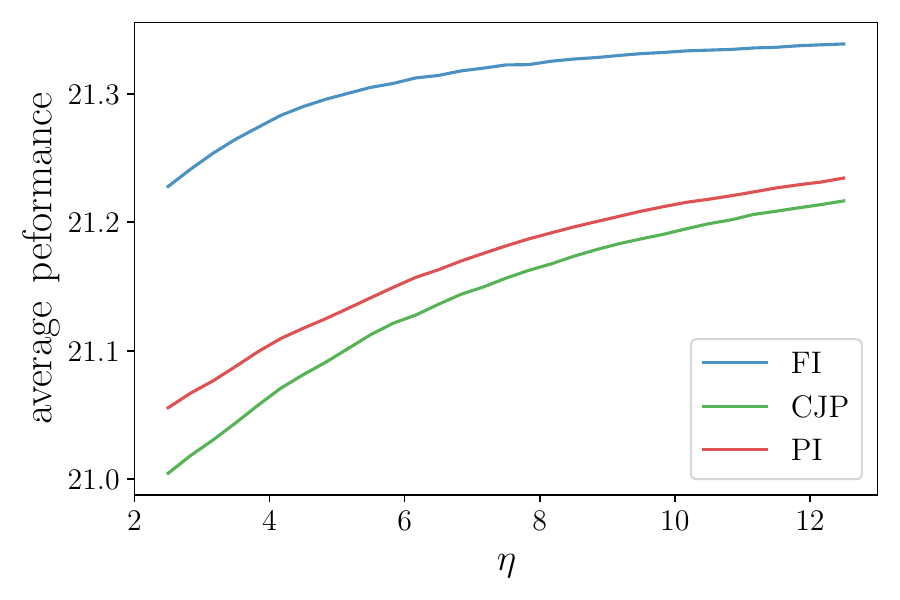}
\caption{Performance of the strategies as a function of $\eta$. }\label{fig:PnL_eta} 
\end{figure}

\textbf{Effect of $\varphi$}: We consider the change in the percentage of informed traders keeping constant the expected number of market arrivals. To this end, given a value for $\varphi\in[0,30]$,
we compute $\psi$ according to \eqref{eq_psi}. The percentage of informed traders is  $(1 - \varphi/30)\times 100 \%$.

In Figure  \ref{fig:spread_psi} we show the quoted bid-ask spread of the market maker and her performance as a function of the percentage of informed traders in the market. The bid-ask spread increases in the fraction of informed traders. This effect is due to classical adverse selection effect going back to \cite{GLOSTEN}: the market maker reacts to the presence of more informed traders by widening the spread.\footnote{Quoting \cite{GLOSTEN}: ``\textit{Generally, ask prices increase and bid prices decrease if the insiders' information becomes better, or the insiders become more numerous relative to liquidity
traders}''.} 
We also observe that the bid-ask spread increases when the running penalty on inventory increases. This is because a higher running penalty makes the market maker more averse to inventory risk.

As far as the performance is concerned, we observe that it decreases with the fraction of informed traders. The root cause of this result is that as the fraction of informed traders increases, the fraction of ''toxic'' trades, e.g., those arriving from informed traders and referring to the fundamental value, goes up. This jeopardises the performance of the market maker.

 \begin{figure}[H]
\centering
\includegraphics[width=0.4\textwidth]{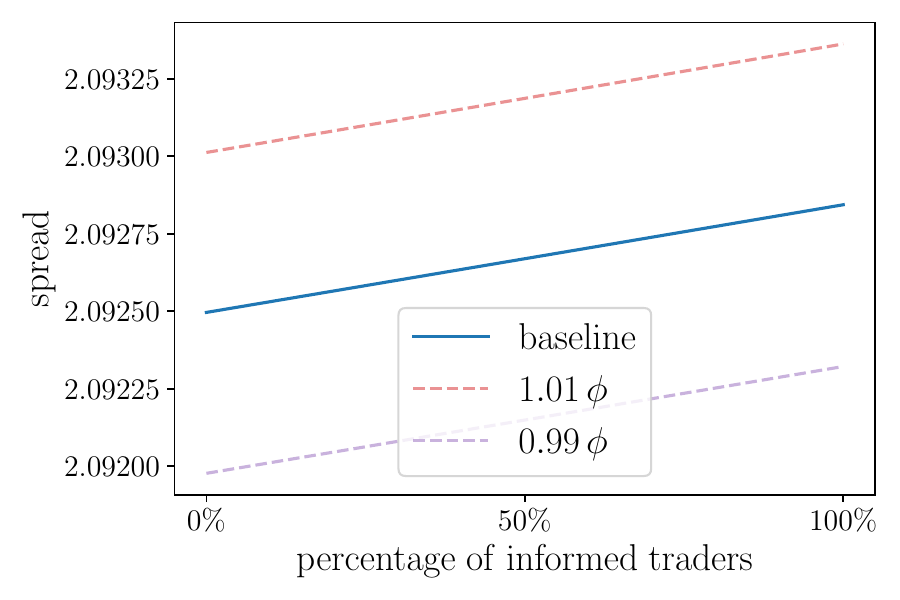}
\includegraphics[width=0.4\textwidth]{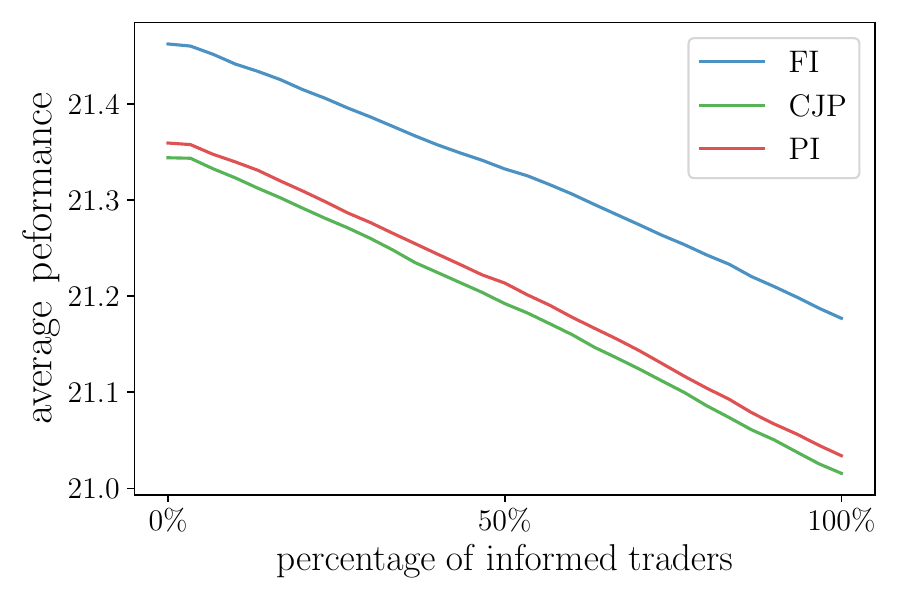}
\caption{Left panel: bid-ask spread of the informed market maker as a function of the percentage of informed traders such that the expected number of market arrivals at $T$ in the full information setting is $30$. The parameter $\phi$ controls the running penalty on inventory. The baseline value is $\phi=0.1$. Right panel: performance of the strategies as a function of the percentage of informed traders.}\label{fig:spread_psi} 
\end{figure}

\section{Calibration and  misspecification of parameters}
\label{CALIB}

Using discrete time observations of the mid price $\process[S]$ and of the market arrivals $\process[M^\pm]$, we are able to obtain estimates of  model parameters $(\qconst,\eta,\varphi,\psi,\gamma)$. More precisely, 
by observing $(S_{t_1},...,S_{t_N})$ we calibrate $(\qconst,\eta)$ using the Kalman--Bucy filter and a maximum likelihood estimator. Once $(\qconst,\eta)$ are calibrated, we construct the filter of the fad $\process[\hat{U}]$ which depends on the calibrated $(\qconst,\eta)$ and use the market arrivals intensities  \eqref{eq:lambdas hat_a}-\eqref{eq:lambdas hat_b}. Then, observing the market arrivals we calibrate $(\psi,\varphi,\gamma)$ using a maximum likelihood estimator. Tables \ref{table_KB_calib} and \ref{table_poisson_calib} report the estimations of the parameters over $1,000$ simulations using these maximization procedures. A detailed derivation of the estimators is provided in Appendix \ref{sec:likelihood_estim}.
    \begin{table}[H]
\centering
\begin{tabular}{l|cc}
\hline
\hline
\noalign{\vskip-1mm}
\noalign{\vskip 2mm}
true parameters & estimation of parameters\\
\noalign{\vskip-1mm}
\noalign{\vskip 2mm}
\hline
\hline
\noalign{\vskip-1mm}
\noalign{\vskip 2mm}
$(\qconst\,,\eta)$ = $(1\,,10)$ &$(0.84\,,12.84)$ \\
$(\qconst\,,\eta)$ = $(0.6\,,10)$ &$(0.62\,,6.28)$ \\
\noalign{\vskip-1mm}
\noalign{\vskip 2mm}
\hline
\hline
\noalign{\vskip-1mm}
\noalign{\vskip 2mm}
$(\qconst\,,\eta)$ = $(0.8\,,10)$ &$(0.74\,,8.52)$ \\
\noalign{\vskip-1mm}
\noalign{\vskip 2mm}
  \hline
  \hline
\noalign{\vskip-1mm}
\noalign{\vskip 2mm}
$(\qconst\,,\eta)$ = $(0.8\,,8)$ &$(0.72\,,7.95)$\\
$(\qconst\,,\eta)$ = $(0.8\,,6)$ &$(0.73\,,7.45)$\\
\hline
\hline
\end{tabular}
  \caption{Calibration of $(\qconst,\eta)$.}
  \label{table_KB_calib}
\end{table}
Table \ref{table_KB_calib} shows the estimation of $\qconst$ and $\eta$ over $1,000$ simulations. The accuracy for $(\eta,\qconst)$ is better when $\qconst$ is large. While the estimate of $\qconst$ remains good, the accuracy for $\eta$ decreases when $\qconst$ decreases. Indeed, a small value of $\qconst$ affects the quality of the filter affecting the quality of the information it contains. 
    \begin{table}[H]
\centering
\begin{tabular}{l|cc}
\hline
\hline
\noalign{\vskip-1mm}
\noalign{\vskip 2mm}
true parameters & estimation of parameters\\
\noalign{\vskip-1mm}
\noalign{\vskip 2mm}
\hline
\hline
\noalign{\vskip-1mm}
\noalign{\vskip 2mm}
$(\varphi\,,\gamma)$ = $(10\,,10)$ &$(9.61\,,10.1)$\\
$(\varphi\,,\gamma)$ = $(5\,,10)$ &$(4.8\,,10.2)$\\
\noalign{\vskip-1mm}
\noalign{\vskip 2mm}
\hline
\hline
\noalign{\vskip-1mm}
\noalign{\vskip 2mm}
$(\varphi\,,\gamma)$ = $(15\,,10)$ & $(14.16\,,10.17)$\\ 
\noalign{\vskip-1mm}
\noalign{\vskip 2mm}
\hline
\hline
\noalign{\vskip-1mm}
\noalign{\vskip 2mm}
$(\varphi\,,\gamma)$ = $(15\,,1)$ &$(6.38\,,1.42)$\\
$(\varphi\,,\gamma)$ = $(15\,,3)$ &$(10.35\,,3.85)$\\
  \hline
  \hline
\end{tabular}
  \caption{Calibration of $(\varphi,\gamma)$.}
  \label{table_poisson_calib}
\end{table}
Next, we estimate $\varphi$ and $\gamma$ over $1,000$ simulations. To isolate the effect of calibrating $\varphi$ and $\gamma$, we focus on the case $\hat{\qconst}=\qconst=1$ and $\hat{\eta}=\eta = 10$; as expected,  if we employ the calibrated parameters from Table \ref{table_KB_calib} the estimation would be worse. Table \ref{table_poisson_calib} shows  that the quality of the filter for $(\varphi,\gamma)$ is high when $\gamma$ is large. Instead, the quality of the calibration drops when $\gamma$ is low. This is because a large value of $\gamma$ increases the degree of influence of the fad in the order arrivals.\footnote{The calibrations use $T= 1$. Increasing the value of $T$ increases the accuracy of the estimations.}

Lastly, to corroborate the quality of the calibration procedure, we investigate the effect of a misspecification made by the market maker on the values of key model parameters.  We limit our analysis to the  full information setting. 
Table \ref{table3} reports the difference in $\% $ between the performance associated with the optimal strategy in the full information setting, and the performance for the same strategy when the parameters are misspecified by $\pm 50 \%$.
For example, if the market maker computes the optimal strategy considering $\qconst=0.3$ rather than the true value ($\qconst=0.6$) then the performance loss is $0.158\%$ (see first value in the third column). 
For the baseline set of model parameters identified at the beginning of this section, we observe that a misspecification of $\qconst$ induces significant changes in profitability, while the strategy is more robust with respect to a misspecification of $\eta$, $\gamma$, or $\varphi$. We can conclude that the key feature for the market maker is how the fad impact the asset price rather than the informed-uninformed composition of the market. 
We remark that these results hold true for the parameters reported at the start of the section. They may change as the reference parameters change.
These results show that the misspecification coming from the calibration of parameters (Tables \ref{table_KB_calib} and \ref{table_poisson_calib}) would induce small changes in performance.

\begin{table}[H]
\centering
\begin{tabular}{l|cc}
\hline
\hline
& \multicolumn{2}{c}{Robustness (in $\%$)}\\
\noalign{\vskip-1mm}
\noalign{\vskip 2mm}
\hline
\hline
\noalign{\vskip-1mm}
\noalign{\vskip 2mm}
baseline & $\text{over-estimation by } 50\%$ & $ \text{under-estimation by } 50\%$\\
\noalign{\vskip-1mm}
\noalign{\vskip 2mm}
\hline
\hline
\noalign{\vskip-1mm}
\noalign{\vskip 2mm}
$\qconst = 0.6$ &$0.160\%^{**}$&$0.158\%^{**}$\\
$ \gamma = 1$&$0.000\%$&$0.002\% $\\
$ \eta = 10$& $0.004\%$ & $0.014\%$\\
$ \varphi = 15$&$0.004\%$&$0.003\% $\\
  \hline
  \hline
\end{tabular}
  \caption{Percentage loss in the performance \eqref{eq: payoff} between the baseline optimal strategy in the full informed setting, and a strategy that misspecifies the parameter by $\pm 50 \%$. The number of simulations is 1,000,000, 
 $``**"$ means significant at $1 \%$. }
  \label{table3}
\end{table}

\section{Future Research Direction}
\label{sec : beyond approx}
\subsection{Beyond first approximation of the intensities}
As a first order approximation, we modelled $\process[{N}^{a,b}]$ as counting processes with intensities 
\begin{align}
    \hat{\lambda}^{a}_t &= (\varphi \, e^{-k \, \delta^{a}_t}  + \psi \, \exp(- k \, \delta_t ^{a} -  \gamma \, (\sigma \, \qconst \, \hat{U}_t) \, \vee \, \mathcal{S}^-)) \, \mathds{1}_{\{Q_{t^-} > \underline{q} \}} \, , \\
      \hat{\lambda}^{b}_t &= (\varphi \, e^{-k \, \delta^{b}_t}  + \psi \, \exp(- k \, \delta_t ^{b} +  \gamma \, (\sigma \, \qconst \, \hat{U}_t) \, \wedge \, \mathcal{S}^+)) \, \mathds{1}_{\{Q_{t^-} < \overline{q} \}} \, .
\end{align}
Here, we construct an exact intensity process $\process[\psi \, \hat{\lambda}^{a,b}]$ for the counting process $\process[{N}^{a,b}]$ in the special case $\mathcal{S}^+ = -\mathcal{S}^- = + \infty$.\footnote{This could be understood as follows: the market maker faces informed traders that know exactly what the fad is, but the market maker does not observe $U_t$. Given  knowledge of the price, the market maker estimates the intensity of the market arrivals with $\mathbb{E}[\psi \, e^{\mp \, \qconst \, \gamma \, U_t} | \mathcal{F}_t^S]$. }

\begin{prop}
Let $\hat{\lambda}_t^a = e^{- k \, \delta_t^a} \,(\varphi + \psi \, \mathbb{E}[e^{- \gamma \, \sigma \, \qconst \, U_t} | \mathcal{F}_t^S]) \, \mathds{1}_{\{Q_t > \underline{q} \}}$. Define $\process[Z]$ as $Z_t := e^{-\gamma \, \sigma \, \qconst \, U_t}$.
    Then, it follows that
    \begin{equation}
        \d Z_t = Z_t \, \Bigg( \eta \,\gamma\, \sigma \, \qconst \, U_t + \frac{(\gamma \,  \sigma \, \qconst)^2}{2} \,  \d t - \gamma \, \sigma \, \qconst \, \d B_t \Bigg) =: h(U_t,Z_t) \, \d t + \sigma(Z_t) \, \d B_t \, .
    \end{equation}
    Define $\hat{Z}_t := \mathbb{E}[Z_t | \mathcal{F}_t^S]$, we have that
    \begin{equation}
         \d \hat{Z}_t = \Bigg(\frac{ (\gamma \,  \sigma \, \qconst)^2}{2} \, \hat{Z}_t + \eta \,\gamma \,  \sigma \, \qconst \, (Q_t + \hat{Z}_t \, \hat{U}_t ) \Bigg) \, \d t - \bigl(  \qconst^2 \, \sigma \, \gamma \, \hat{Z}_t + \eta \, \qconst  \, Q_t \bigl) \, \d \Tilde{I}_t \, ,
    \end{equation}
    where $Q_t = \mathbb{E}[(U_t - \hat{U}_t)(Z_t - \hat{Z}_t) | \mathcal{F}_t^S]$ and $\process[\Tilde{I}]$ is the Brownian Motion defined by the filtering process.
\end{prop}
\begin{proof}
    Consider the process $\process[\Gamma]$ as $\Gamma_t := \hat{Z}_t - 1 - \int_0^t \widehat{h(U_s,Z_s)} \, \d s$. Then, $\process[\Gamma]$ is an $L^2 \text{-bounded } (\mathcal{F}_t^S)$-martingale starting from zero.

    Indeed, consider $0 \leq s \leq t \leq T$, then, using the tower property, we have 
    \begin{equation}
        \mathbb{E}[\Gamma_t | \mathcal{F}_s^S] = \mathbb{E}[Z_t | \mathcal{F}_s^S] - \int_0^s \widehat{h(U_u,Z_u)} \, \d u - \int_s^t \mathbb{E}[h(U_u,Z_u)| \mathcal{F}_s^S] \, \d u - 1\, ,
    \end{equation}
    which implies 
    \begin{equation}
        \mathbb{E}[Z_t | \mathcal{F}_s^S] = \hat{Z}_s + \mathbb{E}\Bigg[\int_s^t h(U_u,Z_u) \, \d u + \sigma(Z_u) \, \d B_u \Big| \mathcal{F}_s^S \Bigg] \, .
    \end{equation}
    Using the tower property, we find 
    \begin{equation}
        \mathbb{E}\Bigg[\int_s^t \sigma(Z_u) \, \d B_u \Big| \mathcal{F}_s^S \Bigg] = \mathbb{E}\Bigg[\mathbb{E}\Bigg[\int_s^t \sigma(Z_u) \, \d B_u \Big| \mathcal{F}_s \Bigg] \Big| \mathcal{F}_s^S \Bigg] \, , 
    \end{equation}
    and since $t \rightarrow \int_0^t \sigma(Z_s) \, \d B_s$ is a $(\mathcal{F}_t)_t$-martingale, we have that
    \begin{equation}
         \mathbb{E}\Bigg[\int_s^t \sigma(Z_u) \, \d B_u \Big| \mathcal{F}_s^S \Bigg] = 0 \, .
    \end{equation}
    Finally, we find 
    \begin{equation}
        \mathbb{E}[\Gamma_t | \mathcal{F}_s^S] = \hat{Z}_s - 1 -\int_0^s \widehat{h(U_u,Z_u)} \, \d u = \Gamma_s \, ,
    \end{equation}
    which means that $\process[\Gamma]$ is an $\mathcal{F}^S$-martingale, with initial value zero. Moreover, by construction, we also have that $\process[\Gamma]$ is bounded in $L^2(\Omega, \mathcal{F}^S, \mathbb{P})$.

    Since $\process[\Gamma]$ satisfies condition of Proposition 2.31 in \cite{bain2009fundamentals}, then, there exists an $\process[\mathcal{F}^S]$-progressively measurable process $\process[\nu]$ such that 
    \begin{equation}
        \Gamma_t = \int_0^t \nu_s \, \d \Tilde{I}_s \, .
    \end{equation}
    Similar calculations to those in Section 4.1 imply that $\nu_t =  \qconst \, \sigma(Z_t) - \eta \, \qconst \, Q_t$, where $Q_t := \mathbb{E}[(U_t - \hat{U}_t)(Z_t - \hat{Z}_t) | \mathcal{F}_t^S]$. Then, we get 
    \begin{equation}
        \d \hat{Z}_t = (\qconst \,\widehat{ \sigma(Z_t)} - \eta \, \qconst \, Q_t) \, \d \Tilde{I}_t + \widehat{h(U_t,Z_t)} \, \d t \, ,
        \end{equation}
    where 
    $$\widehat{h(U_t,Z_t)} = \eta \,\gamma \,  \sigma \, \qconst \, \widehat{Z_t \, U_t} + \frac{(\gamma\,  \sigma \, \qconst)^2}{2} \, \hat{Z}_t = \frac{(\gamma \,  \sigma \, \qconst)^2}{2} \, \hat{Z}_t + \eta \,\gamma \,  \sigma \, \qconst \, (Q_t + \hat{Z}_t \, \hat{U}_t ) \, , $$
    and where 
    $$\widehat{\sigma(Z_t)} = -\gamma\,  \sigma \, \qconst \, \hat{Z}_t \, .$$
    Finally, we have 
    \begin{equation}
        \d \hat{Z}_t = \Bigg(\frac{(\gamma \,  \sigma \, \qconst)^2}{2} \, \hat{Z}_t + \eta \, \gamma \,   \sigma \, \qconst \, (Q_t + \hat{Z}_t \, \hat{U}_t ) \Bigg) \, \d t - \bigl(  \qconst^2 \,\gamma \,  \sigma \, \hat{Z}_t + \eta \, \qconst \, Q_t \bigl) \, \d \Tilde{I}_t \, .
    \end{equation}
\end{proof}
The computation for $\lambda^b$ is similar. Since one cannot find closed-form expressions for the process $\process[Q]$, it would be interesting to study non-linear filtering approximation methods, such as those developed in \cite{ertel2024analysis}.

\subsection{Beyond the Kalman-Bucy filter}\label{sec:Beyond KBf}
Building alternative filters is a promising avenue of future research. Information about the fads $U_t$ can be extracted from prices (as done in this paper), or through arrivals of liquidity taking orders. Recall, that in the full information model the arrivals were modulated by the fads in the asset price. Thus, following \cite{casgrain2019trading}, and similar to \cite{cartea2024automated}, one can discretise $U_t$ and model it as a continuous time Markov chain (CTMC) that we introduce next.

Let us consider that the fad process $\process[U]$ is a continuous-time Markov Chain with generator $L$ taking values in $G = \{ \theta^1, ..., \theta^J\}$.
Instead of modelling the stochastic intensity of filled orders, we keep track of both arrivals and fills as two separate (but related) counting processes. Here,  we filter the fads at the level of market arrivals and we formulate the effect of the control at the level of the fills. The disentanglement between arrivals and fills allows us to employ the results in Theorem 3.1 in \cite{casgrain2019trading} and in \cite{chevalier2024optimal}.

We let $M^{a,b}$ be the market arrivals which happen with stochastic intensity 
\begin{align}
    \lambda_t^{a} &= \sum_{j =1}^J \lambda_t^{a,j}\, \mathds{1}_{\{U_t = \theta^j\}} \, ,\\
    \lambda_t^{b} &= \sum_{j =1}^J \lambda_t^{b,j}\, \mathds{1}_{\{U_t = \theta^j\}} \, ,
\end{align}
where
\begin{equation}
    \lambda_t^{a,j} = \varphi + \psi \, e^{-\gamma \, (\qconst \, \sigma \, \theta^j) \vee \, \mathcal{S}^-} \, 
\end{equation}
\begin{equation}
    \lambda_t^{b,j} = \varphi + \psi \, e^{\gamma \, (\qconst \, \sigma \, \theta^j) \, \wedge \, \mathcal{S}^+)} \, .
\end{equation}
Upon arrivals, the fills probabilities are
\begin{align}
     \exp(-k\,\delta^a_t) \,,\qquad \text{and }\qquad  \exp(-k\,\delta^b_t)\,,
\end{align}
and we keep track of all filled trades in the controlled counting processes  $\process[N^{a,b}]$.

Now, we denote the filtration generated by the Poisson processes $\process[M^{a,b}]$ by  $\process[\mathcal{F}^M]$. As before, the full filtration of all processes involved is $\process[\mcF]$.

We are interested in the process $\process[\pi^j]$ defined by
\begin{equation}
    \pi_t^j := \mathbb{E}[ \mathds{1}_{\{U_t = \theta^j\}} | \mathcal{F}_t^M] \, .
\end{equation}
Conditional on knowing  the dynamic of $\process[\pi^j]$, we would be able to transform our partial-information optimal control problem into a full-information one. 
We have the following result.

\begin{thm}
    For all $j \in [\![1,J]\!]$, the process $\process[\pi^j]$ admits the following representation
    \begin{equation}
        \pi_t^i = \Delta_t^i / \sum_{j=1}^J \Delta_t^j \, ,
    \end{equation}
    where for all $j \in [\![1,J]\!]$, the process $\process[\Delta]$ satisfies
    \begin{equation}
        \d \Delta_t^j = \Delta_{t-}^j(\lambda_{t-}^{a,j} - 1) \, (\d M_t^{a} - \d t) + \Delta_{t-}^j(\lambda_{t-}^{b,j} - 1) \, (\d M_t^{b} - \d t) + \sum_{i=1}^J \Delta_{t-}^i {L}_{i,j} \, \d t \, 
    \end{equation}
\end{thm}
\begin{proof}
    This is direct consequence of Theorem 3.1 in \cite{CASGRAIN}.
\end{proof}

Armed with the  above filter, we formulate the following control problem. 
We write the price dynamics of $\process[S]$ as 
\begin{equation}
   \d S_t = \d P_t^S +\sigma \, \pconst \, \d Z_t= \sum_{j=1}^J  \sigma \, \qconst \, \theta^j \, \d \mathds{1}_{\{U_t = \theta^j\}} +\mu \,  \d t + \sigma \, \pconst \,  \d Z_t \, ,
\end{equation}
and assuming that  the process $\process[Z]$ is independent of $\process[\mathcal{F}^M]$, we have 
\begin{equation}
    \hat{S}_t = \mu \,t + \sum_{j=1}^J  \sigma \, \qconst \, \theta^j \,\pi_t^j  
    \,.
 \end{equation}

\begin{prop}
Let $\hat{\lambda}^{a,b}_t$ be given by
\begin{equation}
    \hat{\lambda}^{a,b}_t = \sum_{j=1}^J \lambda_t^{a,b,j} \, \pi_t^j \, .
\end{equation}
The process $\process[M^{a,b}]$ is an $\process[\mathcal{F}^M]$-Poisson processes with stochastic intensity $\hat{\lambda}^{a,b}_t$.
\end{prop}

\begin{proof}
Let us consider the compensated processes 
\begin{equation}
    D_t^{a,b} := M_t^{a,b} - \int_0^t \lambda_s^{a,b} \, \d s\, .
\end{equation}
We remark that those are $\process[\mathcal{F}]$-martingales. Next,  we consider the processes $\process[\widehat{D}^{a,b}]$ defined by
\begin{equation}
    \widehat{D}^{a,b}_t = D_t^{a,b} + \int_0^t (\lambda_s^{a,b} - \hat{\lambda}_s^{a,b}) \, \d s = M_t^{a,b} - \int_0^t \hat{\lambda}_s^{a,b} \, \d s \, ,
\end{equation}
and we show that it is a $\process[\mathcal{F}^M]$-martingale. Since $\process[M^{a,b}]$ and $\process[\hat{\lambda}^{a,b}]$ are $\process[\mathcal{F}^M]$-adapted, then $\process[\widehat{D}^{a,b}]$ is $\process[\mathcal{F}^M]$-adapted. By the definition of the processes we have that $\process[\widehat{D}^{a,b}] \in L^1(\Omega, \process[\mathcal{F}^M], \mathbb{P})$. Let $0 \leq s \leq t \leq T$, we have 
    $$\mathbb{E}[M_t^{a,b} - \int_0^t \hat{\lambda}^{a,b}_u \d u | \mathcal{F}_s^N] = \widehat{D}_s^{a,b} + \mathbb{E}[D_t^{a,b} - D_s^{a,b} | \mathcal{F}_s^N] - \mathbb{E}[\int_s^t (\lambda_u^{a,b} - \hat{\lambda}_u^{a,b} )\, \d u  | \mathcal{F}_s^M] \, .$$

However, by the tower property and the fact that $\process[D^{a,b}]$ is a martingale, we have $$ \mathbb{E}\bigl[D_t^{a,b} - D_s^{a,b} | \mathcal{F}_s^M \bigl] = \mathbb{E}[ \mathbb{E}[D_t^{a,b} - D_s^{a,b} | \mathcal{F}_s]| \mathcal{F}_s^M] = 0\, .$$ Again, by the tower property, $$\mathbb{E}\Bigg[\int_s^t (\lambda_u^{a,b} - \hat{\lambda}_u^{a,b} )\, \d u  | \mathcal{F}_s^M \Bigg] = \mathbb{E}\Bigg[\mathbb{E}\Bigg[\int_s^t (\lambda_u^{a,b} - \hat{\lambda}_u^{a,b} )\, \d u  | \mathcal{F}_u^M\Bigg]| \mathcal{F}_s^M\Bigg] = 0\, .$$
Finally, we end up with $\mathbb{E}[\widehat{D}^{a,b}_t | \mathcal{F}_s^M] = \widehat{D}^{a,b}_s$, and then $\process[\widehat{D}^{a,b}]$ is a $\process[\mathcal{F}^M]$-martingale. 
    \\
    \newline
    Finally, by the  Watanabe's characterization theorem, $\process[M^{a,b}]$ are $\process[\mathcal{F}^M]$ Poisson processes, with stochastic intensities $\process[\hat{\lambda}^{a,b}]$.
\end{proof} 
Once we have this, we can rewrite the dynamics of the processes as $\process[\mathcal{F}^M]$-predicable processes.
\begin{equation}
    \d \hat{S}_t = \d \hat{P}_t^S = \mu \, \d t + \sum_{j=1}^J  \sigma \, \qconst \, \theta^j \, \d \pi_t^j \, .
\end{equation}
 The processes $\process[M^{a,b}]$ are $\process[\mathcal{F}^M]$-adapted Poisson processes with stochastic intensities $\process[\hat{\lambda}^{a,b}]$ given by 
\begin{equation}
   \hat{\lambda}^{a,b}_t = \sum_{j=1}^J \lambda_t^{a,b,j} \, \pi_t^j \, ,
\end{equation}
and the process $\process[\hat{U}]$ is given by 
\begin{equation}
    \hat{U}_t = \sum_{j=1}^J \theta^j \, \pi_t^j \, .
\end{equation}
Since all the processes are now $\process[\mathcal{F}^M]$-predictable, we are able to formulate the full information optimal control problem similar to that in the previous setting. The optimal control problem is 
\begin{equation}
    V(t,x,q,s, (\Delta^i)_{i \in [\![1,J]\!]}) = \sup_{\delta^{a,b}} \mathbb{E}\left[ \hat{X}_T + {Q}_T \, \hat{S}_T - \alpha \, {Q}_T^2 - \phi \, \int_t^T {Q}_u^2 \, \d u \right] \, , 
\end{equation}
with 
\begin{align}
\begin{cases}
    &\d {S}_t = \mu \, \d t + \sum_{j = 1}^J \sigma \, \qconst \, \theta^j \d \pi_t^j \, , \\
    &\d \hat{U}_t = \sum_{j=1}^J \theta^j \,\d  \pi_t^j \, ,\\
    &\d {Q}_t = \d \Tilde{N}^b_t - \d \Tilde{N}^a_t \,, \\
    &\d \hat{X}_t = (\hat{S}_t + \delta^a_t) \, \d \Tilde{N}^a_t - (\hat{S}_t - \delta^b_t) \, \d \Tilde{N}^b_t \,, \\
    &\pi_t^i = \Delta_t^i / \sum_{j=1}^J \Delta_t^j \, , \\
    &\d \Delta_t^j = \Delta_{t-}^j(\lambda_{t-}^{a,j} - 1) \, (\d M_t^{a} - \d t) + \Delta_{t-}^j(\lambda_{t-}^{b,j} - 1) \, (\d M_t^{b} - \d t) + \sum_{i=1}^J \Delta_{t-}^i {L}_{i,j} \, \d t \, ,
    \end{cases}
\end{align}
where $\process[\Tilde{N}^{a,b}]$ are Poisson process with intensity\footnote{Similarly to the previous cases, we assume that the inventory of the market maker is in $\mathcal{Q}$.} $\left(\hat{\lambda}^{a,b}_t \, e^{-k\, \delta_t^{a,b}} \, \mathds{1}_{\{Q_t \in \mathring{\mathcal{Q}}\}}\right)_{t \in \mfT}$ and where the controls belong to the set of $\process[\mathcal{F}^{M,N}]$-predictable and bounded processes, where $\process[\mathcal{F}^{M,N}]$ is the filtration generated by the Poisson processes $\process[M^{a,b}]$ and $\process[\Tilde{N}^{a,b}]$. We leave the computation of the optimal strategies for future research.

\section{Conclusions}
\label{CONC}
We computed approximate closed-form solutions for a market making problem à la Avellaneda-Stoikov when the market is populated by informed and uninformed traders.\footnote{Here, information is knowledge about the fad in the market.} To the best of our knowledge, this is the first time that informed and uninformed trading is incorporated in the optimal market making problem. 
Our work builds a bridge for the classical results of \cite{GLOSTEN} in the continuous-time optimal market making problem of \cite{AVELL}.
In a market with informed and uninformed traders where the fads are impounded in the asset price, temporary mispricing generates an endogenous asymmetry in optimal liquidity provision by the market maker. The model provides a structural mechanism to link short-term price deviations (short living information) to asymmetric liquidity.

We carried out an extensive sensitivity analysis to find the value of information (by comparing strategies with increasing levels of information) for the market maker. We  showed how the toxicity of the order flow deteriorates the market maker's performance and how it widens the bid-ask spread.

\section*{Acknowledgements:}
We thank \'A.~Cartea for  comments on earlier versions of this draft. We are also grateful to participants at the Oxford-Man Institute's internal seminar, the Mathematical and Computational Finance Seminar at the University of Oxford, and the  Qfinlab Seminar  at the Politecnico di Milano.\\
\noindent For the purpose of open access, the authors have applied a CC BY public copyright licence to any author accepted manuscript arising from this submission.

\newpage
\bibliographystyle{plain}
\bibliography{references}

@article{AVELL,
  title={High-frequency trading in a limit order book},
  author={Avellaneda, Marco and Stoikov, Sasha},
  journal={Quantitative Finance},
  volume={8},
  number={3},
  pages={217--224},
  year={2008},
  publisher={Taylor \& Francis}
}

@article{SUMME,
  title={Does the stock market rationality reflect fundametnal values?},
  author={Summers, Lawrence},
  journal={Journal of Finance},
  volume={41},
  number={},
  pages={591-601},
  year={1986},
  publisher={}
}

@article{FAFRE,
  title={Permanent and temporary components of stock prices},
  author={Fama, Edward and French, Kenneth},
  journal={Journal of Political Economy},
  volume={96},
  number={},
  pages={246-273},
  year={1988},
  publisher={University of Chicago Press}
}

@article{GLOSTEN,
author={Glosten, Lawrence R and Milgrom, Paul R},
title = {Bid, ask and transaction prices in a specialist market with heterogeneously informed traders},
journal = {Journal of Financial Economics},
volume = {14},
number = {1},
pages = {71-100},
year = {1985},
}

@article{baldacci2023mean,
  title={A mean-field game of market-making against strategic traders},
  author={Baldacci, Bastien and Bergault, Philippe and Possama{\"\i}, Dylan},
  journal={SIAM Journal on Financial Mathematics},
  volume={14},
  number={4},
  pages={1080--1112},
  year={2023},
  publisher={SIAM}
}

@article{CONT,
author={Cont, Rama and Kukanov, A. and Stoikov, S.},
title = {The price impact of order book events},
journal = {Journal of Financial Economics},
volume = {12},
number = {1},
pages = {47-88},
year = {2014},
}

@article{ALM,
title = {Optimal execution of portfolio transactions},
  author={Almgren, Robert and Chriss, Neil},
journal = {Journal of Risk},
volume = {3},
number = {1},
pages = {5-39},
year = {2000},
}

@article{CART-JAI,
  title={Incorporating order-flow into optimal execution},
  author={Cartea, \'Alvaro and Jaimungal, Sebastian},
  journal={Mathematical Financial Economics},
  volume={10},
  number={},
  pages={339-364},
  year={2016},
  publisher={Springer}
}

@article{BANK,
  title={Optimal investment with a noisy signal of future stock prices},
  author={Bank, Peter and Dolinsky, Yan},
  journal={Applied Mathematics \&Optimization},
  volume={89},
  number={35},
  pages={1-35},
  year={2024},
  publisher={Springer}
}

@book{bensoussan2018estimation,
  title={Estimation and control of dynamical systems},
  author={Bensoussan, Alain},
  volume={48},
  year={2018},
  publisher={Springer}
}

@article{CART-SAN,
  title={Brokers and informed traders: dealing with toxic flow and extracting trading signals},
  author={Cartea, {\'A}lvaro and S{\'a}nchez-Betancourt, Leandro},
  journal={SIAM Journal on Financial Mathematics},
  volume={16},
  number={2},
  pages={243--270},
  year={2025},
  publisher={SIAM}
}

@article{CAMP,
  title={Smart money, noise trading and stock price behaviour},
  author={Campbell, John Y and Kyle, Albert S},
  journal={The Review of Economic Studies},
  volume={60},
  number={1},
  pages={1--34},
  year={1993},
  publisher={Wiley-Blackwell}
}

@book{gueant2016financial,
  title={The Financial Mathematics of Market Liquidity: From optimal execution to market making},
  author={Gu{\'e}ant, Olivier},
  volume={33},
  year={2016},
  publisher={CRC Press}
}

@article{CDB,
  title={Detecting Toxic Flow},
  author={Cartea, \'Alvaro and Duran-Martin, Gerardo and S\'anchez-Betancourt},
  journal={arXiv preprint arXiv:2312:05827},
  year={2023}
}

@article{ELH,
  title={Flow Toxicity and Liquidity in a High-frequency World},
  author={Easley,David and López de Prado, Marcos and O'Hara, Maureen},
  journal={The Review of Financial Studies},
  volume={25},
  number={5},
  pages={1457-1493},
  year={2012},
  publisher={Oxford University Press}
}

@article{cartea2020market,
  title={Market making with alpha signals},
  author={Cartea, {\'A}lvaro and Wang, Yixuan},
  journal={International Journal of Theoretical and Applied Finance},
  volume={23},
  number={03},
  pages={2050016},
  year={2020},
  publisher={World Scientific}
}

@book{CART-book,
  title={Algorithmic and High-frequency Trading},
  author={Cartea, {\'A}lvaro and Jaimungal, Sebastian and Penalva, Jos{\'e}},
  year={2015},
  publisher={Cambridge University Press}
}

@article{JUSS,
  title={Optimal market making with persistent order flow},
  author={Jusselin, Paul},
  journal={SIAM Journal on Financial Mathematics},
  volume={12},
  number={3},
  pages={1150--1200},
  year={2021},
  publisher={SIAM}
}

@article{NEU,
  title={Optimal signal-adaptive trading with temporary and transient price-impact},
  author={Neuman, Eyal and Voss, Moritz},
  journal={SIAM Journal on Financial Mathematics},
  volume={13},
  number={2},
  pages={551-575},
  year={2022},
  publisher={SIAM}
}

@book{SHLEI,
    author = {Shleifer, Andrei},
    title = "{Inefficient Markets: An Introduction to Behavioral Finance}",
    publisher = {Oxford University Press},
    year = {2000},
 }

@book{FPR,
    author = {Foucault, Thierry and Pagano, Marco and Ailsa Roell},
    title = "{Market Liquidity}",
    publisher = {Oxford University Press},
    year = {2013},
 }

@book{GEN_SHLEI,
    author = {Gennaioli, Nicola and Shleifer, Andrei},
    title = "{A Crisis of Beliefs: Investor Psychology and Financial Fragility}",
    publisher = {Princeton University Press},
    year = {2018},
 }

@article{LEHA,
  title={Incorporating singals into optimal trading},
  author={Lehalle, Charles-Albert and Neuman, Eyal},
  journal={Finance Stochastics},
  volume={23},
  number={},
  pages={275-311},
  year={2019},
  publisher={Springer}
}

@article{cartea2024automated,
  title={Strategic Bonding Curves in Automated Market Makers
},
  author={Cartea, {\'A}lvaro and Drissi, Fay{\c{c}}al and S{\'a}nchez-Betancourt, Leandro and Siska, David and Szpruch, Lukasz},
  journal={Available at SSRN 5018420},
  year={2024}
}

@article{BAN_CART,
  title={Optimal execution and speculation with trade signals},
  author={Bank, Peter and Cartea, \'Alvaro and Korber, Laura},
  journal={arXiv:2306.00621},
  volume={},
  number={},
  pages={},
  year={2023},
  publisher={}
}

@article{GUA,
  title={Asymmetric information in fads models},
  author={Guasoni, Paolo},
  journal={Finance and Stochastics},
  volume={10},
  number={2},
  pages={159--177},
  year={2006},
  publisher={Springer}
}

@article{GUE,
  title={Optimal portfolio liquidation with limit orders},
  author={Gu{\'e}ant, Olivier and Lehalle, Charles-Albert and Fernandez-Tapia, Joaquin},
  journal={SIAM Journal on Financial Mathematics},
  volume={3},
  number={1},
  pages={740--764},
  year={2012},
  publisher={SIAM}
}

@article{GUEANT2013,
  title={Dealing with the inventory risk: a solution to the market making problem},
  author={Gu{\'e}ant, Olivier and Lehalle, Charles-Albert and Fernandez-Tapia, Joaquin},
  journal={Mathematics and financial economics},
  volume={7},
  pages={477--507},
  year={2013},
  publisher={Springer}
}

@article{casgrain2019trading,
  title={Trading algorithms with learning in latent alpha models},
  author={Casgrain, Philippe and Jaimungal, Sebastian},
  journal={Mathematical Finance},
  volume={29},
  number={3},
  pages={735--772},
  year={2019},
  publisher={Wiley Online Library}
}

@article{drissi2022solvability,
  title={Solvability of differential Riccati equations and applications to algorithmic trading with signals},
  author={Drissi, Fay{\c{c}}al},
  journal={Applied Mathematical Finance},
  volume={29},
  number={6},
  pages={457--493},
  year={2022},
  publisher={Taylor \& Francis}
}

@article{knochenhauer2024continuous,
  title={Continuous-Time Dynamic Decision Making with Costly Information},
  author={Knochenhauer, Christoph and Merkel, Alexander and Zhang, Yufei},
  journal={arXiv preprint arXiv:2408.09693},
  year={2024}
}

@article{HOSTO,
  title={Optimal dealer pricing under transactions and return uncertainty},
  author={Ho, Thomas and Stoll, Hans R},
  journal={Journal of Financial Economics},
  volume={9},
  number={1},
  pages={47--73},
  year={1981},
  publisher={Elsevier}
}

@book{bain2009fundamentals,
  title={Fundamentals of stochastic filtering},
  author={Bain, Alan and Crisan, Dan},
  volume={3},
  year={2009},
  publisher={Springer}
}

@article{BARZY,
  title={Unwinding Toxic Flow with Partial Information},
  author={Barzykin, Alexander and Boyce, Robert and Neuman, Eyal},
  journal={arXiv preprint arXiv:2407.04510},
  year={2024}
}

@article{bergault2018closed,
  title={Closed-form approximations in multi-asset market making},
  author={Bergault, Philippe and Evangelista, David and Gu{\'e}ant, Olivier and Vieira, Douglas},
  journal={arXiv preprint arXiv:1810.04383},
  year={2018}
}

@article{boyce2024market,
author = {Boyce, Robert and Herdegen, Martin and S\'{a}nchez-Betancourt, Leandro},
title = {Market Making with Exogenous Competition},
journal = {SIAM Journal on Financial Mathematics},
volume = {16},
number = {2},
pages = {692-706},
year = {2025}
}

@article{CASGRAIN,
  title={Trading algorithms with learning in latent alpha models},
  author={Casgrain, Philippe and Jaimungal, Sebastian},
  journal={Mathematical Finance},
  volume={29},
  number={3},
  pages={735--772},
  year={2019},
  publisher={Wiley Online Library}
}

@article{sun2023stochastic,
  title={Stochastic Linear-Quadratic Optimal Control with Partial Observation},
  author={Sun, Jingrui and Xiong, Jie},
  journal={SIAM Journal on Control and Optimization},
  volume={61},
  number={3},
  pages={1231--1247},
  year={2023},
  publisher={SIAM}
}

@article{ertel2024analysis,
  title={Analysis of the ensemble Kalman--Bucy filter for correlated observation noise},
  author={Ertel, Sebastian W and Stannat, Wilhelm},
  journal={The Annals of Applied Probability},
  volume={34},
  number={1B},
  pages={1072--1107},
  year={2024},
  publisher={Institute of Mathematical Statistics}
}

@book{touzi2012optimal,
  title={Optimal stochastic control, stochastic target problems, and backward SDE},
  author={Touzi, Nizar},
  volume={29},
  year={2012},
  publisher={Springer Science \& Business Media}
}

@article{chevalier2024optimal,
  title={Optimal Execution under Incomplete Information},
  author={Chevalier, Etienne and Hafsi, Yadh and Vath, Vathana Ly},
  journal={arXiv preprint arXiv:2411.04616},
  year={2024}
}

@article{euch2021optimal,
  title={Optimal make--take fees for market making regulation},
  author={Euch, Omar El and Mastrolia, Thibaut and Rosenbaum, Mathieu and Touzi, Nizar},
  journal={Mathematical Finance},
  volume={31},
  number={1},
  pages={109--148},
  year={2021},
  publisher={Wiley Online Library}
}

@book{shumway2006time,
  title={Time series analysis and its applications: with R examples},
  author={Shumway, Robert H and Stoffer, David S},
  year={2006},
  publisher={Springer}
}

\newpage
\appendix
\section{Appendix}
\label{appendix}
\subsection{Gaussian processes properties}
Let us present few results related to the Kalman-Bucy filter.
\begin{prop}
    If we define the process $\process[\mathbb{X}]$ as 
    \begin{equation}
    \label{eq : KF}
        \d \mathbb{X}_t = (\mu - \mathbb{A} \, \mathbb{X}_t) \, \d t + \Sigma_t \, \d \mathbb{W}_t \, , \qquad \mathbb{X}_0 \in \mathbb{R}^d \, ,
    \end{equation}
    where $\process[\mathbb{W}]$ is a p-dimensional Brownian motion, $\process[\mathbb{X}] \in \mathbb{R}^d$, $\Sigma \in C^0( \mathbb{R} ; \mathbb{R}^{d \times p}), \mathbb{A} \in \mathbb{R}^{d \times d}$, and $\mu \in \mathbb{R}^d$, and such that 
    $$
     \int_0^t \lVert e^{-\mathbb{A} \, (t-s)} \, \Sigma_s \rVert^2 \, \d s < \infty \, ,
    $$ then
    \begin{equation}
        \mathbb{X}_t = \int_0^t e^{-\mathbb{A} \, (t-s)} \, \Sigma_s \, \d \mathbb{W}_s +e^{-\mathbb{A} \, t} \, \mathbb{X}_0 + \int_0^t e^{-\mathbb{A} \, (t-s)} \, \mu \, \d s \, .
    \end{equation}
\end{prop}
Then the process $\process[\mathbb{X}]$ is a Gaussian process. For instance, if we consider two one-dimensional processes $\process[X]$ and $\process[Y]$ such that their joint dynamic is given by (\ref{eq : KF}), then they are jointly Gaussian. Consequently, and using Gaussian property, it holds that the conditional expectation of the process $\process[X]$ given the natural filtration generated by $\process[Y]$ is also Gaussian. Finally, we have the following corollary. 
\begin{corollary}
The process $\process[\Tilde{X}]$ defined by
    \begin{equation}
        \Tilde{X}_t := X_t - \mathbb{E}[X_t | \sigma(Y_u)_{u \in [0,t]}] \, ,
    \end{equation}
    is also a Gaussian process, independent of $\sigma(Y)$.
\end{corollary}

\subsection{Proof of Theorem \ref{thm:filter}}
\begin{proof}
The above Kalman-Bucy filtering equations are derived  as follows. 

Since $\mathbb{E}\bigl[\int_0^t \lvert \, \pi_s (h) \,  \d s \,  \rvert ^2 \bigl] < \infty$ for all $t >0$, we can use  Proposition \ref{marting:rpz}, and given that $\process[\Gamma] \in L^2(\Omega, \mathcal{F}, \mathbb{P})$, we have that
    $$\Gamma_t = \hat{U}_t + \eta \int_0^t \hat{U}_s \, \d s = \int_0^t \nu_s \,\sigma^{-1} \, \d I_s\,,$$ 
    where our goal is to find $\nu$. 
    Let $\process[\xi]$ be an $\process[\mcF^S]$ progressively measurable process in $L^2(\Omega, \mathcal{F},\mathbb{P})$, and  define $\eta_t := \int_0^t \sigma^{-1} \, \xi_s \,\d I_s$. Since $\process[\Tilde{I}]$ is a Brownian Motion and $\process[\xi]$ is square-integrable, we have that $\process[\eta]$ is a martingale, and we get 
    \begin{equation}
        \mathbb{E}[\Gamma_t \,\eta_t ] = \mathbb{E}\bigl[\int_0^t  \nu_s \,\xi_s\, \d s \bigl] = \mathbb{E}[\hat{U}_t \eta_t + \eta \,\int_0^t \hat{U}_s \,\eta_t \, \d s]\,.
    \end{equation}
    Using the martingale property of $\process[\eta]$, we get that 
    \begin{equation}
        \mathbb{E}[\hat{U}_s \,\eta_t] =  \mathbb{E}[\hat{U}_s\,\eta_s] \,.
    \end{equation}
    Hence, we get $\mathbb{E}[\Gamma_t \, \eta_t ] = \mathbb{E}[\hat{U}_t \, \eta_t + \eta \, \int_0^t \hat{U}_s \, \eta_s \, \d s]$\,.
    Moreover, we have $$\d \eta_t =  \sigma^{-1} \, \xi_t(\d Y_t - \pi_t(h) \, \d t) = - \sigma^{-1} \,  \xi_t(\sigma \, \qconst \, \eta \, \Tilde{U}_t \, \d t - \sigma \, \d  \Bar{W}_t)\, ,$$
    where $\Tilde{U}_t := U_t - \hat{U}_t$. \\
    \newline
    Then, using integration by parts, we find that
    \begin{equation}
        \mathbb{E}[U_t \, \eta_t] = \mathbb{E}[ \int_0^t - \eta \, \eta_s \, U_s \d s - \xi_s \, U_s \,  \qconst \, \eta \Tilde{U}_s \, \d s + \xi_s \, \qconst \,  \d s]\, .
    \end{equation}
    Now, using  that $U_t = \Tilde{U}_t + \hat{U}_t$ and that $\Tilde{U}_t$ is independent\footnote{See Appendix \ref{appendix} for the details.} of $\mathcal{F}_t^S$ we find that
    \begin{equation}
        \mathbb{E}[\int_0^t \xi_s \, U_s \, \Tilde{U}_s \, \qconst \, \eta \, \d s] = \mathbb{E}[\int_0^t \xi_s \, \hat{P}_s \,  \qconst \, \eta \, \d s]\, .
    \end{equation}
    Thus, on the one hand 
    \begin{equation}
        \mathbb{E}[\hat{U}_t \, \eta_t] = - \eta \, \mathbb{E}[\int_0^t \eta_s \,  U_s\, \d s] - \mathbb{E}[\int_0^t \xi_s \, \hat{P}_s \, \eta \,  \qconst \, \d s] + \mathbb{E}[\int_0^t \xi_s \,   \qconst \, \d s]\, ,
    \end{equation}
    and on the other hand 
    \begin{equation}
            \mathbb{E}[\hat{U}_t \, \eta_t] = \mathbb{E}[\int_0^t  \xi_s \, \nu_s  \, \d s] - \eta \, \mathbb{E}[\int_0^t \hat{U}_s \, \eta_s \, \d s]\, .
    \end{equation}
    Since $\mathbb{E}[\hat{U}_s \, \eta_s] = \mathbb{E}[U_s \, \eta_s]$, and given that $\process[\xi]$ is arbitrary, we get that 
    \begin{equation}
        \nu_t = - \hat{P}_t \, \eta \, \qconst + \qconst\, .
    \end{equation}
Now we can rewrite $\Gamma_t$ as 
\begin{equation}
    \Gamma_t = \int_0^t  - \hat{P}_s \, \sigma^{-1} \, \qconst \, \eta \, \d I_s + \sigma^{-1} \, \qconst \,  I_t = \hat{U}_t + \eta \,  \int_0^t \hat{U}_s \, \d s \, ,
\end{equation}
and compute the dynamics of $\hat{U}_t $ as
\begin{equation}
    \d \hat{U}_t =- \eta \, \hat{U}_t \, \d t  - \sigma^{-1} \, (  \, \hat{P}_t  \,  \qconst \, \eta - \qconst) \, \d I_t \, . 
\end{equation}
Next, we compute the dynamics of $\hat{P}_t$.
    Since $$\d U_t = - \eta \, U_t \, \d t + \d B_t \, ,$$ we have that 
\begin{align*}  
    \d \Tilde{U}_t &= - \eta \, \Tilde{U}_t \, \d t + \d B_t - (\hat{P}_t \,  \eta \, \qconst - \qconst )\,  \qconst \, \eta \,  \Tilde{U}_t \, \d t + (\hat{P}_t \,  \eta \, \qconst - \qconst) \,  \d \Bar{W}_t \\
    &=  - \eta \, \Tilde{U}_t \, (1 +   \qconst \, (\hat{P}_t  \, \eta \, \qconst - \qconst)\, \d t + (1 + (\hat{P}_t \, \eta \, \qconst - \qconst ) \, \qconst ) \, \d B_t + (\hat{P}_t \,  \eta \,  \qconst - \qconst )\, \pconst \, \d Z_t \, .
\end{align*}
    Now, we compute $\frac{\d}{\d t}\mathbb{E}[\Tilde{U}_t^{ 2}] = \frac{\d\hat{P}_t}{\d t}$ and we use that $\mathbb{E}[\Tilde{U}_0^{ 2}] = 0$. We find that
    \begin{align*}
        \hat{P}_t &=  \mathbb{E}[\int_0^t - 2 \,  \eta \, \hat{P}_s \, (1 +   \qconst \, (\hat{P}_s \,  \eta \, \qconst - \qconst))\, \d s] \\
        &\qquad\qquad\qquad+ \mathbb{E}[\int_0^t (1 + (\hat{P}_s \,   \eta \, \qconst - \qconst)\, \qconst  )^2 \, \d s] + \mathbb{E}[\int_0^t ((\hat{P}_s  \,  \eta \, \qconst - \qconst) \,   \pconst)^2 \, \d s \bigl] \\
        &=- 2 \int_0^t \,  \eta \, \hat{P}_s \, (1 +  \qconst \, (\hat{P}_s \, \eta \, \qconst - \qconst ))\, \d s \\
        &\qquad\qquad\qquad+\int_0^t (1 + (\hat{P}_s \, \eta \, \qconst - \qconst )\,  \qconst)^2 \, \d s + \int_0^t ((\hat{P}_s \,  \eta \, \qconst - \qconst) \, \pconst)^2 \, \d s \, .
    \end{align*}
    Hence, we get that $\hat{P}_t$ solves the following Riccati equation 
    \begin{equation}
        \frac{\d \hat{P}_t}{\d t} =-2 \, \eta \, \hat{P}_t \, (1 +   \qconst \, (\hat{P}_t  \, \eta \, \qconst - \qconst )) + (1 + (\hat{P}_t  \, \eta \, \qconst - \qconst)  \,  \qconst)^2 +((\hat{P}_t \,  \eta \, \qconst - \qconst)  \pconst)^2 \, .
    \end{equation}
    Existence and uniqueness of the Riccati equation follows by reverting the time and Lemma 2.1 in \cite{sun2023stochastic}.
\end{proof}

\subsection{CJP-strategy}
Under the CJP-strategy we consider a model that follows from Chapter 10 of \cite{CART-book}, where we consider an uninformed market maker that does not model  the market arrivals and the mid-price with fads.\footnote{That is, she considers $\qconst = 0$.} According to her beliefs, the mid-price $\process[S]$ follows 
\begin{equation}
    \d S_t = \mu \, \d t + \sigma \, \d B_t \, ,
\end{equation}
and his inventory $\process[Q]$ and wealth $\process[X]$ are defined through Poisson processes $\process[N^{a,b}]$ with intensity 
\begin{equation}
\label{eq:MA_CJP}
    \lambda_t^a =\kappa \, \, e^{- k \, \delta^a_t} \, , \qquad \, \lambda_t^b = \kappa \, e^{- k \, \delta^b_t} \, ,
\end{equation}
and where $\kappa \in \mathbb{R}^+$ is such that the expected number of market arrivals at $T$ in this setting is equal to the expected number of arrivals at $T$ in the full information setting.\footnote{We take $\kappa = \varphi + \frac{\psi}{T} \, \int_0^T \mathbb{E}[e^{\pm \gamma \, \sigma \, \qconst \, U_t}] \, \d t$.}

Considering this model, the market maker computes her optimal quotes $(\delta^{*,a}_t, \delta^{*,b}_t)_{t \in \mfT}$ as solution of the control problem \eqref{eq:VF_1}.

\subsection{Limit cases for $\qconst \in \{0,1\}$}
In this subsection, we present two limit results linking the PnL of our strategies when $\qconst \in \{0,1\}$. To simplify the calculations, we only consider the cases $\delta_\infty = \infty$, $\mathcal{Q} = \mathbb{Z}$ and $\mathcal{S}^+ = - \mathcal{S}^- = + \infty$. 

Let us start with $\qconst = 0$. In this case, we have the following property:
\begin{prop}
\label{prop_q_0}
    When $\qconst = 0$, the PnL of the full information optimal strategy is equal to  the PnL of the partial information  optimal strategy and the CJP-strategy.
\end{prop}
\begin{proof}
    Consider $\qconst = 0$ in the HJB equation of the full informed problem, we obtain 
    \begin{equation}
        0 = \partial_t V - \eta \, u \, \partial_u V - \phi \, q^2 + \frac{1}{2} \, \partial^2_{u,u} V  + q \, \mu + \frac{e^{-1}}{k} \, e^{k \, \Delta^-V} \, (\varphi + \psi ) \, ,
    \end{equation}
    where the processes satisfy 
    \begin{align}
        \begin{cases}
            &\d U_t = - \eta \, U_t \, \d t + \d B_t \, , \\
            &\d S_t = \mu \, \d t + \sigma \, d \, Z_t \, ,\\
            & \d Q_t = \d N_t^b - \d N_t^a \, ,
        \end{cases}
    \end{align}
    and where $\process[N^{a,b}]$ are Poisson processes with stochastic intensity $$\lambda^{a,b}_t = (\varphi + \psi) \, e^{-k \, \delta^{a,b}_t} \, .$$
    On the other hand, the HJB equation satisfied by in the partial informed setting is given by 
    \begin{equation}
        0 = \partial_t V - \phi \, q^2 + q \, \mu - \eta \, u \, \partial_u V + \frac{e^{-1}}{k} \, (\varphi + \psi) \, (e^{k \, \Delta^- V} + e^{k \, \Delta^+ V }) \, ,
    \end{equation}
    where the processes satisfy 
    \begin{align}
        \begin{cases}
            &\hat{U}_t = 0 \, , \\
            &\d S_t =  \mu \, \d t + \sigma \, d \, Z_t \, ,\\
              &\d {Q}_t = \d {N}_t^b - \d {N}_t^a \, ,
        \end{cases}
    \end{align}
      and where $\process[{N}^{a,b}]$ are Poisson processes with stochastic intensity $${\lambda}^{a,b}_t = (\varphi + \psi) \, e^{-k \, \delta^{a,b}_t} \, .$$
      The equality $\hat{U}_t = 0$ comes from the independence of the two Brownian Motions $\process[B]$ and $\process[Z]$. 

      We take $V$ to be independent of $u$, and we obtain
      \begin{equation}
            0 = \partial_t V - \phi \, q^2 + q \, \mu + \frac{e^{-1}}{k} \, (\varphi + \psi) \, (e^{k \, \Delta^- V} + e^{k \, \Delta^+ V }) \, ,
      \end{equation}
      that can be solved following the ideas developed in Chapter 10 of \cite{CART-book}. This is exactly the CJP-strategy.\footnote{Indeed, in this case $\kappa = \varphi + \psi$.} Since the common value function is independent of $u$ and since the optimal quotes are given by $\delta^{a,b}_t = \frac{1}{k} - \Delta^{\pm} V(t,q)$, then, neither the value of $u$ nor the value of $s$ impacts the quotes. Finally, and following our previous notations, we have that an equality of the PnL following the strategies $\delta^{*,\text{FI}}(t,Q_{t-},U_t)$, $\delta^{*,\text{CJP}}(t,Q_{t-})$, $\delta^{*,\text{PI}}(t,Q_{t-},\hat{U}_t)$.
\end{proof}
Now, we consider $\qconst = 1$. In this case, we have the following property:
\begin{prop}
\label{prop_q_1}
    When $\qconst = 1$, the PnL of the optimal strategy in the full information setting and in the partial information setting are the same.
\end{prop}
\begin{proof}
    Consider $\qconst = 1$, then the HJB equation in the full informed setting becomes
    \begin{equation}
    \begin{aligned}
        0 =& \partial_t V - \eta \, u \, \partial_u V - \phi \, q^2 + \frac{1}{2} \, \partial^2_{u,u} V + (\mu - \eta \, \sigma \, u) \, q \\& \qquad \, \qquad+ \frac{e^{-1}}{k} \,e^{k \, \Delta^- V} (\varphi + \psi \, e^{- \gamma \, \sigma \, u}) + \frac{e^{-1}}{k} \,e^{k \, \Delta^- V} (\varphi + \psi \, e^{\gamma \, \sigma \, u}) \, ,
        \end{aligned}
    \end{equation}
    where the processes satisfy 
    \begin{align}
        \begin{cases}
            & \d U_t = - \eta \, U_t \, \d t + \d B_t \, , \\
            & \d S_t = (\mu - \eta \, \sigma \, U_t ) \, \d t + \sigma \, \d B_t \, , \\
            &\d Q_t = \d N_t^b - \d N_t^a \, ,
        \end{cases}
    \end{align}
   and where $\process[N^{a,b}]$ are Poisson processes with intensity 
    $$\lambda^{a,b}_t = (\varphi + \psi \, e^{\mp \gamma \, \sigma \, U_t}) \, e^{- k \, \delta_t^{a,b}} \, .$$
    The HJB equation satisfied in the partial informed setting is given by 
    \begin{equation}
    \begin{aligned}
                0 =& \partial_t V - \eta \, u \, \partial_u V - \phi \, q^2 + \frac{1}{2} \, \partial^2_{u,u} V + (\mu - \eta \, \sigma \, u) \, q \\& \qquad \, \qquad + \frac{e^{-1}}{k} \,e^{k \, \Delta^- V} (\varphi + \psi \, e^{-\gamma \, \sigma \, u}) + \frac{e^{-1}}{k} \,e^{k \, \Delta^- V} (\varphi + \psi \, e^{ \gamma \, \sigma \, u}) \, ,
    \end{aligned}
    \end{equation}
    since we have $\hat{P}_t = 0$ in this setting. And the processes satisfy 
    \begin{align}
        \begin{cases}
            & \hat{U}_t = U_t \, , \\
              & \d S_t = (\mu - \eta \, \sigma \, U_t ) \, \d t + \sigma \, \d B_t \, , \\
            &\d Q_t = \d N_t^b - \d N_t^a \, ,
        \end{cases}
    \end{align}
      where $\process[N^{a,b}]$ are Poisson processes with intensity 
    $$\lambda^{a,b}_t = (\varphi + \psi \, e^{\mp \gamma \, \sigma \, U_t}) \, e^{- k \, \delta_t^{a,b}} \, .$$
     We then have $V^{\text{FI}} = V^{\text{PI}}$, and since the dynamics of their processes are the same, we have equality of the PnL.
\end{proof}

\subsection{Evaluation at the fundamental price}
\label{sec:fundamental}
In our analysis, we assumed that the residual inventory at time $T$ is evaluated at the mid-price. In what follows, we assume that the market maker evaluates the final inventory at the fundamental price $(S_T- \sigma \, \qconst \, U_T)$ and therefore the performance criterion to be considered in the maximization problem becomes:
\begin{equation}
\label{eq: payoff2}
X_T + Q_T\, (S_T- \sigma \, \qconst \, U_T) - \alpha \, Q_T^2 - \phi \, \int_0^T Q_u^2 \, \d u 
\end{equation}
instead of \eqref{eq: payoff}.

Table \ref{table2}  reports the mean (and standard deviation) of the performance of the optimal strategies  when the mark-to-market of the terminal inventory is computed at the fundamental value (100,000 simulations).  Here, the optimal strategies are those computed when the performance criterion is that in \eqref{eq: payoff2}  following the same steps as above.

The main insights we collected from Table  \ref{table1}  are confirmed in Table \ref{table2}. Notice that evaluating the final inventory at the fundamental leads to a very small increase in the performance of the market. The full information strategy outperforms the partial information strategy, which in turn outperforms the CJP strategy. The effect associated with a stress of key model parameters is also similar.

\begin{table}[H]
\centering
\begin{tabular}{l|ccc}
\hline
\hline
& \multicolumn{3}{c}{Performance criterion }\\
& \multicolumn{3}{c}{$X_T + Q_T\, (S_T- \sigma \, \qconst \, U_T) - \alpha \, Q_T^2 - \phi \, \int_0^T Q_u^2 \, \d u $}\\
\noalign{\vskip-1mm}
\noalign{\vskip 2mm}
\hline
\hline
\noalign{\vskip-1mm}
\noalign{\vskip 2mm}
parameters & $\delta^{*,\mathrm{FI}}(t,Q_{t^-}, U_t)$ & $\delta^{*\mathrm{CJP}}(t,Q_{t^-})$& $\delta^{*,\mathrm{PI}}(t,{Q}_{t^-}, \hat{U}_t)$\\
\noalign{\vskip-1mm}
\noalign{\vskip 2mm}
\hline
\hline
\noalign{\vskip-1mm}
\noalign{\vskip 2mm}
$\varphi= 15$, $\eta=10$, $\gamma=1$, $\qconst=0.6$ &21.33 (4.93)&21.16 (4.93)&21.19 (4.93)\\
\hline
$ \qconst = 0$   &21.34 (5.10) &21.34 (5.10) &21.34 (5.10)\\
$ \qconst = 0.2$ &21.34 (5.08) &21.32 (5.08) &21.32 (5.08)\\
$ \qconst = 0.4$ &21.34 (5.03) &21.26 (5.03) &21.27 (5.03)\\
$ \qconst = 0.6$ &21.33 (4.93) &21.16 (4.93) &21.19 (4.93)\\
$ \qconst = 0.8$ &21.33 (4.79) &21.03 (4.79) &21.12 (4.79)\\
$ \qconst = 1.0$ &21.33 (4.60) &20.86 (4.61) &21.33 (4.60)\\
\hline
$ \gamma = 0$   &21.49 (4.96) &21.34 (4.96) &21.36 (4.95) \\
$ \gamma = 1$   &21.33 (4.93) &21.16 (4.93) &21.19 (4.93) \\
$ \gamma = 2$   &21.16 (4.91) &20.97 (4.91) &20.99 (4.92) \\
$ \gamma = 3$   &20.98 (4.90) &20.76 (4.90) &20.79 (4.91) \\
\hline
$\eta = 2.5$   &21.29 (4.94)&20.78 (4.97)&20.87 (4.99) \\
$\eta = 5.0$   &21.32 (4.93)&21.02 (4.94)&21.06 (4.95)\\
$\eta = 7.5$   &21.33 (4.93)&21.11 (4.93)&21.14 (4.94)\\
$\eta = 10.0$   &21.33 (4.93)&21.16 (4.93)&21.19 (4.93)\\
$\eta = 12.5$   &21.34 (4.93)&21.20 (4.93)&21.22 (4.93)\\
\hline
$\varphi,\psi : 0\%$ informed &21.49 (4.96)&21.34 (4.96)&21.36 (4.95)\\
$\varphi,\psi : 25\%$ informed &21.41 (4.94)&21.26 (4.94)&21.28 (4.94)\\
$\varphi,\psi : 50\%$ informed  &21.33 (4.93)&21.16 (4.93)&21.19 (4.93)\\
$\varphi,\psi : 75\%$ informed &21.25 (4.91)&21.07 (4.92)&21.09 (4.93)\\
$\varphi,\psi : 100\%$ informed &21.16 (4.91)&20.97 (4.91)&20.99 (4.92)\\
  \hline
  \hline
\end{tabular}
  \caption{Average profit and loss (with standard deviation) when the final inventory is evaluated at the fundamental value of the asset for a market maker who follows either (i) the full information optimal strategy $\delta^{*,\mathrm{FI}}(t,Q_{t^-}, U_t)$ observing the process $U_t$,  (ii) the optimal strategy of Cartea-Jaimungal-Penalva (CJP), or (iii)  the imperfect information optimal strategy  $\delta^{*,\mathrm{PI}}(t,{Q}_{t^-}, \hat{U}_t)$ using the filtered process $\hat{U}_t$. The baseline set of parameters is provided in the first row then the parameter changes are reported.
  On each line (varying $\qconst$, $\gamma$, or $\eta$) $\psi$ is rescaled as in \eqref{eq_psi}. The last rows deal with variations of $\varphi,\psi$ keeping the expected number of market arrivals at 30.}
  \label{table2}
\end{table}

\subsection{Further robustness checks}
\label{NORENOR}

As a robustness check, we repeat the analysis developed in Section 
\ref{sec : Simulations}
without rescaling $\psi$; thus, the expected number of market arrivals is not fixed at $30$ as the parameters change. As the number of trade arrivals changes, the figures are different from those in Section \ref{sec : Simulations} but the order (according to performance) among the three strategies is confirmed as well as the main results. 

\textbf{Effect of $\qconst$ :}
We observe an increase of the performance for the FI strategy, while the CJP-strategy performances decreases. The behaviour of the three strategies observed in Section \ref{sec : Simulations} when $\qconst \rightarrow 0$ and $\qconst \rightarrow 1$ is confirmed. The performance increase is due to the higher number of market arrivals. A higher number of market arrivals leads to a higher expected profitability because the market maker benefits from the round trip trades made by liquidity takers.\footnote{The expected number of market arrivals (before sieving according to displacements) goes from $29.86$ (when $\qconst=0$) to $30.22$ (when $\qconst=1$).  }

\begin{figure}[H]
\centering
\includegraphics[width=0.5\textwidth]{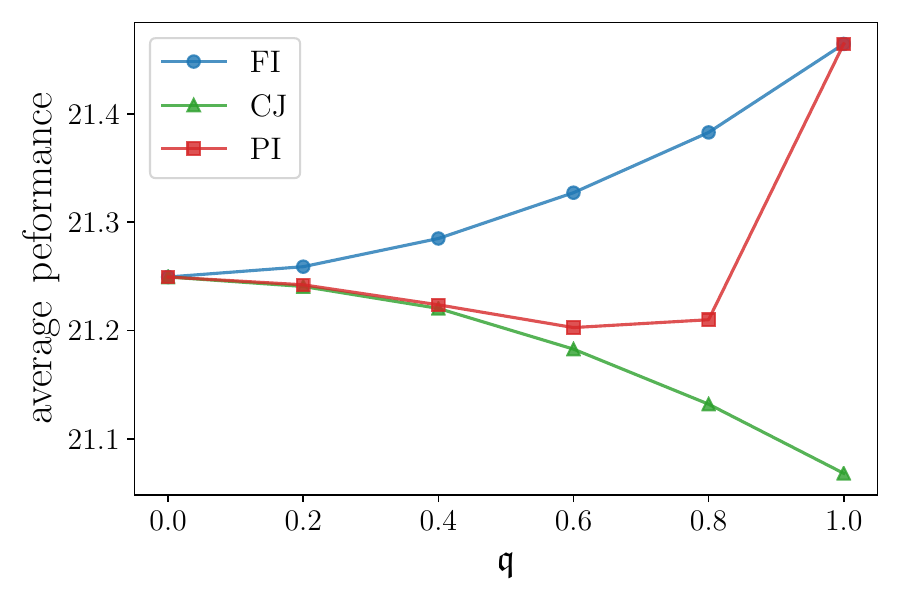}
\caption{Performance of the strategies as a function of $\qconst$.}\label{fig:PnL_qconst_2} 
\end{figure}

\textbf{Effect of $\gamma$ :}
We observe a decrease of the performances up until $\gamma = 1$, followed by an increase in the performances as gamma increases ($\gamma>1$). The decrease for $\gamma \leq 1$ is due to the higher toxicity of the trades sent by the informed traders (in line with the previous results) while the performance increase for $\gamma \geq 1$ is due to the higher number of market arrivals that lead to an increase in profitability that overshadows the increase of toxicity in the order flow of informed traders.\footnote{The expected number of market arrivals (before sieving according to displacements) goes from $29.86$ (when $\gamma=0$) to $31.06$ (when $\gamma = 3$). }
\begin{figure}[H]
\centering
\includegraphics[width=0.5\textwidth]{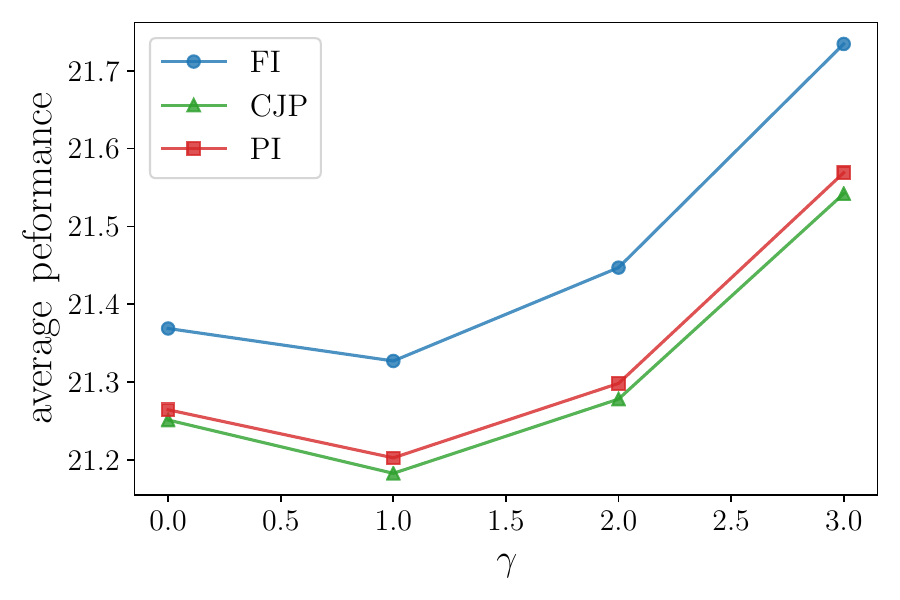}
\caption{Performance of the strategies as a function of $\gamma$.}\label{fig:PnL_gamma_2} 
\end{figure} 

\textbf{Effect of $\eta$ :}
For the full information case we observe a decrease of the performance as $\eta$ increases. This is due to the lower arrivals of orders. For the other two strategies, we refer to the arguments made for the closing of the gap in Figure \ref{fig:PnL_eta}.
\begin{figure}[H]
\centering
\includegraphics[width=0.5\textwidth]{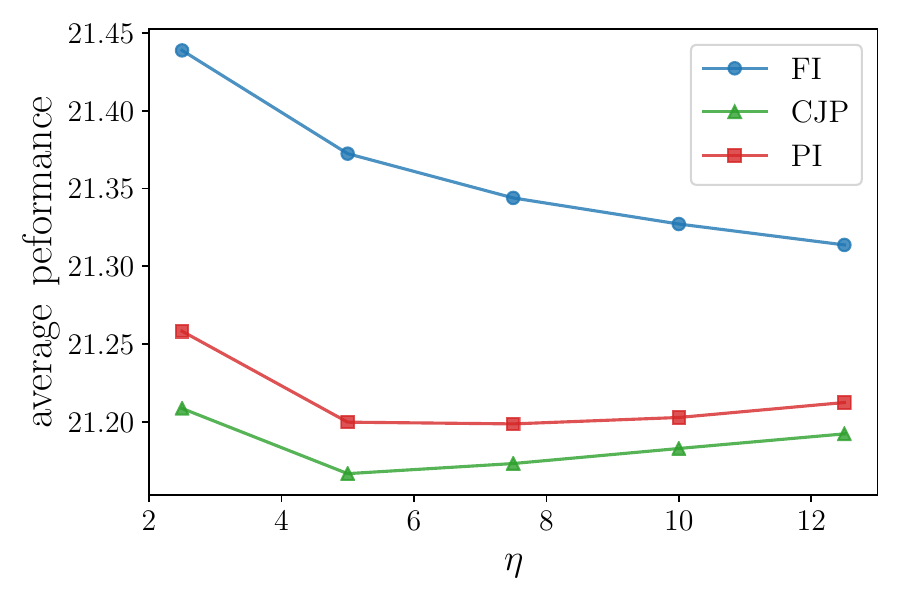}
\caption{Performance of the strategies as a function of $\eta$.}\label{fig:PnL_eta_2} 
\end{figure}

\subsection{How good is the closed-form second order approximation?}\label{sec:how-good-second-order-approx}

Figure \ref{fig:optimal quotes ODE PDE} shows the bid and ask optimal displacements obtained using a finite-difference scheme for the  solution of the PDE \eqref{eq:HJB_EQ_FI} and the second order approximation closed-form solution in \eqref{eq:systemODE_FI}. We observe that the differences are small justifying the second-order approximation of the solution.

\begin{figure}[H]
\centering
\includegraphics[width=0.6\textwidth]{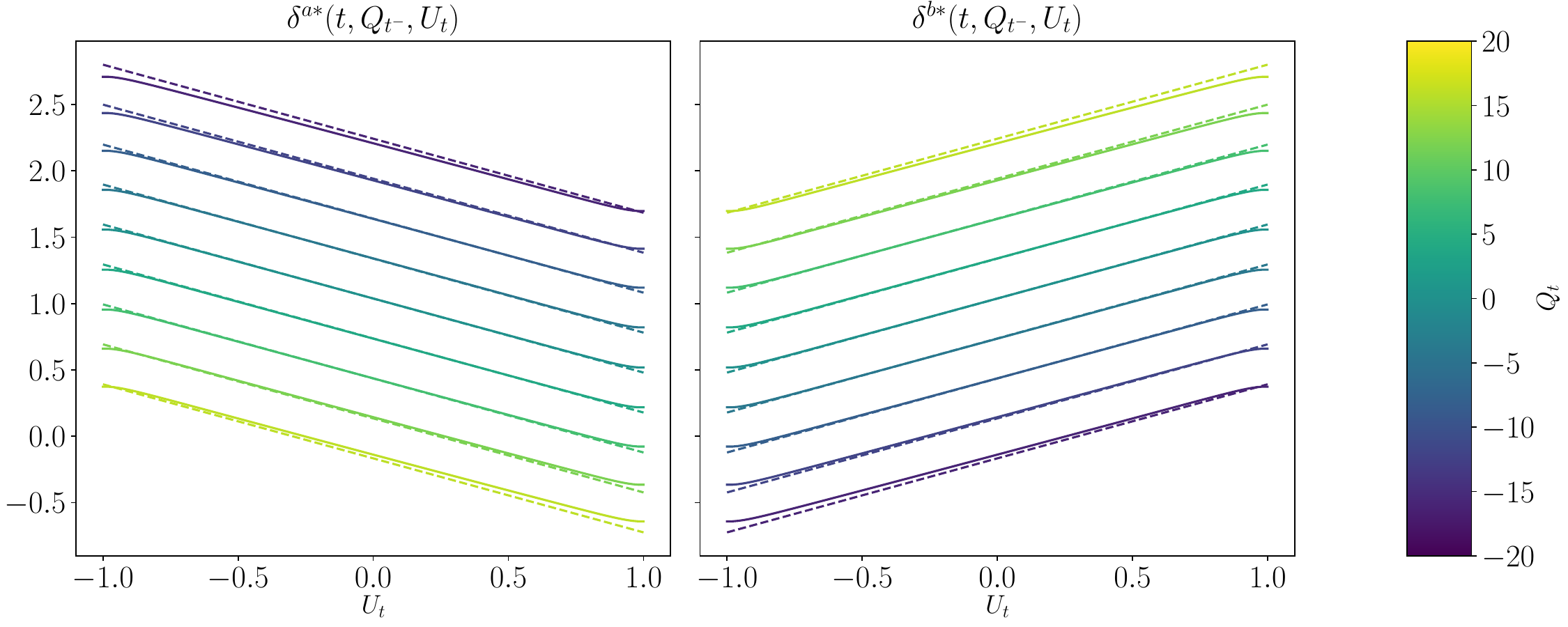}
\caption{Ask and bid displacements as a function of the fad ($x$-axis) and the inventory of the market maker (colour bar) under full information.
Solid lines: numerical simulation of the PDE; dashed lines: numerical simulation of the linear-quadratic approximation. Left panel is for the optimal ask displacement, and right panel is for the optimal bid displacement.}\label{fig:optimal quotes ODE PDE} 
\end{figure} 
\subsection{Maximum likelihood estimators}\label{sec:likelihood_estim}
To calibrate  $(\eta,\qconst)$, we compute a Kalman--Bucy filter for $\process[U]$ for each choice of $(\eta,\qconst)$, and then consider maximising the log-likelihood of the resulting Kalman--Bucy filter in discrete time.\footnote{See Chapter 6.6 in \cite{shumway2006time}.} More precisely, we assume that we observe the mid price $\process[S]$ at discrete times $0=t_0<t_1<...<t_N=T$, with $\Delta = t_{i+1}-t_i$, and we wish to estimate $(\qconst,\eta)$. First, we rewrite the dynamics for $U$ and $S$ as
\begin{equation}
\label{eq:discrete system}
    \begin{cases}
          &U_{k+1} = \phi\, U_k + \varepsilon_k\,,\qquad \phi = e^{-\eta\,\Delta}\,,\qquad \varepsilon_k \sim \mathcal{N}(0,Q)\,,\qquad Q = \frac{1-\phi^2}{2\,\eta}\\&
          y_k = H\,U_k + \xi_k\,,\qquad H = \qconst\,(\phi-1)\,,\qquad y_k = \frac{ S_{k+1}-S_k-\mu\,\Delta}{\sigma}\,,\qquad \xi_k = \qconst\,\varepsilon_k + \pconst\,(Z_{k+1} - Z_k)\,.
    \end{cases}
\end{equation}
Using \eqref{eq:discrete system} we write the associated discrete time Kalman--Bucy filter
\begin{equation}
\label{eq:density_filter_2}
    y_k\,|\,y_0,..,y_{k-1} \sim \mathcal{N}(\hat{y}_k,F_k)\,,
\end{equation}
where $$\hat{y}_k = H\,\widehat{U}_{k|k-1}\,,\qquad F_k = H^2\,P_{k|k-1} + \qconst^2\,Q + \pconst^2\,\Delta\,,$$
and where 
\begin{equation}
    \begin{cases}
        &\widehat{U}_{k+1|k} = \phi\,\hat{U}_{k|k-1}+ K_k\,\nu_k\,,\qquad
        K_k = \frac{\phi\, P_{k|k-1}\,H + \qconst\,Q}{F_k}\,,\qquad \nu_k = y_k -\hat{y}_k\,,\\&
        P_{k+1|k} = \phi^2\, P_{k|k-1} + Q -K_k^2\,F_k\,,\qquad \widehat{U}_{1|0} = 0\,, \qquad P_{1|0} = Q\,.
    \end{cases}
\end{equation}
Using  that
$$p(y_0,...,y_N) = \prod_{k=0}^N p(y_k\,|\,y_0,..,y_{k-1})\,,$$
and equation \eqref{eq:density_filter_2} we derive the log--likelihood 
\begin{equation}
\label{eq:log_likelihood}
    l(\qconst,\eta) = -\frac{1}{2}\,\sum_{k=0}^N \Bigg[\log(F_k) + \frac{\nu_k^2}{F_k} \Bigg]\,,
\end{equation}
and we optimise this quantity over $(\qconst,\eta)$. \\
\newline
 For $(\psi,\varphi,\gamma)$, we start by disentangling the market orders from the filled orders (as in Section \ref{sec:Beyond KBf}), and consider the market arrivals $M^{\pm}_t$ under a probability measure where we use $\hat{U}$ instead of $U$. More precisely, for the calibration exercise we assume that $M$ has stochastic intensities given by 
    $${\vartheta}^{\varphi,\gamma,\pm}_t = \varphi + \psi\, e^{\pm \gamma\,\hat{\qconst}\,\sigma\,\hat{U}_t}\,,$$
    where $\process[\hat{U}]$ is the Kalman--Bucy filter of the fad using the calibrated $(\hat{\qconst}, \hat{\eta})$. 
    Finally, using the independence of the increments of the Poisson processes conditionally on the observation of $\process[\hat{U}]$, we run a maximization of the log--likelihood of the discrete time  processes $({M_t^\pm})_{t \in \mfT}$ to estimate $(\gamma,\varphi)$, where the discrete time log--likelihood is given by 
    $$l(\varphi,\gamma)= \Bigg[\sum_{i= 1}^N\Delta{M^{+}_i}\,\log(\vartheta_i^{\varphi,\gamma,+}) - \sum_{i=1}^N\Delta\,\vartheta^{\varphi,\gamma,+}_i + \sum_{i= 1}^N\Delta{M^{-}_i}\,\log(\vartheta_i^{\varphi,\gamma,-}) - \sum_{i=1}^N\Delta\,\vartheta^{\varphi,\gamma,-}_i\Bigg]\,.$$
    Tables \ref{table_KB_calib} and \ref{table_poisson_calib} (in the main text) show the performance of the above calibration procedures.

\end{document}